\documentclass[journal,onecolumn]{IEEEtran}
\ifCLASSINFOpdf
\else
\fi

\usepackage{amsmath}

\usepackage{amsmath,bm}
\usepackage{amsfonts}
\usepackage{amssymb}
\usepackage{mathtools}
\usepackage{enumerate}
\usepackage{dsfont} 
\usepackage{graphicx}
	\graphicspath{{figures/}}
\usepackage{subfigure}
\usepackage{tikz}
	\usetikzlibrary{quotes,angles} 
	\usetikzlibrary{calc}
	\usetikzlibrary{shapes,arrows}
	\tikzstyle{frame} = [draw, -latex]
	\tikzstyle{lineUD} = [draw]
	\tikzstyle{line} = [draw, -latex']
	\tikzstyle{line2} = [draw, -latex', dashdotted]
	\tikzstyle{line3} = [draw, -latex', dashed]
	\tikzstyle{line3UD} = [draw, dashed]
	\tikzstyle{place} = [circle, draw=black, fill=white, thick, inner sep=2pt, minimum size=1mm]
	\tikzstyle{placeRed} = [circle, draw=red, fill=red, thick, inner sep=2pt, minimum size=1mm]
	\tikzstyle{vertex} = [circle, draw=black, fill=black, thick, inner sep=2pt, minimum size=1mm]
\usepackage{tkz-euclide}
\usepackage{amsthm}
\newtheorem{definition}{Definition}
\newtheorem{theorem}{Theorem}
\newtheorem{lemma}{Lemma}
\newtheorem{remark}{Remark}
\newtheorem{proposition}{Proposition}

\newtheorem{assumption}{Assumption}

\usepackage{color}
\usepackage{caption}
\usepackage{cite}
\usepackage{url}

\usepackage[colorlinks=true, allcolors=blue]{hyperref}

\newcommand{\mbf}{\mathbf}

\newcommand{\myemph}{\emph}

\DeclareMathOperator{\rank}{{rank}}
\DeclareMathOperator{\vspan}{{span}}

\DeclareMathOperator{\vnull}{{Null}}
\DeclarePairedDelimiter \card{\lvert}{\rvert}
\DeclarePairedDelimiter \norm{\|}{\|}

\providecommand\given{}
\newcommand\SetSymbol[1][]{\nonscript\;#1\vert\nonscript\;
\mathopen{}\allowbreak}
\DeclarePairedDelimiterX\set[1]\{\}{%
\renewcommand\given{\SetSymbol[\delimsize]}
#1
}

\begin{document}
%
\title{
Generalized weak rigidity: Theory, and local and global convergence of formations
}

%
%

\author{
Seong-Ho Kwon and Hyo-Sung Ahn
\thanks{S.-H. Kwon and H.-S. Ahn are with Distributed Control \& Autonomous Systems Lab. (DCASL), School of Mechanical Engineering, Gwangju Institute of Science and Technology (GIST), Gwangju 61005, Korea. E-mails: seongho@gist.ac.kr; hyosung@gist.ac.kr}
}

\maketitle

\begin{abstract}
This paper proposes a generalized weak rigidity theory, and aims to apply the theory to formation control problems with a gradient flow law. The generalized weak rigidity theory is utilized in order to characterize desired rigid formations by a general set of pure inter-agent distances and angles. As the first result of its applications,  this paper provides analysis of locally exponential stability for a formation control system with pure distance/angle constraints in $2$- and $3$-dimensional spaces. Then, as the second result, it is shown that if there are three agents in $2$-dimensional space then almost globally exponential stability is ensured for a formation control system with pure distance/angle constraints.
Through numerical simulations, the validity of analyses is illustrated.
\end{abstract}

\section{Introduction} \label{Sec:introduce}
Based on rigidity theories, distributed formation control has been investigated under the networked multi-agent systems \cite{oh2015survey, anderson2008rigid}. In formation control problems, the rigidity theories have been key concepts  to characterize a rigid formation shape for a framework with specific constraints, such as distances, bearings, subtended angles, etc., where the theories can briefly be classified according to types of constraints; for example, distance-based rigidity theory, bearing-based rigidity theory, angle-based rigidity theory and mixed rigidity theory. 

In particular, based on use of the distance-based rigidity (distance rigidity) theory \cite{asimow1978rigidity,asimow1979rigidity,roth1981rigid,C:Hendrickson:SIAM1992}, formation control problems have been extensively studied \cite{krick2009stabilisation,cortes2009global,sun2014finite,sun2016exponential,dorfler2010geometric,dimarogonas2008stability,cai2014rigidity}, where a rigid formation is characterized by constraints of inter-agent distances. In formation control with the distance rigidity theory, each agent is required to sense relative positions of its neighbors.
In terms of the bearing-based rigidity (bearing rigidity) theory \cite{whiteley1996some,franchi2012decentralized,zhao2016bearing}, 
inter-agent bearings are used to achieve a unique formation shape with which
formation control problems have been also studied \cite{zhao2016bearing,zhao2017translational}.
This approach makes use of measurement of relative bearings or positions of its neighbors in formation control.
In recent years, formation control problems with the angle-based and mixed rigidity theories have attracted much research interest \cite{anderson2008rigid,bishop2015distributed, park2017rigidity,jing2018weak,kwon2018infinitesimal,kwon2018infinitesimal_inp.,cao2019bearing}.

This paper particularly focuses on formation control based on the mixed rigidity theory with distances and angles, where the mixed rigidity theory with distance and angle information is called weak rigidity theory \cite{park2017rigidity,jing2018weak,kwon2018infinitesimal_inp.,kwon2018infinitesimal}.
In fact, the weak rigidity theory has been interpreted in a different way according to publications, that is,
the weak rigidity theories studied in publications \cite{park2017rigidity,jing2018weak,kwon2018infinitesimal_inp.,kwon2018infinitesimal} are conceptually similar in the sense that angle information is used; but there are differences among the theories. 
In this paper, to distinguish the existing works, we call the theories of  Park et al. (2017)\cite{park2017rigidity}, Jing et al. (2018)\cite{jing2018weak} and Kwon et al. (2018)\cite{kwon2018infinitesimal_inp.} \textit{basic weak rigidity theory}, \textit{type-1 weak rigidity theory} and  \textit{type-2 weak rigidity theory}, respectively.
%

In the work based on the basic weak rigidity theory\cite{park2017rigidity}, the authors introduce the weak rigidity theory for the first time, where the theory is studied with some special cases in the $2$-dimensional space. In accordance with the definition of the basic weak rigidity theory, a rigid formation has to be composed of triangular formations, and each triangular formation should have two adjacency distance constraints to define a subtended angle constraint.
For example, as shown in Fig.~\ref{Fig:diff_figs_a}, two distance constraints for a subtended angle constraint should be defined for the triangular formation. 
Based on the type-1 weak rigidity theory\cite{jing2018weak},
inner products of inter-agent relative positions are used as angle constraints to characterize rigid formations, where the inner products are distinct from the cosines of the angles among agents, i.e., they are different from inner products of inter-agent relative bearings. In this approach, it is remarkable that an inner product includes distance and angle information simultaneously and
thus it could include redundant information when characterizing rigid formations; for example, considering two inner products $z_{21}^\top z_{31}$ and $z_{13}^\top z_{23}$, where $z_{ij}$ denotes a relative position from agent $j$ to agent $i$, we can observe that the Euclidean norm of $z_{13}$ is redundantly used.
In recent years, the type-2 weak rigidity theory\cite{kwon2018infinitesimal,kwon2018infinitesimal_inp.} has been introduced, where the concept of the type-2 weak rigidity theory is extended from the basic weak rigidity theory but distinguished from the type-1 weak rigidity theory by types of constraints.
Compared with the type-1 weak rigidity theory, the type-2 weak rigidity theory involves pure distance/angle constraints without any redundant information; for example, see Fig.~\ref{Fig:diff_figs_b}. Moreover, one can achieve a rigid formation with only pure angle constraints that does not need any distance constraints as shown in Fig.~\ref{Fig:diff_figs_c} whereas one cannot with the type-1 weak rigidity theory.
The comparison between the type-1 and type-2 weak rigidity theories is again highlighted in Remark~\ref{Remark:comparison} in Section~\ref{Sec:Weak_rigidity}.

\begin{figure}[]
\centering
\subfigure[Triangular formation characterized by two distance constraints and one angle constraint subtend by the two distance constraints.]{ \label{Fig:diff_figs_a}
\qquad \begin{tikzpicture}[scale=.7]
\node[place] (node1) at (0,2) [label=above:$1$] {};
\node[place] (node2) at (-2,0) [label=left:$2$] {};
\node[place] (node3) at (2,0) [label=right:$3$] {};

\draw[lineUD] (node1)  -- node [above left] {$d_{12}$} (node2);
\draw[lineUD] (node1)  -- node [above right] {$d_{13}$} (node3);
\pic [draw, -, "${\theta}^1_{23}$", angle eccentricity=1.7, angle radius=0.4cm] {angle = node2--node1--node3};
\end{tikzpicture}\qquad%
} \quad
\subfigure[Triangular formation characterized by one distance constraint and two angle constraints.]{ \label{Fig:diff_figs_b}
\qquad\begin{tikzpicture}[scale=0.7]
\node[place] (node1) at (0,2) [label=above:$1$] {};
\node[place] (node2) at (-2,0) [label=left:$2$] {};
\node[place] (node3) at (2,0) [label=right:$3$] {};

\draw[lineUD] (node1)  -- node [above left] {$d_{12}$} (node2);
\draw[dashed] (node2)  -- (node3);
\draw[dashed] (node1)  -- (node3);

\pic [draw, -, "${\theta}^1_{23}$", angle eccentricity=1.7, angle radius=0.4cm] {angle = node2--node1--node3};
\pic [draw, -, "${\theta}^3_{12}$", angle eccentricity=1.7, angle radius=0.4cm] {angle = node1--node3--node2};
\end{tikzpicture}\qquad%
} \quad
\subfigure[Triangular formation characterized by two angle constraints.]{ \label{Fig:diff_figs_c}
\qquad\begin{tikzpicture}[scale=0.7]
\node[place] (node1) at (0,2) [label=above:$1$] {};
\node[place] (node2) at (-2,0) [label=left:$2$] {};
\node[place] (node3) at (2,0) [label=right:$3$] {};

\draw[dashed] (node1)  -- node {} (node2);
\draw[dashed] (node2)  -- node {} (node3);
\draw[dashed] (node1)  -- node {} (node3);

\pic [draw, -, "${\theta}^1_{23}$", angle eccentricity=1.7, angle radius=0.4cm] {angle = node2--node1--node3};
\pic [draw, -, "${\theta}^3_{12}$", angle eccentricity=1.7, angle radius=0.4cm] {angle = node1--node3--node2};
\end{tikzpicture}\qquad%
}
\caption{Triangular formations with different constraints. The symbol $d_{ij}$ denotes a distance constraint between vertices $i$ and $j$, and the symbol $\theta_{ij}^{k}$ denotes an angle constraint subtended by edges $(i,k)$ and $(j,k)$. The dashed lines indicate virtual edges which are not distance constraints. 
} 
\end{figure}
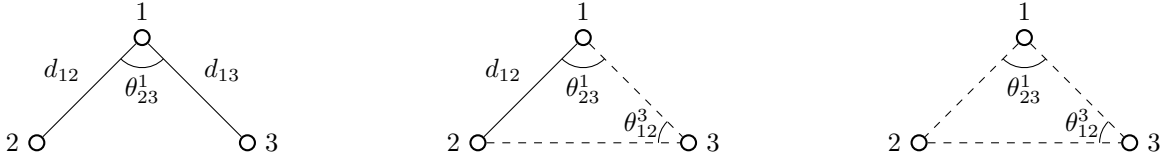

Based on the type-1 weak rigidity theory, the studies on multi-agent formation control in the $d$-dimensional space are almost completed in Jing et al. (2018)\cite{jing2018weak}. 
On the other hand, there are still many tasks that need to be studied in the case of the type-2 weak rigidity theory in $d$-dimensional space. 
In this sense, this paper aims to explore the type-2 weak rigidity theory and to apply the theory to formation control. In this paper, to differentiate between the weak rigidity theories, the extended concept from the type-2 weak rigidity theory is named \textit{generalized weak rigidity}.
Consequently, the main contributions of this paper are summarized as follows. 
First, we introduce the concepts of \textit{generalized weak rigidity} and \textit{generalized infinitesimal weak rigidity} in the $2$- and $3$-dimensional spaces. 
These concepts are used to examine whether or not a given formation with pure distance/angle constraints is globally or locally  rigid.
We then show that both concepts are generic properties.
Moreover,  it is shown that the generalized weak rigidity theory is a weaker condition than the classic distance rigidity theory.
Second, we apply the generalized weak rigidity theory to formation control with a gradient flow law. Based on the generalized weak rigidity theory, we provide analysis of locally exponential stability on a $n$-agent formation control system in the $2$- and $3$-dimensional spaces, and further analysis of almost globally exponential stability on a $3$-agent formation control system in the $2$-dimensional space.

The rest of this paper is organized as follows. Preliminaries, notations and motivation are briefly given in Section \ref{Sec:Pre_Notation}. Section \ref{Sec:Weak_rigidity} presents the generalized weak rigidity theory. Based on the rigidity theory, Sections \ref{Sec:Formation_control_local_stab.} and \ref{Sec:Formation_control_almost} discuss analysis of local convergence and almost global convergence of formations, respectively. Section \ref{Sec:Simul} presents numerical simulations to support our analysis. Finally, Section \ref{Sec:Conclusion} provides conclusion and summary.
\section{Preliminary} \label{Sec:Pre_Notation}
Let $\norm{\cdot}$ and $\card{\mathcal{S}}$ denote the Euclidean norm of a vector and cardinality of a set $\mathcal{S}$, respectively. The symbols $\vnull(\cdot)$ and $\rank (\cdot)$ denote the null space and rank of a matrix, respectively. The symbol $I_N \in \mathbb{R}^{N \times N}$ denotes the identity matrix, and the symbol $\mathds{1}_n\in\mathbb{R}^n$ denotes a vector whose all entries are $1$ as $\mathds{1}_n = [1, ..., 1]^\top$.
We define an undirected graph $\mathcal{G}$ as $\mathcal{G} = (\mathcal{V},\mathcal{E})$, where $\mathcal{V}=\set{1,2,...,n}$ denotes a vertex set and $\mathcal{E} \subseteq \mathcal{V} \times \mathcal{V}$ denotes an edge set with $m=\card{\mathcal{E}}$.
Since an undirected graph is considered, it is assumed that $(i,j) = (j,i)$ for all $i,j \in \mathcal{V}$. 
An angle set $\mathcal{A} \subseteq \mathcal{V} \times \mathcal{V}\times \mathcal{V}$ is defined as $\mathcal{A} = \set{(k,i,j) \given \theta_{ij}^{k} \text{ is assigned to } i,j,k \in \mathcal{V}, \theta_{ij}^{k} \in [0,\pi]}$ with $w=\card{\mathcal{A}}$, where $\theta_{ij}^{k}$  denotes an angle subtended by the adjacent edges $(i,k)$ and $(j,k)$, where the adjacent edges $(i,k)$ and $(j,k)$ do not necessarily belong to $\mathcal{G}$. Angles used in this paper have no directions and signs.
For a position vector $p_i \in \mathbb{R}^{d}$, we define a configuration $p$ of $\mathcal{G}$ as $p= [p_{1}^\top,...,p_{n}^\top]^\top \in \mathbb{R}^{dn}$ and define a framework as $(\mathcal{G},\mathcal{A},p)$ in $\mathbb{R}^{d}$.
We define a relative position vector as $z_{ij} = p_{i} - p_{j}$ for a framework $(\mathcal{G},\mathcal{A},p)$, $(i,j)\in \mathcal{E}$ and $i \neq j$. We set the order of the associated relative position vectors $z_{ij}$ as $z_{g_{ij}} = z_{ij}, g\in \{ 1,..., m\}$. 
Similarly, for $(k,i,j) \in \mathcal{A}$ and $h\in \{ 1,..., w\}$, a cosine $A_{h_{kij}}$ is defined as 
$A_{h_{kij}} = \cos{\theta_{ij}^{k}}$. It is remarkable that $A_{h_{kij}}$ is equivalently represented as $A_{h_{kij}}=\cos{\theta_{ij}^{k}} = \frac{z_{ki}^\top z_{kj}}{\norm{z_{ki}}\norm{z_{kj}}} = \frac{\norm{z_{ki}}^{2} + \norm{z_{kj}}^{2} - \norm{z_{ij}}^{2}}{2\norm{z_{ki}}\norm{z_{kj}}}$. 
We occasionally make use of $z_g$ and $A_h$ for notational convenience instead of $z_{g_{ij}}$ and $A_{h_{kij}}$, respectively, if no confusion is expected.
Note that, in this paper, we focus on problems only in $2$- and $3$-dimensional spaces, i.e., $d=2,3$.

\begin{remark}
The advantages of formation control studied in this paper are mainly fourfold.
First,  the proposed formation control with pure distance/angle constraints is convenient to control scalings of formations; for example, when we want to control a scaling of the formation illustrated in Fig.~\ref{Fig:weak_superiority_02}, we only need to control the distance constraint between agents $1$ and $2$ while all distance constraints of the formation illustrated in Fig.~\ref{Fig:weak_superiority_01} have to be controlled.
This is due to the fact that pure angle constraints are invariant to trivial motions corresponding to translations, rotations and scalings of an entire formation while distance constraints are invariant to only a subset of the motions, i.e., translations and rotations. 
Second, the proposed control system is a distributed multi-agent system, that is, each agent only needs to measure relative positions of its neighbor agents with respect to  its local coordinate system.
Third, each agent does not require any wireless communication if scaling control of formations is not considered.
Fourth, if wireless communications are inevitable when we control scalings of formations, then we can reduce the communication load; for example, when we want to control a scaling of the formation in Fig.~\ref{Fig:weak_superiority_02}, it is only necessary to command agents $1$ and $2$ to change desired distance constraints while agents $3$ and $4$ do not have to be commanded. 
\end{remark}

\begin{figure}[]
\centering
\subfigure[Locally unique formation with pure distance constraints]{\label{Fig:weak_superiority_01}
\qquad\qquad\qquad\quad\begin{tikzpicture}[scale=.9]
\node[place] (node1) at (0,2) [label=left:$1$] {};
\node[place] (node2) at (-1,0) [label=left:$2$] {};
\node[place] (node3) at (1,0) [label=right:$3$] {};
\node[place] (node4) at (2,2) [label=right:$4$] {};

\draw[lineUD] (node1)  -- node [above left] {} (node2);
\draw[lineUD] (node1)  -- (node3);
\draw[lineUD] (node2)  -- (node3);
\draw[lineUD] (node1)  -- (node4);
\draw[lineUD] (node3)  -- (node4);
\end{tikzpicture} \qquad\qquad\qquad\quad
}\,
\subfigure[Locally unique formation with pure distance and angle constraints]{\label{Fig:weak_superiority_02}
\qquad\qquad\qquad\quad\begin{tikzpicture}[scale=.9]
\node[place] (node1) at (0,2) [label=left:$1$] {};
\node[place] (node2) at (-1,0) [label=left:$2$] {};
\node[place] (node3) at (1,0) [label=right:$3$] {};
\node[place] (node4) at (2,2) [label=right:$4$] {};

\draw[lineUD] (node1)  -- node [above left] {} (node2);
\draw[dashed] (node1)  -- (node3);
\draw[dashed] (node2)  -- (node3);
\draw[dashed] (node1)  -- (node4);
\draw[dashed] (node3)  -- (node4);

\pic [draw, -, "${\theta}^2_{31}$", angle eccentricity=1.7, angle radius=0.4cm] {angle = node3--node2--node1};
\pic [draw, -, "${\theta}^3_{12}$", angle eccentricity=1.7, angle radius=0.4cm] {angle = node1--node3--node2};
\pic [draw, -, "${\theta}^1_{34}$", angle eccentricity=1.7, angle radius=0.4cm] {angle = node3--node1--node4};
\pic [draw, -, "${\theta}^4_{13}$", angle eccentricity=1.7, angle radius=0.4cm] {angle = node1--node4--node3};
\end{tikzpicture} \qquad\qquad\qquad\quad
} 
\caption{Examples of locally rigid formations in $\mathbb{R}^{2}$, where the solid lines denote distance constraints, and the dashed lines denote virtual edges which are not distance constraints. Angle constraints are denoted by $\theta_{ij}^{k}, (k,i,j) \in \mathcal{A}$.} \label{Fig:weak_superiority}
\end{figure}
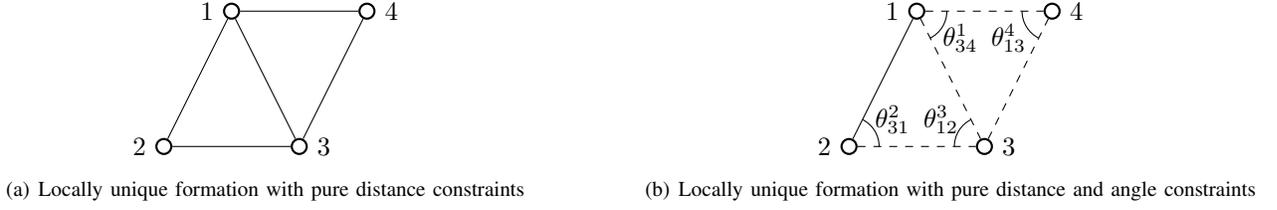
\section{Generalized weak rigidity} \label{Sec:Weak_rigidity}
In this section, we introduce a generalized weak rigidity theory in $\mathbb{R}^{d}$. The basic concept on the theory is related to  how to examine whether or not a formation shape can be determined up to a translation and a rotation (and additionally, for specific cases, a scaling factor) by given relative distance and angle constraints.
\subsection{Generalized weak rigidity (GWR)}
In order to define the concept of the generalized weak rigidity, we make use of the following definition used in the distance rigidity theory. It is well known that two frameworks $(\mathcal{G},\mathcal{A},p)$ and $(\mathcal{G},\mathcal{A},q)$ are said to be \myemph{congruent} if $\norm{p_{i}-p_{j}}=\norm{q_{i}-q_{j}}$ for all $i,j \in \mathcal{V}$. We now define the fundamental concepts on the generalized weak rigidity.

\begin{definition}[Strong equivalency]
\label{Def:strongEquiv}
Two frameworks $(\mathcal{G},\mathcal{A},p)$ and $(\mathcal{G},\mathcal{A},q)$ are said to be \myemph{strongly equivalent} if the following two conditions hold
\begin{itemize}
\item $\norm{p_{i}-p_{j}} = \norm{q_{i}-q_{j}}, \forall (i,j) \in \mathcal{E}$,
\item $\cos\left({\theta_{ij}^{k}}\right)_{\in(\mathcal{G},\mathcal{A},p)} = \cos\left({\theta_{ij}^{k}}\right)_{\in(\mathcal{G},\mathcal{A},q)}, \forall (k,i,j) \in \mathcal{A}$,
\end{itemize}
where $\left({\theta_{ij}^{k}}\right)_{\in(\mathcal{G},\mathcal{A},p)}$ and $\left({\theta_{ij}^{k}}\right)_{\in(\mathcal{G},\mathcal{A},q)}$ denote the angles belonging to $(\mathcal{G},\mathcal{A},p)$ and $(\mathcal{G},\mathcal{A},q)$, respectively.
\end{definition}

\begin{definition}[Angle equivalency]
\label{Def:angleEquiv}
Two frameworks $(\mathcal{G},\mathcal{A},p)$ and $(\mathcal{G},\mathcal{A},q)$ with $\mathcal{E} = \emptyset$ are said to be \myemph{angle equivalent} if $\cos\left({\theta_{ij}^{k}}\right)_{\in(\mathcal{G},\mathcal{A},p)} = \cos\left({\theta_{ij}^{k}}\right)_{\in(\mathcal{G},\mathcal{A},q)}, \forall (k,i,j) \in \mathcal{A}$.
\end{definition}
In this paper, $\mathcal{E} \neq \emptyset$ means that there exists
at least one distance constraint, on the other hand, $\mathcal{E} = \emptyset$ means that any distance constraint does not exist. 


\begin{definition}[Proportional congruency]
\label{Def:angleCong}
Two frameworks $(\mathcal{G},\mathcal{A},p)$ and $(\mathcal{G},\mathcal{A},q)$ with $\mathcal{E} = \emptyset$ are said to be \myemph{proportionally congruent} if 
$\norm{p_{i}-p_{j}} = C\norm{q_{i}-q_{j}}, \forall i,j \in \mathcal{V}$,
where $C$ denotes a positive proportional constant.
\end{definition}
Fig.~\ref{Fig:weak_example01} shows three examples for the above definitions, where the solid lines denote distance constraints, and the dashed lines denote virtual edges which are not distance constraints. Moreover, distance and angle constraints are denoted by $d_{ij}, (i,j) \in \mathcal{E}$ and $\theta_{ij}^{k}, (k,i,j) \in \mathcal{A}$, respectively.
\begin{figure}[]
\centering
\subfigure[Strongly equivalent frameworks.]{\label{Fig:weak_example01_a}
\begin{tikzpicture}[scale=0.6]
\node[place] (node1) at (0,2) [label=above:$1$] {};
\node[place] (node2) at (-3,0) [label=left:$2$] {};
\node[place] (node3) at (1.5,0) [label=right:$3$] {};
\node[place] (node4) at (0,-2) [label=below:$4$] {};

\draw[lineUD] (node1)  -- node [above left] {$d_{12}$} (node2);
\draw[dashed] (node1)  -- (node3);
\draw[dashed] (node2)  -- (node4);
\draw[dashed] (node1)  -- (node4);
\draw[lineUD] (node3)  -- node [below right] {$d_{34}$}  (node4);

\pic [draw, -, "${\theta}^2_{41}$", angle eccentricity=1.7, angle radius=0.4cm] {angle = node4--node2--node1};
\pic [draw, -, "${\theta}^1_{43}$", angle eccentricity=1.7, angle radius=0.4cm] {angle = node4--node1--node3};
\pic [draw, -, "${\theta}^4_{31}$", angle eccentricity=1.7, angle radius=0.4cm] {angle = node3--node4--node1};
\end{tikzpicture} \qquad %
\begin{tikzpicture}[scale=0.6]
\node[place] (node1) at (0,2) [label=above:$1$] {};
\node[place] (node2) at (-3,0) [label=left:$2$] {};
\node[place] (node3) at (-1.5,0) [label=right:$3$] {};
\node[place] (node4) at (0,-2) [label=below:$4$] {};

\draw[lineUD] (node1)  -- node [above left] {$d_{12}$} (node2);
\draw[dashed] (node1)  -- (node3);
\draw[dashed] (node2)  -- (node4);
\draw[dashed] (node1)  -- (node4);
\draw[lineUD] (node3)  -- node [left] {$d_{34}$}  (node4);

\pic [draw, -, "${\theta}^2_{41}$", angle eccentricity=1.8, angle radius=0.3cm] {angle = node4--node2--node1};
\pic [draw, -, "${\theta}^1_{34}$", angle eccentricity=1.7, angle radius=0.4cm] {angle = node3--node1--node4};
\pic [draw, -, "${\theta}^4_{13}$", angle eccentricity=1.7, angle radius=0.4cm] {angle = node1--node4--node3};
\end{tikzpicture}%
} \qquad\qquad\qquad
\subfigure[Angle equivalent frameworks.]{\label{Fig:weak_example01_b}
\begin{tikzpicture}[scale=0.6]
\node[place] (node1) at (0,2) [label=above:$1$] {};
\node[place] (node2) at (-3,0) [label=left:$2$] {};
\node[place] (node3) at (1.5,0) [label=right:$3$] {};
\node[place] (node4) at (0,-2) [label=below:$4$] {};

\draw[dashed] (node1)  -- (node2);
\draw[dashed] (node1)  -- (node3);
\draw[dashed] (node2)  -- (node4);
\draw[dashed] (node1)  -- (node4);
\draw[dashed] (node3)  -- (node4);

\pic [draw, -, "${\theta}^2_{41}$", angle eccentricity=1.7, angle radius=0.4cm] {angle = node4--node2--node1};
\pic [draw, -, "${\theta}^1_{43}$", angle eccentricity=1.7, angle radius=0.4cm] {angle = node4--node1--node3};
\pic [draw, -, "${\theta}^3_{14}$", angle eccentricity=1.7, angle radius=0.4cm] {angle = node1--node3--node4};
\pic [draw, -, "${\theta}^4_{12}$", angle eccentricity=1.7, angle radius=0.4cm] {angle = node1--node4--node2};
\end{tikzpicture} \qquad %
\begin{tikzpicture}[scale=0.6]
\node[place] (node1) at (0,2) [label=above:$1$] {};
\node[place] (node2) at (-3,0) [label=left:$2$] {};
\node[place] (node3) at (-1.5,0) [label=above:$3$] {};
\node[place] (node4) at (0,-2) [label=below:$4$] {};

\draw[dashed] (node1)  -- (node2);
\draw[dashed] (node1)  -- (node3);
\draw[dashed] (node2)  -- (node4);
\draw[dashed] (node1)  -- (node4);
\draw[dashed] (node3)  -- (node4);

\pic [draw, -, "${\theta}^2_{41}$", angle eccentricity=1.8, angle radius=0.3cm] {angle = node4--node2--node1};
\pic [draw, -, "${\theta}^1_{34}$", angle eccentricity=1.7, angle radius=0.4cm] {angle = node3--node1--node4};
\pic [draw, -, "${\theta}^4_{12}$", angle eccentricity=1.7, angle radius=0.4cm] {angle = node1--node4--node2};
\pic [draw, -, "${\theta}^3_{41}$", angle eccentricity=1.7, angle radius=0.4cm] {angle = node4--node3--node1};
\end{tikzpicture}%
}  \\
\subfigure[Proportionally congruent frameworks, where it is assumed that the two frameworks are globally rigid, i.e., the shapes of the frameworks are not globally deformable up to translations, rotations and scalings.]{\label{Fig:weak_example01_c}
\qquad\qquad\quad\begin{tikzpicture}[scale=0.58]
\node[place] (node1) at (0,2) [label=above:$1$] {};
\node[place] (node2) at (-3,0) [label=left:$2$] {};
\node[place] (node3) at (1.5,0) [label=right:$3$] {};
\node[place] (node4) at (0,-2) [label=below:$4$] {};

\draw[dashed] (node1)  -- node [above left] {$\norm{z_{12}}$} (node2);
\draw[dashed] (node1)  -- node [above right] {$\norm{z_{13}}$} (node3);
\draw[dashed] (node2)  -- node [below left] {$\norm{z_{24}}$} (node4);
\draw[dashed] (node3)  -- node [below right] {$\norm{z_{34}}$} (node4);
\draw[dashed] (node1)  -- node [above left] {$\norm{z_{14}}$} (node4);
\draw[dashed] (node2)  -- node [below left] {$\norm{z_{23}}$} (node3);

\end{tikzpicture}\qquad\qquad%
\begin{tikzpicture}[scale=0.75]
\node[place] (node1) at (0,2) [label=above:$1$] {};
\node[place] (node2) at (-3,0) [label=left:$2$] {};
\node[place] (node3) at (1.5,0) [label=right:$3$] {};
\node[place] (node4) at (0,-2) [label=below:$4$] {};

\draw[dashed] (node1)  -- node [above left] {$C\norm{z_{12}}$} (node2);
\draw[dashed] (node1)  -- node [above right] {$C\norm{z_{13}}$} (node3);
\draw[dashed] (node2)  -- node [below left] {$C\norm{z_{24}}$} (node4);
\draw[dashed] (node3)  -- node [below right] {$C\norm{z_{34}}$} (node4);
\draw[dashed] (node1)  -- node [above left] {$C\norm{z_{14}}$} (node4);
\draw[dashed] (node2)  -- node [below left] {$C\norm{z_{23}}$} (node3);

\end{tikzpicture}\qquad\qquad\quad%
} 
\caption{Examples of strongly equivalent, angle equivalent and proportionally congruent frameworks in $\mathbb{R}^{2}$.} \label{Fig:weak_example01}
\end{figure}
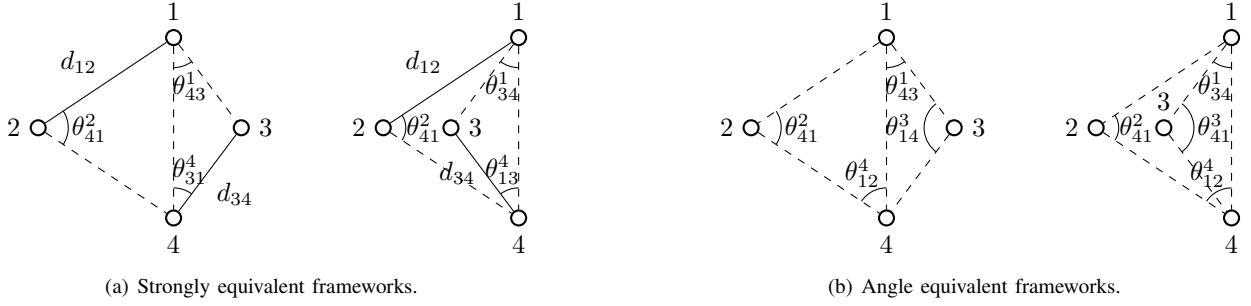
\begin{definition}[Generalized weak rigidity (GWR)]
\label{Def:weakRigidity}
A framework $(\mathcal{G},\mathcal{A},p)$ is \myemph{generalized weakly rigid (GWR)} in $\mathbb{R}^{d}$ if there exists a neighborhood $\mathcal{B}_{p} \subseteq \mathbb{R}^{dn}$ of $p$ such that each framework $(\mathcal{G},\mathcal{A},q)$, $q \in \mathcal{B}_{p}$, strongly equivalent to $(\mathcal{G},\mathcal{A},p)$ is congruent to $(\mathcal{G},\mathcal{A},p)$. Moreover, a framework $(\mathcal{G},\mathcal{A},p)$ with $\mathcal{E} = \emptyset$ is also \myemph{generalized weakly rigid (GWR)} in $\mathbb{R}^{d}$ if there exists a neighborhood $\mathcal{B}_{p} \subseteq \mathbb{R}^{dn}$ of $p$ such that each framework $(\mathcal{G},\mathcal{A},q)$, $q \in \mathcal{B}_{p}$, angle equivalent to $(\mathcal{G},\mathcal{A},p)$ is proportionally congruent to $(\mathcal{G},\mathcal{A},p)$.
\end{definition}

\begin{definition}[Global GWR]
\label{Def:G.weakRigidity}
A framework $(\mathcal{G},\mathcal{A},p)$ is \myemph{globally GWR} in $\mathbb{R}^{d}$ if any framework $(\mathcal{G},\mathcal{A},q)$ strongly equivalent to $(\mathcal{G},\mathcal{A},p)$ is congruent to $(\mathcal{G},\mathcal{A},p)$. Moreover, a framework $(\mathcal{G},\mathcal{A},p)$ with $\mathcal{E} = \emptyset$ is also \myemph{globally GWR} in $\mathbb{R}^{d}$ if any framework $(\mathcal{G},\mathcal{A},q)$ angle equivalent to $(\mathcal{G},\mathcal{A},p)$ is proportionally congruent to $(\mathcal{G},\mathcal{A},p)$.
\end{definition}

If a framework is GWR (resp. globally GWR), then the framework's shape is locally (resp. globally) rigid and not deformable up to translations and rotations of a given framework for $\mathcal{E} \neq \emptyset$ or up to translations, rotations and scalings for $\mathcal{E} = \emptyset$.
Fig.~\ref{Fig:weak_example02} shows several examples of GWR and non-GWR formations
in $\mathbb{R}^{2}$. The formations represented in Fig.~\ref{Fig:weak_example_(a)}, Fig.~\ref{Fig:weak_example_(b)} and Fig.~\ref{Fig:weak_example_(d)} are GWR since they cannot be deformed (in the case of Fig.~\ref{Fig:weak_example_(b)}, a deformed formation by scaling is GWR). 
In particular,  the formation in Fig.~\ref{Fig:weak_example_(a)} is globally GWR, and thus its shape is globally not deformable up to translations and rotations. The formation in Fig.~\ref{Fig:weak_example_(b)} is not globally GWR but GWR, that is, it is locally rigid up to translations, rotations and scalings since the agent 4 (or agent 2) can be flipped over edge $(1,3)$. 
Similarly, the formation in Fig.~\ref{Fig:weak_example_(d)} is not globally GWR but GWR.
On the other hand, the formation represented in Fig.~\ref{Fig:weak_example_(c)} is neither GWR nor globally GWR since it can be deformed by a smooth motion on a circle containing vertices $1$,$3$ and $4$.


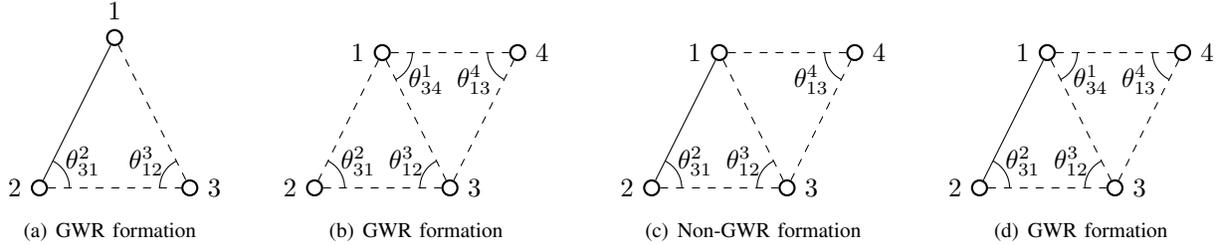
\begin{figure}[]
\centering
\subfigure[GWR formation]{\label{Fig:weak_example_(a)}
\begin{tikzpicture}[scale=1]
\node[place] (node1) at (0,2) [label=above:$1$] {};
\node[place] (node2) at (-1,0) [label=left:$2$] {};
\node[place] (node3) at (1,0) [label=right:$3$] {};

\draw[lineUD] (node1)  -- node [above left] {} (node2);
\draw[dashed] (node1)  -- (node3);
\draw[dashed] (node2)  -- (node3);

\pic [draw, -, "${\theta}^2_{31}$", angle eccentricity=1.7, angle radius=0.4cm] {angle = node3--node2--node1};
\pic [draw, -, "${\theta}^3_{12}$", angle eccentricity=1.7, angle radius=0.4cm] {angle = node1--node3--node2};
\end{tikzpicture}%
} \quad
\subfigure[GWR formation]{\label{Fig:weak_example_(b)}
\begin{tikzpicture}[scale=0.9]
\node[place] (node1) at (0,2) [label=left:$1$] {};
\node[place] (node2) at (-1,0) [label=left:$2$] {};
\node[place] (node3) at (1,0) [label=right:$3$] {};
\node[place] (node4) at (2,2) [label=right:$4$] {};

\draw[dashed] (node1)  -- (node2);
\draw[dashed] (node1)  -- (node3);
\draw[dashed] (node2)  -- (node3);
\draw[dashed] (node1)  -- (node4);
\draw[dashed] (node3)  -- (node4);

\pic [draw, -, "${\theta}^2_{31}$", angle eccentricity=1.7, angle radius=0.4cm] {angle = node3--node2--node1};
\pic [draw, -, "${\theta}^3_{12}$", angle eccentricity=1.7, angle radius=0.4cm] {angle = node1--node3--node2};
\pic [draw, -, "${\theta}^1_{34}$", angle eccentricity=1.7, angle radius=0.4cm] {angle = node3--node1--node4};
\pic [draw, -, "${\theta}^4_{13}$", angle eccentricity=1.7, angle radius=0.4cm] {angle = node1--node4--node3};
\end{tikzpicture} 
}  \quad
\subfigure[Non-GWR formation]{\label{Fig:weak_example_(c)}
\begin{tikzpicture}[scale=0.9]
\node[place] (node1) at (0,2) [label=left:$1$] {};
\node[place] (node2) at (-1,0) [label=left:$2$] {};
\node[place] (node3) at (1,0) [label=right:$3$] {};
\node[place] (node4) at (2,2) [label=right:$4$] {};

\draw[lineUD] (node1)  -- node [above left] {} (node2);
\draw[dashed] (node1)  -- (node3);
\draw[dashed] (node2)  -- (node3);
\draw[dashed] (node1)  -- (node4);
\draw[dashed] (node3)  -- (node4);

\pic [draw, -, "${\theta}^2_{31}$", angle eccentricity=1.7, angle radius=0.4cm] {angle = node3--node2--node1};
\pic [draw, -, "${\theta}^3_{12}$", angle eccentricity=1.7, angle radius=0.4cm] {angle = node1--node3--node2};
\pic [draw, -, "${\theta}^4_{13}$", angle eccentricity=1.7, angle radius=0.4cm] {angle = node1--node4--node3};
\end{tikzpicture} 
}\quad
\subfigure[GWR formation]{\label{Fig:weak_example_(d)}
\begin{tikzpicture}[scale=0.9]
\node[place] (node1) at (0,2) [label=left:$1$] {};
\node[place] (node2) at (-1,0) [label=left:$2$] {};
\node[place] (node3) at (1,0) [label=right:$3$] {};
\node[place] (node4) at (2,2) [label=right:$4$] {};

\draw[lineUD] (node1)  --  (node2);
\draw[dashed] (node1)  -- (node3);
\draw[dashed] (node2)  -- (node3);
\draw[dashed] (node1)  -- (node4);
\draw[dashed] (node3)  -- (node4);

\pic [draw, -, "${\theta}^2_{31}$", angle eccentricity=1.7, angle radius=0.4cm] {angle = node3--node2--node1};
\pic [draw, -, "${\theta}^3_{12}$", angle eccentricity=1.7, angle radius=0.4cm] {angle = node1--node3--node2};
\pic [draw, -, "${\theta}^1_{34}$", angle eccentricity=1.7, angle radius=0.4cm] {angle = node3--node1--node4};
\pic [draw, -, "${\theta}^4_{13}$", angle eccentricity=1.7, angle radius=0.4cm] {angle = node1--node4--node3};
\end{tikzpicture} 
} 
\caption{Examples of GWR and non-GWR formations in $\mathbb{R}^{2}$. The solid lines denote distance constraints belonging to  $\mathcal{E}$, but the dashed lines which do not belong to  $\mathcal{E}$ are not distance constraints.} \label{Fig:weak_example02}
\end{figure}

\subsection{Generalized infinitesimal weak rigidity (GIWR)}
 We now introduce the concept of the \myemph{generalized infinitesimal weak rigidity} which plays an important role in formation control studied in this paper.
  To define the concept, we first introduce a \myemph{weak rigidity matrix} with which we can check whether or not a formation is locally rigid by an algebraic manner, i.e., rank condition of the weak rigidity matrix. 
 
We define the following \myemph{weak rigidity function} $F_W:\chi \rightarrow \mathbb{R}^{m+w}$ for $\chi \subset\mathbb{R}^{dn}$, where $\chi$ is well defined not to make a denominator in $A_{h},h\in \{ 1,..., w\}$ zero, which describes constraints of edge lengths and angles in a framework:
\begin{equation}
F_{W}(p) = \left[ \norm{z_{1}}^2, ... ,\norm{z_{m}}^2, A_{1}, ... ,A_{w} \right]^\top
\in \mathbb{R}^{m+w}. 
\end{equation}
We then define the \myemph{weak rigidity matrix} as the Jacobian of the weak rigidity function: 
\begin{equation}\label{Weak_rigidity_matrix}
R_{W}(p) = \frac{\partial F_{W}(p)}{\partial p} =  \begin{bmatrix}
     \frac{\partial \mbf{D}}{\partial p}\\ 
    \frac{\partial \mbf{A}}{\partial p}
  \end{bmatrix} \in 
\mathbb{R}^{(m+w) \times dn},
\end{equation}
where $\mbf{D} = \left[\norm{z_{1}}^2,\norm{z_{2}}^2, ... ,\norm{z_{m}}^2\right]^\top \in \mathbb{R}^{m}$ and $\mbf{A} = \left[A_1,A_2,...,A_w\right]^\top \in \mathbb{R}^{w}$. 
Next, consider the constraints
\begin{align} 
\norm{p_{i}-p_{j}}^2&=constant, \,\forall (i,j) \in \mathcal{E},  \label{eq:inf_const01}\\
\cos{\theta_{ij}^{k}}&= constant, \, \forall (k,i,j) \in \mathcal{A}. \label{eq:inf_const02}
\end{align}
Then, the time derivative of \eqref{eq:inf_const01} is given by
\begin{align} 
&2\left(p_i-p_j \right)^\top \left(v_i-v_j \right) =0, \,\forall (i,j) \in \mathcal{E}, \label{eq:inf_motions01}
\end{align}
and the time derivative of \eqref{eq:inf_const02} is given as
\begin{align}
\left(v_k-v_i \right)^\top P^\top_{z_{ki}}\frac{z_{kj}}{\norm{z_{kj}}} 
+ \frac{z^\top_{ki}}{\norm{z_{ki}}}  P_{z_{kj}} \left(v_k-v_j \right)
= 0, \, \forall (k,i,j) \in \mathcal{A}, \label{eq:inf_motions02}
\end{align}
where $v_i$ is an infinitesimal motion of vertex $i$, and $P_{z_{ki}}=\frac{1}{\norm{z_{ki}}}\left[ I_d - \frac{z_{ki} z_{ki}^\top}{\norm{z_{ki}}^2}\right]$ and $P_{z_{kj}}=\frac{1}{\norm{z_{kj}}}\left[ I_d - \frac{z_{kj} z_{kj}^\top}{\norm{z_{kj}}^2}\right]$. 
For both cases $\mathcal{E} \neq \emptyset$ and $\mathcal{E} = \emptyset$, the equations \eqref{eq:inf_motions01} and \eqref{eq:inf_motions02} can be written in matrix form as $\dot{F}_W=\frac{\partial F_{W}(p)}{\partial p}\dot{p}=R_W(p)\dot{p}=0$.
We here denote an infinitesimal motion of $(\mathcal{G},\mathcal{A},p)$ by $\delta p$ if $R_W(p)\delta p=0$. The infinitesimal motions include rigid-body translations and rotations when $\mathcal{E} \neq \emptyset$. If $\mathcal{E} = \emptyset$ then the infinitesimal motions additionally include scalings, that is, the motions include rigid-body translations, rotations and scalings. We finally have the concept of the generalized infinitesimal weak rigidity with the following definition.

\begin{definition}[Trivial infinitesimal motion\cite{kwon2018infinitesimal}] \label{trivial}
An infinitesimal motion of a framework $(\mathcal{G},\mathcal{A},p)$ is called \myemph{trivial} if it corresponds to a rigid-body translation or a rigid-body rotation (or additionally, when $\mathcal{E} = \emptyset$, a scaling factor) of the entire framework.
\end{definition}

\begin{definition}[Generalized infinitesimal weak rigidity (GIWR)] \label{weak_rigidity_trivial}
A framework $(\mathcal{G},\mathcal{A},p)$ is \myemph{generalized infinitesimally weakly rigid (GIWR)} in $\mathbb{R}^{d}$ if all of its infinitesimal motions are trivial. 
\end{definition} 
 
We next explore properties of these concepts. For $d=2$ case, it is already shown that the GIWR can be checked by the rank condition of $R_W$ in \cite{kwon2018infinitesimal}. Therefore, we explore the properties only for $d=3$ case. We first state the trivial infinitesimal motions with mathematical expressions.
For $d=3$ case, we define the rotational matrix $J_i, \forall i \in \set{1,2,3}$ as
\begin{align} 
J_1=\begin{bmatrix}
0 & 0 & 0 \\
0 & 0 &-1\\
0 & 1 &0
\end{bmatrix},
J_2=\begin{bmatrix}
0 & 0 & 1 \\
0 & 0 &0\\
-1 & 0 &0
\end{bmatrix},
J_3=\begin{bmatrix}
0 & -1 & 0 \\
1 & 0 &0\\
0 & 0 &0
\end{bmatrix}.
\end{align} \normalsize
Note it always holds that $x^\top J_i x=0, \forall i \in \set{1,2,3}$ for any vector $x \in \mathbb{R}^{3}$. Referring to Lemma 1 in Sun et al. (2017)\cite{sun2017distributed}, we have that the vectors in the following set, $L_R$, are linearly independent.
\begin{align}
L_R=\set{\mathds{1}_n\otimes I_3, (I_n\otimes J_1)p, (I_n\otimes J_2)p, (I_n\otimes J_3)p},
\end{align}
where $(\mathds{1}_n\otimes I_3)$ and $(I_n\otimes J_i)p, i \in \set{1,2,3}$ correspond to a rigid-body translation and a rigid-body rotation of an entire framework, respectively.
We define a set $L_N$ for a rigid-body translation, a rigid-body rotation and a scaling of a framework in $\mathbb{R}^{3}$ as 
\begin{align}
L_N=\set{\mathds{1}_n\otimes I_3, (I_n\otimes J_1)p, (I_n\otimes J_2)p, (I_n\otimes J_3)p, p}.
\end{align}
The sets $L_R$ and $L_N$ can be regarded as the bases for $d$-dimensional rigid transformations and similarity transformations of a formation, respectively. Moreover, it is obvious that any linear combination of the vectors in $L_R$ cannot be equal to $\vspan\{p \}$ since a framework induced from $\vspan\{L_R \}$ is embedded in the $3$-dimensional group of rigid transformations, i.e., Special Euclidean group $SE(3)$, which means that rigid transformations $\vspan\{L_R \}$ cannot be equal to nonrigid transformations $\vspan\{p \}$. Hence, the vectors in the set $L_N$ are also linearly independent.

We state some notations to prove Lemmas \ref{lem_null_of_rigid matrix01} and \ref{lem_null_of_rigid matrix02} presented in what follows.
We first define a graph $\mathcal{G'}$ as $\mathcal{G'} = (\mathcal{V'},\mathcal{E'},\mathcal{A'})$ induced from $\mathcal{G}$ in such a way that:
\begin{itemize}
\item $\mathcal{V'} = \mathcal{V}$,

\item
$\mathcal{E'} =
\set{(i,j),(i,k),(j,k)\given (i,j) \in \mathcal{E} \lor(k,i,j) \in \mathcal{A}}$,

\item $\mathcal{A'} = \mathcal{A}$.
\end{itemize}
For any edge $(i,j) \in \mathcal{E'}$, we consider a new associated relative position vector ${z'}_{ij}$, and set the order of the new  relative position vector as follows
$$
{z'}_{s} = {z'}_{ij}, \forall s\in \{ 1,..., \eta\}, \eta \geq m,
$$
where ${z'}_{ij} = p_{i} - p_{j}$ for all $(i,j)\in \mathcal{E'}$, and $\eta=\card{\mathcal{E'}}$.
The anew defined relative position vector satisfies the following condition
$$
{z'}_{u} = z_{u}, \forall u\in \{ 1,..., m\}.
$$ 
We denote a new associated column vector composed of relative position vectors as $z' = \big[{z'}_{1}^\top, {z'}_{2}^\top,...,{z'}_{\eta}^\top \big]^\top \in \mathbb{R}^{3\eta}$.
The oriented incidence matrix $H' \in \mathbb{R}^{\eta \times n}$ of the induced graph $\mathcal{G'}$ is the $\{0, \pm1\}$-matrix with rows indexed by edges and columns indexed by vertices as follows: 
$$
[H']_{si}=\begin{cases}
    1 & \text{if the $s$-th edge sinks at vertex $i$} \\
    -1 & \text{if the $s$-th edge leaves vertex $i$} \\
	0 & \text{otherwise}
  \end{cases},
$$
where $[H']_{si}$ is an element at row $s$ and column $i$ of the matrix $H'$. Note that ${z'}$ satisfies ${z'}=\bar{H'}p$ where $\bar{H'}= H'\otimes I_d$. 
We are now ready to define the following properties.
\begin{lemma}\cite[Lemma 3.3]{kwon2018infinitesimal}
\label{Lem:null_of_rigid matrix_d2}
Let $J_0$ denote a rotational matrix defined as 
$J_0 = \begin{bmatrix}
0 &-1\\
1 & 0
\end{bmatrix}$ in $\mathbb{R}^2$.
For $d=2$ case,  it is satisfied that $\vspan\{\mathds{1}\otimes I_2, (I_n\otimes J_0)p \} \subseteq$ Null$(R_{W}(p))$ and $\rank(R_{W}(p))\leq 2n-3$ if $\mathcal{E} \neq \emptyset$. In addition, for $d=2$ case, it is satisfied that $\vspan\{\mathds{1}\otimes I_2, (I_n\otimes J_0)p, p \}\subseteq$ Null$(R_{W}(p))$ and $\rank(R_{W}(p))\leq 2n-4$ if $\mathcal{E} = \emptyset$.
\end{lemma}

\begin{lemma}\label{lem_null_of_rigid matrix01}
For $d=3$ case,  it is satisfied that, when $\mathcal{E} \neq \emptyset$ and $\mathcal{E} = \emptyset$, $\vspan(L_R) \subseteq$ $\vnull(R_{W}(p))$ and $\vspan(L_N) \subseteq$ $\vnull(R_{W}(p))$, respectively.
\end{lemma}

\begin{proof}
This property is proved by a similar approach to Lemma \ref{Lem:null_of_rigid matrix_d2}.
When $\mathcal{E} \neq \emptyset$, the equation (\ref{Weak_rigidity_matrix}) can be written as 
\begin{align}
R_{W}(p) = \frac{\partial F_{W}(p)}{\partial p}=  \begin{bmatrix}
     \frac{\partial \mbf{D}}{\partial {z'}} \frac{\partial {z'}}{\partial p}\\ 
    \frac{\partial \mbf{A}}{\partial {z'}} \frac{\partial {z'}}{\partial p}
  \end{bmatrix} = \begin{bmatrix}
     \frac{\partial \mbf{D}}{\partial {z'}} \bar{H'} \\ 
    \frac{\partial \mbf{A}}{\partial {z'}} \bar{H'}
  \end{bmatrix}  = \begin{bmatrix}
    \frac{\partial \mbf{D}}{\partial {z'}}\\ 
    \frac{\partial \mbf{A}}{\partial {z'}}
  \end{bmatrix}  \bar{H'}.
\end{align}
Then, it is obvious that $\vspan\{\mathds{1}_n\otimes I_3\} \subseteq \vnull (\bar{H'})$ $\subseteq$ $\vnull(R_{W}(p))$ since $\vspan\{\mathds{1}_n\} \subseteq \vnull(H')$. 
We next check whether $R_{W}(p) (I_n\otimes J_i)p=0$ or not. $\bar{H'}(I_n\otimes J_i)p, \forall i \in \set{1,2,3}$ can be of such form
\begin{align}
\bar{H'}(I_n\otimes J_i)p &= (H'\otimes I_3)(I_n\otimes J_i)p =(H'\otimes J_i)p \nonumber \\
&= (I_\eta H'\otimes J_i I_3)p = (I_\eta \otimes J_i)(H'\otimes I_3)p \nonumber \\
&=(I_\eta \otimes J_i){z'}= \begin{bmatrix}
J_i{z'}_1 \\
\vdots\\
J_i{z'}_\eta
\end{bmatrix}. 
\end{align}
From the viewpoint of $A_h=\frac{\norm{z_{ki}}^{2} + \norm{z_{kj}}^{2} - \norm{z_{ij}}^{2}}{2\norm{z_{ki}}\norm{z_{kj}}},(k,i,j) \in \mathcal{A}$,
if $A_h$ consists of ${z'}_a$, ${z'}_b$ and ${z'}_c$ for $a \neq b \neq c$ and $a,b,c \in \{ 1,..., \eta\}$ then  almost all elements of $\frac{\partial A_h}{\partial {z'}}$ are zero except for $\frac{\partial A_h}{\partial {z'}_a}$, $\frac{\partial A_h}{\partial {z'}_b}$ and $\frac{\partial A_h}{\partial {z'}_c}$. With reference to the form of $\frac{\partial A_h}{\partial {z'}}$ as presented in Lemma 3.1  in Kwon et al. (2018)\cite{kwon2018infinitesimal}, we have 
\begin{align}
\frac{\partial A_h}{\partial {z'}}\bar{H'}(I_n\otimes J_i)p=
\frac{\partial A_h}{\partial {z'}}\begin{bmatrix}
J_i{z'}_1 \\
\vdots\\
J_i{z'}_\eta
\end{bmatrix} 
=\frac{\partial A_h}{\partial {z'}_a}J_i{z'}_{a}+\frac{\partial A_h}{\partial {z'}_b}J_i{z'}_{b}+\frac{\partial A_h}{\partial {z'}_c}J_i{z'}_{c} 
= 0,
\end{align}
where ${{z'}_a}^\top J_i {z'}_a = 0$, ${{z'}_b}^\top J_i {z'}_b = 0$ and ${{z'}_c}^\top J_i {z'}_c = 0$ for all $i \in \set{1,2,3}$. Thus, $\frac{\partial \mbf{A}}{\partial {z'}}\bar{H'}(I_n\otimes J_i)p = 0$. We also have
\begin{align}
\frac{\partial \mbf{D}}{\partial {z'}}\bar{H'}(I_n\otimes J_i)p 
= \frac{\partial \mbf{D}}{\partial {z'}}\begin{bmatrix}
J_i{z'}_1 \\
\vdots\\
J_i{z'}_\eta
\end{bmatrix} 
= \begin{bmatrix} 
2D^\top & 0_{m,(3\eta-3m)}
\end{bmatrix} \begin{bmatrix}
J_i{z'}_1 \\
\vdots\\
J_i{z'}_\eta
\end{bmatrix} 
=0,
\end{align} 
where $D=$diag$({z'}_1,...,{z'}_m) \in \mathbb{R}^{3m \times m}$, and $0_{m,(3\eta-3m)}$ is a $m \times (3\eta-3m)$ zero matrix.
Using the above results, we have
\begin{align}
R_{W}(p)(I_n\otimes J_i)p = 0, \forall i \in \set{1,2,3},  
\end{align}
which implies that, when $\mathcal{E} \neq \emptyset$,  $\vspan\{(I_n\otimes J_i)p\} \subseteq \vnull(R_{W}(p)), \forall i \in \set{1,2,3}$.

If $\mathcal{E} = \emptyset$, then $R_{W}(p)$ is of the form
\begin{align}
R_{W}(p) = \frac{\partial F_{W}(p)}{\partial p}=
    \frac{\partial \mbf{A}}{\partial {z'}} \bar{H'}.
\end{align}
Then, $R_{W}(p)p = \frac{\partial \mbf{A}}{\partial {z'}} \bar{H'}p = \frac{\partial \mbf{A}}{\partial {z'}}{z'}$. 
With reference to Lemma 3.1 in Kwon et al.\cite{kwon2018infinitesimal}, the elements of $\frac{\partial A_h}{\partial {z'}}$ are zero except for $\frac{\partial A_h}{\partial {z'}_a}$, $\frac{\partial A_h}{\partial {z'}_b}$ and $\frac{\partial A_h}{\partial {z'}_c}$, and we have the following result:
\begin{align}
\frac{\partial A_h}{\partial {z'}}{z'} = \frac{\partial A_h}{\partial {z'}}\begin{bmatrix}
{z'}_1 \\
\vdots\\
{z'}_\eta
\end{bmatrix}
=&
\frac{\partial A_h}{\partial {z'}_a}{z'}_{a}+\frac{\partial A_h}{\partial {z'}_b}{z'}_{b}+\frac{\partial A_h}{\partial {z'}_c}{z'}_{c} \nonumber \\
=& \frac{\norm{{z'}_{a}}^{2} - \norm{{z'}_{b}}^{2} + \norm{{z'}_{c}}^{2}}{2\norm{{z'}_{a}}\norm{{z'}_{b}}} 
+\frac{-\norm{{z'}_{a}}^{2} + \norm{{z'}_{b}}^{2} + \norm{{z'}_{c}}^{2}}{2\norm{{z'}_{a}}\norm{{z'}_{b}}} +
\frac{-2\norm{{z'}_{c}}^{2}}{2\norm{{z'}_{a}}\norm{{z'}_{b}}} \nonumber \\
=& 0. \nonumber
\end{align}
Thus, we have $R_{W}(p)p=0$, which implies that $\vspan\{p\} \subseteq \vnull(R_{W}(p))$.
It also holds that, when $\mathcal{E} = \emptyset$, $\vspan\{\mathds{1}_n\otimes I_3\}\subseteq$ $\vnull(R_{W}(p))$ and $\vspan\{(I_n\otimes J_i)p\} \subseteq \vnull(R_{W}(p)), \forall i \in \set{1,2,3}$ in the same way as the case of $\mathcal{E} \neq \emptyset$.
Consequently, the statement is proved.
\end{proof}

\begin{lemma}\label{lem_null_of_rigid matrix02}
If $\mathcal{E} \neq \emptyset$, then $\rank(R_{W}(p))\leq dn-d(d+1)/2$ for a framework $(\mathcal{G},\mathcal{A},p)$ in $\mathbb{R}^{d}$. On the other hand, if $\mathcal{E} = \emptyset$, then $\rank(R_{W}(p))\leq dn-(d^2+d+2)/2$ for a framework $(\mathcal{G},\mathcal{A},p)$ in $\mathbb{R}^{d}$.
\end{lemma}
\begin{proof} 
For $d=2$ case, it holds that $\rank(R_{W}(p))\leq dn-d(d+1)/2$ and $\rank(R_{W}(p))\leq dn-(d^2+d+2)/2$ when $\mathcal{E} \neq \emptyset$ and $\mathcal{E} = \emptyset$, respectively, from Lemma \ref{Lem:null_of_rigid matrix_d2}.

For $d=3$ case, from Lemma \ref{lem_null_of_rigid matrix01}, we have $\vspan(L_R) \subseteq$ $\vnull(R_{W}(p))$ when $\mathcal{E} \neq \emptyset$, which implies that $\rank(R_{W}(p))\leq dn-d(d+1)/2$ since the vectors in $L_R$ are linearly independent. Similarly, when $\mathcal{E} = \emptyset$, we have $\vspan(L_N) \subseteq \vnull(R_{W}(p))$, which implies that $\rank(R_{W}(p))\leq dn-(d^2+d+2)/2$ since the vectors in $L_N$ are linearly independent.
\end{proof}

The following result shows the necessary and sufficient condition for the GIWR.
\begin{theorem}
\label{Thm:Inf_Rank} 
A framework $(\mathcal{G},\mathcal{A},p)$ with $n \ge 3$ and $\mathcal{E} \neq \emptyset$ is GIWR in $\mathbb{R}^{d}$ if and only if the weak rigidity matrix $R_{W}(p)$ has rank $dn-d(d+1)/2$. In addition, a framework $(\mathcal{G},\mathcal{A},p)$ with $n \ge 3$ and $\mathcal{E} = \emptyset$ is GIWR in $\mathbb{R}^{d}$ if and only if the weak rigidity matrix $R_{W}(p)$ has rank $dn-(d^2+d+2)/2$.
\end{theorem}
\begin{proof}
For $d=2$ case, the theorem was proved in Theorem 3.1 in Kwon et al. (2018)\cite{kwon2018infinitesimal}. We now prove it for $d=3$ case.

From Lemmas \ref{lem_null_of_rigid matrix01} and \ref{lem_null_of_rigid matrix02}, when $\mathcal{E} \neq \emptyset$, $\rank\left(R_{W}(p)\right) = dn-d(d+1)/2$  if and only if $\vnull\left(R_{W}(p)\right)=\vspan(L_R)$. Note that
$(\mathds{1}_n\otimes I_d)$ and $(I_n\otimes J_i)p$, $i \in \set{1,2,3}$ in $L_R$ correspond to a rigid-body translation and a rigid-body rotation of the entire framework, respectively. Therefore, for the case of $\mathcal{E} \neq \emptyset$, the theorem directly follows from Definition \ref{weak_rigidity_trivial}.

Similarly, when $\mathcal{E} = \emptyset$, $\rank\left(R_{W}(p)\right) = dn-(d^2+d+2)/2$ if and only if $\vnull\left(R_{W}(p)\right)=\vspan(L_N)$. Since
$(\mathds{1}_n\otimes I_d)$, $(I_n\otimes J_i)p,i \in \set{1,2,3}$ and $p$ in $L_N$ correspond to a rigid-body translation, a rigid-body rotation and a scaling of the entire framework, respectively, the remainder of the theorem for the $\mathcal{E} = \emptyset$ case directly follows from Definition \ref{weak_rigidity_trivial}.
\end{proof}

\begin{remark} \label{Remark:comparison}
Comparison with another publications: 
As Lemma~\ref{Lem:null_of_rigid matrix_d2} and Lemma~\ref{lem_null_of_rigid matrix01}, trivial infinitesimal motions in terms of $R_W$ correspond to translations, rotations and scalings when considering no distance constraint whereas those motions related to $\hat{R}_w$ correspond to only a subset of the motions, i.e., translations and rotations without scaling motions, where $\hat{R}_w$ denotes the rigidity matrix introduced in Jing et al. (2018)\cite{jing2018weak}. This difference is due to the fact that, in our work, inner products of inter-agent relative bearings, i.e., cosines of angles, are used to characterize a rigid formation whereas inner products of inter-agent relative positions are considered in Jing et al. (2018)\cite{jing2018weak}.
This fact can be checked by Lemma 3.6 in Jing et al. (2018)\cite{jing2018weak}. Therefore, the type-1 weak rigidity theory is distinct from our work.

In Jing et al. (2019)\cite{jing2019angle}, the angle rigidity theory is introduced, which is a similar concept to the weak rigidity theory in this paper when no distance constraint is considered. However, we deal with not only 2-dimensional cases but also 3-dimensional cases whereas Jing et al. (2019)\cite{jing2019angle} only studies 2-dimensional cases. In addition, we study globally exponential convergence of  formations whereas Jing et al. (2019)\cite{jing2019angle} does not. Therefore, our work can include the work in Jing et al. (2019)\cite{jing2019angle}. 
\end{remark}
\subsection{Relationship between distance rigidity and GWR}\label{Sec:Relationship_weak}
This subsection shows that the proposed theory, i.e., weak rigidity theory, is necessary for the distance rigidity theory\cite{asimow1978rigidity,asimow1979rigidity,roth1981rigid,C:Hendrickson:SIAM1992}. First, let us denote a classic framework without an angle set by $(\mathcal{G},p)$. We then reach the following result.
\begin{proposition}
\label{Pro:relation_rigidity}
If a classic framework  $(\mathcal{G},p)$ is distance rigid, globally distance rigid and infinitesimally distance rigid in $\mathbb{R}^{d}$, then the framework  $(\mathcal{G},\mathcal{A},p)$ is GWR, globally GWR and GIWR in $\mathbb{R}^{d}$, respectively.
\end{proposition}
\begin{proof}
First, the assumption that  $(\mathcal{G},p)$ is distance rigid means that 
there exists a neighborhood $\bar{\mathcal{B}}_{p} \subseteq \mathbb{R}^{dn}$ of $p$ such that $(\mathcal{G},q)$, $q \in \bar{\mathcal{B}}_{p}$, equivalent to $(\mathcal{G},p)$ is congruent to $(\mathcal{G},p)$ \cite{asimow1978rigidity}. 
Then, since the rigid shape of  $(\mathcal{G},p)$ is locally determined, it is obvious that $(\mathcal{G},\mathcal{A},p)$ is strongly equivalent and congruent  to  $(\mathcal{G},\mathcal{A},q)$, $q \in \bar{\mathcal{B}}_{p}$ for any $\mathcal{A}$. 
Therefore, $(\mathcal{G},\mathcal{A},p)$ is GWR from Definition \ref{Def:weakRigidity}. In the same way, it can be shown that global distance rigidity implies global GWR.

Next, consider the distance rigidity matrix $R_D$ defined as   $R_{D}(p) = \frac{1}{2}\frac{\partial \mbf{D}}{\partial p}$. If  $(\mathcal{G},p)$ is  infinitesimally distance rigid in $\mathbb{R}^{d}$, then 
$R_{D}(p)$ is of rank $dn-d(d+1)/2$ \cite{asimow1979rigidity,C:Hendrickson:SIAM1992}. With this fact, we can observe from the definition
$
R_{W}  =  \begin{bmatrix}
     \frac{\partial \mbf{D}}{\partial p}\\
    \frac{\partial \mbf{A}}{\partial p}
  \end{bmatrix}
$
that there exists a nonzero $(dn-d(d+1)/2) \times (dn-d(d+1)/2)$ minor of $R_{W}$.
Moreover, from Lemma \ref{lem_null_of_rigid matrix02}, we have that 
$\rank(R_{W}(p)) \leq dn-d(d+1)/2$ for $\mathcal{E} \neq \emptyset$. Therefore, $R_{W}$ is of rank $dn-d(d+1)/2$, which implies that $(\mathcal{G},\mathcal{A},p)$ is GIWR from Theorem \ref{Thm:Inf_Rank}.
\end{proof}
Due to the angle constraints, the GWR theory is not sufficient for the distance rigidity theory. The concept of `weak' is induced from the fact that the GWR theory is a weaker condition than the classic distance rigidity theory. 

\subsection{Generic property}
\label{Sec:Generic}
In this subsection, we show that both GWR and GIWR are generic properties.
First, we define two smooth manifolds as two sets $\mathcal{M}$ and $\mathcal{M}'$ composed of points congruent to $p$ and proportionally congruent to $p$, respectively. If the affine span of the configuration $p$ is $\mathbb{R}^{d}$ (or equivalently $p$ does not lie on any hyperplane in $\mathbb{R}^{d}$), then $\mathcal{M}$ is $d(d+1)/2$-dimensional and $\mathcal{M}'$ is $(d^2+d+2)/2$-dimensional, because $\mathcal{M}$ arises from the $d(d-1)/2$- and $d$-dimensional manifold of rotations and translations of $\mathbb{R}^{d}$, respectively, and $\mathcal{M}'$ arises from $d(d-1)/2$-, $d$- and $1$-dimensional manifold of rotations, translations and scalings of $\mathbb{R}^{d}$, respectively.

With the smooth map $F_W:\chi \rightarrow \mathbb{R}^{m+w}$ for some properly defined  $\chi\subset\mathbb{R}^{dn}$, let $r = \text{max}\set{\rank(\frac{\partial F_W}{\partial p}) \given p \in \mathbb{R}^{dn}}$. Then $p \in \mathbb{R}^{dn}$ is a \myemph{regular point} of $F_W$ if $\rank(\frac{\partial F_W}{\partial p})=r$, and a \myemph{singular point} otherwise.
With reference to Proposition 2 in Asimow et al. (1978)\cite{asimow1978rigidity}, if $p$ is a regular point of $F_W$ then there exists a neighborhood $\mathcal{B}_{p}$ of $p$ such that $F_W^{-1}(F_W(p))\cap\mathcal{B}_{p}$ is a $(dn-r)$-dimensional smooth manifold.

If $p_1, ..., p_n$ do not lie on any hyperplane in $\mathbb{R}^{d}$ when $\mathcal{E} \neq \emptyset$ then it follows from Lemma \ref{lem_null_of_rigid matrix02} that 
\begin{align}
\rank(\frac{\partial F_W}{\partial p}) = dn - \vnull(\frac{\partial F_W}{\partial p}) \leq dn-d(d+1)/2.
\end{align}
Moreover, if $p_1, ..., p_n$ do not lie on any hyperplane in $\mathbb{R}^{d}$ when $\mathcal{E} = \emptyset$ then, from Lemma \ref{lem_null_of_rigid matrix02}, we have that 
\begin{align}
\rank(\frac{\partial F_W}{\partial p}) = dn - \vnull(\frac{\partial F_W}{\partial p}) \leq dn-(d^2+d+2)/2.
\end{align}
In particular,  we have that if $p$ is a regular point of $F_W$ then $\rank(R_{W}(p))=dn-d(d+1)/2$ for $\mathcal{E} \neq \emptyset$ or $\rank(R_{W}(p))=dn-(d^2+d+2)/2$ for $\mathcal{E} = \emptyset$. We then have the following lemma.

\begin{lemma}
\label{Lem:generic}
Suppose that $p$ is a regular point of $F_W$ and the affine span of $p_1, ..., p_n$ is $\mathbb{R}^{d}$. A framework $(\mathcal{G},\mathcal{A},p)$ with $\mathcal{E} \neq \emptyset$ is GWR in $\mathbb{R}^{d}$ if and only if $\rank(R_W(p)) = dn-d(d+1)/2$. In addition, a framework $(\mathcal{G},\mathcal{A},p)$ with $\mathcal{E} = \emptyset$ is GWR in $\mathbb{R}^{d}$ if and only if $\rank(R_{W}(p)) = dn-(d^2+d+2)/2$.
\end{lemma}
\begin{proof}
Let us consider the case of $\mathcal{E} \neq \emptyset$. We have the fact that $R_{W}(p)$ has the maximum rank, i.e.,  $\rank(R_{W}(p))=dn-d(d+1)/2$. Then, $F_W^{-1}(F_W(p))\cap\mathcal{B}_{p}$ is $d(d+1)/2$-dimensional. 
Thus, $\mathcal{M}$ and $F_W^{-1}(F_W(p))\cap\mathcal{B}_{p}$ have the same dimension, which implies that the two sets agree near $p$. Consequently, $F_W^{-1}(F_W(p))\cap\mathcal{B}_{p}$ is the set of $q \in \mathbb{R}^{dn}$ such that $(\mathcal{G},\mathcal{A},q),q\in\mathcal{B}_{p}$, is strongly equivalent to $(\mathcal{G},\mathcal{A},p)$, and $\mathcal{M}$ is the set of $q \in \mathbb{R}^{dn}$ such that $q$ is congruent to $p$. 
Therefore, $(\mathcal{G},\mathcal{A},p)$ is GWR in $\mathbb{R}^{d}$ as defined in Definition \ref{Def:weakRigidity}. 

Similarly, when $\mathcal{E} = \emptyset$, $R_{W}(p)$ has the maximum rank, i.e., $\rank(R_{W}(p))=dn-(d^2+d+2)/2$. 
Therefore, $F_W^{-1}(F_W(p))\cap\mathcal{B}_{p}$ is $(d^2+d+2)/2$-dimensional. Two sets $\mathcal{M}'$ and $F_W^{-1}(F_W(p))\cap\mathcal{B}_{p}$ have the same dimension, and this implies that the two sets agree close to $p$. Consequently,
$F_W^{-1}(F_W(p))\cap\mathcal{B}_{p}$ is the set of $q \in \mathbb{R}^{dn}$ such that $(\mathcal{G},\mathcal{A},q),q\in\mathcal{B}_{p}$, is angle equivalent to $(\mathcal{G},\mathcal{A},p)$, and $\mathcal{M}'$ is the set of $q \in \mathbb{R}^{dn}$ such that $q$ is proportionally congruent to $p$. Therefore, $(\mathcal{G},\mathcal{A},p)$ is GWR in $\mathbb{R}^{d}$ as defined in Definition \ref{Def:weakRigidity}.

If $(\mathcal{G},\mathcal{A},p)$ is GWR in $\mathbb{R}^{d}$, then the two sets $F_W^{-1}(F_W(p))\cap\mathcal{B}_{p}$ and $\mathcal{M}$ are coincident near $p$, which implies that $F_W^{-1}(F_W(p))\cap\mathcal{B}_{p}$ and $\mathcal{M}$ have the same dimension and $\rank(R_{W}(p)) =r= dn-d(d+1)/2$ when $\mathcal{E} \neq \emptyset$ (resp. $\rank(R_{W}(p)) =r= dn-(d^2+d+2)/2$ when $\mathcal{E} = \emptyset$). Hence, we can conclude that the framework $(\mathcal{G},\mathcal{A},p)$ with $\mathcal{E} \neq \emptyset$ (resp. $\mathcal{E} = \emptyset$) is GWR in $\mathbb{R}^{d}$ if and only if $\rank(R_{W}(p)) = dn-d(d+1)/2$ (resp. $\rank(R_{W}(p)) = dn-(d^2+d+2)/2$).
\end{proof}

In general, a generic point introduced in Connelly et al. (2005)\cite{connelly2005generic} is used to derive a generic property; however, 
the notion of the generic point cannot be applied to our work since it cannot describe an equation involving angle constraints in a polynomial form. Thus, in this paper, we do not make use of the notion of the generic point. We next provide the following result to explore a relationship between GWR and GIWR
\begin{proposition}[Relationship between GWR and GIWR]
\label{Pro:rel_weak_infweak}
Suppose a framework $(\mathcal{G},\mathcal{A},p)$, $p = [p_{1}^\top,...,p_{n}^\top]^\top \in \mathbb{R}^{dn}$, is in $\mathbb{R}^{d}$ and the affine span of $p_1, ..., p_n$ is $\mathbb{R}^{d}$. Then, the framework $(\mathcal{G},\mathcal{A},p)$ is GIWR in $\mathbb{R}^{d}$ if and only if $p$ is a regular point of $F_W$ and $(\mathcal{G},\mathcal{A},p)$ is GWR in $\mathbb{R}^{d}$.
\end{proposition}
\begin{proof}
If a framework $(\mathcal{G},\mathcal{A},p)$ is GIWR, then it follows from Theorem \ref{Thm:Inf_Rank} that $R_{W}(p)$ is of rank $dn-d(d+1)/2$ or $dn-(d^2+d+2)/2$, and thus $p$ is a regular point. Moreover, with reference to the proof of Lemma \ref{Lem:generic}, we have that $(\mathcal{G},\mathcal{A},p)$ is GWR in $\mathbb{R}^{d}$. 

If $p$ is a regular point of $F_W$ and $(\mathcal{G},\mathcal{A},p)$ is GWR in $\mathbb{R}^{d}$, then $R_{W}(p)$ has the max rank, i.e., $dn-d(d+1)/2$ or $dn-(d^2+d+2)/2$, from the proof of Lemma \ref{Lem:generic}, which implies that the framework $(\mathcal{G},\mathcal{A},p)$ is GIWR from Theorem \ref{Thm:Inf_Rank}.
\end{proof}

We finally have the following result which shows that both GWR and GIWR for a framework are generic properties.
\begin{proposition}[Generic property]\label{Pro:generic_weak}
If a framework $(\mathcal{G},\mathcal{A},p)$ in $\mathbb{R}^{d}$ for a regular point $p$ of $F_W$  is GWR (resp. GIWR), then $(\mathcal{G},\mathcal{A},q)$ in $\mathbb{R}^{d}$ for any regular point $q$ of $F_W$ is GWR (resp. GIWR).
\end{proposition}
\begin{proof}
First, if $(\mathcal{G},\mathcal{A},p)$ is GIWR in $\mathbb{R}^{d}$, then $\rank(R_{W}(p))$ is equal to $dn-d(d+1)/2$ or $dn-(d^2+d+2)/2$. Moreover, it is clear that $(\mathcal{G},\mathcal{A},q)$ is also GIWR in $\mathbb{R}^{d}$ since $q$ is a regular point and it holds that $R_{W}(q)=R_{W}(p)$.

Next, if a framework $(\mathcal{G},\mathcal{A},p)$ is GWR and $p$ is a regular point of $F_W$ in $\mathbb{R}^{d}$, then the framework $(\mathcal{G},\mathcal{A},p)$ is GIWR in $\mathbb{R}^{d}$ from Proposition~\ref{Pro:rel_weak_infweak}. Moreover, $(\mathcal{G},\mathcal{A},q)$ is also GIWR, which implies that $(\mathcal{G},\mathcal{A},q)$ is GWR from Proposition~\ref{Pro:rel_weak_infweak}.
\end{proof}
\section{Application to formation control: local convergence of $n$-agent formations in $\mathbb{R}^d$}
\label{Sec:Formation_control_local_stab.}
We now apply the GWR theory to formation control problems. In this section, we particularly explore  local stability on $n$-agent formations in $\mathbb{R}^d$. This section aims to show local stability for minimally GIWR formations, and for non-minimally GIWR formations, where `local' means `close to the desired formation'. 
In distributed multi-agent systems,
the gradient flow law \cite{sakurama2015distributed,krick2009stabilisation,bishop2015distributed,park2014stability} is a popular approach, and we make use of the gradient flow approach to stabilize rigid formation shapes in this paper. 
We first rigorously define the concept of the minimally GIWR formation as follows.
\begin{definition}[Minimally GIWR]
A framework $(\mathcal{G},\mathcal{A},p)$ is \myemph{minimally GIWR} in $\mathbb{R}^d$ if the framework $(\mathcal{G},\mathcal{A},p)$  is GIWR in $\mathbb{R}^d$ and if no single distance or angle constraint can be removed without losing its GIWR.
\end{definition}
It is remarkable that if $(\mathcal{G},\mathcal{A},p)$ is minimally GIWR in $\mathbb{R}^d$ then $\rank(R_W)$ is exactly equal to the number of edge and angle constraints in the case of $\mathcal{E} \neq \emptyset$ (or only angle constraints in the case of $\mathcal{E} = \emptyset$), i.e., $\rank(R_W)=m+w$.
\subsection{Equations of motion based on gradient flow approach} \label{subsec:eq_motion}
We assume that each agent is governed by a single integrator, i.e.,
\begin{equation}
\frac{d}{dt}p_i=\dot{p}_i=u_i,\, i \in \mathcal{V},
\end{equation}
where time $t \in [0,\infty)$, and $u_i$ is a control input. 
Any entries in $u_i$ can be expressed by the relative position vectors of neighbors if a gradient flow law is employed. Note our formation control system makes use of the relative positions of neighbors as sensing variables, and the inter-agent distances and angles of neighbors as control variables.

We define the following two column vectors composed of $\norm{z_{g}}^2$ and $A_h$ as
\begin{align}
d_c(p) = \left[\ldots, \norm{z_{g_{ij}}}^2, \ldots \right]^\top_{(i,j) \in \mathcal{E}}, \,
c_c(p) = \left[\ldots, A_{h_{kij}}, \ldots \right]^\top_{(k,i,j) \in \mathcal{A}}.
\end{align}
Similarly, $d_c^*$ and $c_c^*$ are defined as
\begin{align} 
d_c^* = \left[\ldots, \norm{z^*_g}^2, \ldots \right]^\top, \,
c_c^* = \left[\ldots, A_h^*, \ldots \right]^\top,
\end{align}
where $\norm{z^*_g}^2$ and $A_h^*$ denote the desired values of $\norm{z_{g}}^2$ and $A_h$, respectively, and both of them are constants.
With the above definitions, an error vector is defined as follows
\begin{equation}
e(p)=\left[d_c(p)^\top c_c(p)^\top \right]^\top - \left[d_c^{*\top} c_c^{*\top} \right]^\top.
\end{equation}
The simple gradient flow law is employed to analyze a formation control system as follows
\begin{equation}
\dot{p}=u = -\left(\nabla\left(\frac{1}{2}e^\top(p) e(p)\right)\right)^\top.
\end{equation}
The control law can be expressed as
\begin{align}
\dot{p}=u 
&= -\left(\nabla\left(\frac{1}{2}e^\top(p) e(p)\right)\right)^\top =-R_{W}^\top(p) e(p) \nonumber \\
&=-\begin{bmatrix}s_1^\top & s_2^\top & \cdots & s_n^\top\end{bmatrix}^\top =-(E(p)\otimes I_d)p \label{control_law01}
\end{align} 
for $s_i \in \mathbb{R}^d$, $i \in \set{1, \cdots, n}$ and $E(p)\in \mathbb{R}^{n \times n}$. In $E(p)$, $\left[E(p) \right]_{ij}$ is an element at row $i$ and column $j$ and $\left[E(p) \right]_{ij}$ is the coefficient of the vector $p_j$ in $s_i$. According to the structure of \eqref{control_law01}, we can observe that the matrix $E(p)$ is symmetric (See an example (12) in Kwon et al. (2018)\cite{kwon2018infinitesimal}). 
The formation control system (\ref{control_law01}) is Lipschitz continuous since the system is continuously differentiable, which implies that the solution of (\ref{control_law01}) exists globally. With (\ref{control_law01}), we have the following error dynamics
\begin{equation}
\dot{e}=\frac{\partial e}{\partial p}\dot{p}=R_{W}(p)\dot{p}=-R_{W}(p)R_{W}^\top(p)e. \label{error_dym}
\end{equation}
The controller for agent $k$ in \eqref{control_law01} can be written by
\begin{align}
\dot{p}_k=
&-\underbrace{2\sum_{j\in\mathcal{N}^d_k}\left(\norm{z_{kj}}^2-\norm{z_{kj}^*}^2\right)(p_k-p_j)}_{(j,k) \in \mathcal{E}} 
- \underbrace{\sum_{i,j\in\mathcal{N}^a_k}\left( \cos{\theta_{ij}^{k}}- \cos{\left(\theta_{ij}^{k}\right)^*}\right) \left(\frac{\partial}{\partial p_k}\cos{\theta_{ij}^{k}}\right)^\top}_{(k,i,j) \in \mathcal{A}} \nonumber \\
&- \underbrace{\sum_{j,k\in\mathcal{N}^a_i}\left( \cos{\theta_{jk}^{i}}- \cos{\left(\theta_{jk}^{i}\right)^*}\right) \left(\frac{\partial}{\partial p_k}\cos{\theta_{jk}^{i}}\right)^\top}_{\text{if }\exists (i,j,k) \in \mathcal{A}}, \label{eq:distributed_system}
\end{align}
where $\norm{z_{kj}^*}$ and $\left(\theta_{ij}^{k}\right)^*$ are the desired values for $\norm{z_{kj}}$ and $\theta_{ij}^{k}$, respectively, and $\mathcal{N}^d_k=\set{j \in\mathcal{V} \given (j,k)\in\mathcal{E}}$ and $\mathcal{N}^a_k=\set{i,j \in\mathcal{V} \given  (k,i,j) \in \mathcal{A}}$ denote the neighbor sets for agent $k$ related to distance and angle constraints, respectively; refer to Example 1 in Appendix.
Therefore, it is clear that the system is a distributed system since each agent requires only local information.
Moreover, according to the control system \eqref{eq:distributed_system}, we need to define the following assumption for a sensing topology.
\begin{assumption} \label{asm:sensig}
The sensing graph is characterized by an undirected graph $\mathcal{G}_s = (\mathcal{V}_s,\mathcal{E}_s)$ and agent $k$ can measure relative position vectors in terms of its neighbor set $\mathcal{N}^s_k$, where $\mathcal{V}_s = \mathcal{V}$, $\mathcal{E}_s =\set{(i,j),(i,k),(j,k)\given (i,j) \in \mathcal{E} \lor(k,i,j) \in \mathcal{A}}$ and $\mathcal{N}^s_k=\set{j \in \mathcal{V}\given  (j,k)\in\mathcal{E}_s}$.
\end{assumption}
Note that the proposed controller does not require any communication among agents.
The following result will be useful for next analysis, which shows that if a differential equation $\dot{X}(t)=f(t,X)$ satisfies the following result then the rank of the solution $X(t)$ is constant for all $t \geq 0$ and $\dot{X}(t)$ is said to be \myemph{rank-preserving}.
\begin{lemma}\cite[Lemma 2]{sun2015rigid}
\label{Lem:rank_preserving}
Let $A(t) \in \mathbb{R}^{M \times M}$ and $B(t) \in \mathbb{R}^{N \times N}$ be a continuous time-varying family of matrices. Then, the following differential equation
\begin{align}
\dot{X}(t)=A(t)X(t)+X(t)B(t), \, X(0) \in \mathbb{R}^{M \times N}
\end{align}
is rank-preserving.
\end{lemma}
We next show some properties of the formation control system with the gradient flow approach.
\begin{lemma}
\label{Lem:grad_law_properties}
Under the gradient flow law, the formation control system designed in (\ref{control_law01}) has the following properties:
\begin{enumerate}[(i)]
\item The controller is distributed.
\item The controller and measurement for each agent are independent of any global coordinates. That is, only the local coordinate system for each agent is required to measure relative positions and to implement the control signals. 
\item \label{Lem:cent_scale_inv} The centroid $p^o=\frac{1}{n}\sum_{i=1}^{n}p_i$ is stationary. In the case of $\mathcal{E} = \emptyset$, the centroid and the scale $p^s=\sqrt{\frac{1}{n}\sum_{i=1}^{n}\norm{p_i-p^o}^2}$ are both invariant for all $t\geq0$.

\item \label{Lem:grad_law_properties_Cp} Denote $C_p = \begin{bmatrix}p_1 & p_2 & \cdots & p_n\end{bmatrix}\in \mathbb{R}^{d\times n}$. Then, $\rank\left(C_p(0)\right)=\rank\left(C_p(t)\right)$ for all time $t\geq0$. Moreover, if $C_p$ is of full row rank, then all of $p_i, \forall i\in \set{1,\cdots,n}$ do not lie on a hyperplane. On the other hand, if $C_p$ is not of full row rank, then there exists a hyperplane containing all $p_i, \forall i \in \set{1,2,\cdots,n}$.

\item \label{Lem:col_avo}[Collision avoidance] Let $p^*=[p^{*\top}_{1},...,p^{*\top}_{n}]^\top \in \mathbb{R}^{dn}$ denote the desired configuration. Then, it is guaranteed that $\norm{p_i(t)-p_j(t)}>\zeta$ for all $t \geq 0$ and $i,j \in \mathcal{V}$ if 
$\norm{p_i^*-p_j^*}-\sqrt{n}\norm{p(0)-(\mathds{1}_n \otimes p^o)}-\sum_{l=1}^n\norm{p^o-p_l^*} >\zeta$ for $\zeta>0$.

\item \label{Lem:rank_preserv_d+1} If a framework $(\mathcal{G},\mathcal{A},p(0))$ with $n=d+1$ vertices is minimally GIWR in $\mathbb{R}^{d}$ and $C_p(0)$ is of full row rank, then $(\mathcal{G},\mathcal{A},p(t))$ is minimally GIWR in $\mathbb{R}^{d}$ for all $t\geq0$, i.e., $\rank\left(R_W(p(0))\right)=\rank\left(R_W(p(t))\right)$ for all $t\geq0$. 
\end{enumerate}
\end{lemma} 
\begin{proof} 
(i) This property is obvious from (\ref{control_law01}). 

\quad(ii) This property is proved in a similar way to Lemma 4 in Sun et al. (2016)\cite{sun2016finite}.
First, let us denote a measurement in a global coordinate system by $(\cdot)^g$.
Observe the fact that there exists a rotation matrix $Q_k \in \mathbb{R}^{d \times d}$ such that $p_j=Q_k p^g_j+v$, where $v$ denotes a translation vector. 
Then, we can express  \eqref{eq:distributed_system} in terms of the global coordinate system as follows
\begin{align}
\dot{p}^g_k 
=u^g_k
=Q_k^{-1}u_k
=&-2Q_k^{-1}\sum_{j\in\mathcal{N}^d_k}\left(\norm{z_{kj}}^2-\norm{z_{kj}^*}^2\right)^gQ_k z_{kj}^g
- Q_k^{-1}\sum_{i,j\in\mathcal{N}^a_k}\left( \cos{\theta_{ij}^{k}}- \cos{\left(\theta_{ij}^{k}\right)^*}\right)^g Q_k\left(\frac{\partial}{\partial p_k}\cos{\theta_{ij}^{k}}\right)^{g\top} \nonumber \\
&- Q_k^{-1}\underbrace{\sum_{j,k\in\mathcal{N}^a_i}\left( \cos{\theta_{jk}^{i}}- \cos{\left(\theta_{jk}^{i}\right)^*}\right)^g Q_k \left(\frac{\partial}{\partial p_k}\cos{\theta_{jk}^{i}}\right)^{g\top}}_{\text{for }\exists (i,j,k) \in \mathcal{A}} \nonumber \\
%
%
=&-2\sum_{j\in\mathcal{N}^d_k}\left(\norm{z_{kj}}^2-\norm{z_{kj}^*}^2\right)^g z_{kj}^g 
-\sum_{i,j\in\mathcal{N}^a_k}\left( \cos{\theta_{ij}^{k}}- \cos{\left(\theta_{ij}^{k}\right)^*}\right)^g \left(\frac{\partial}{\partial p_k}\cos{\theta_{ij}^{k}}\right)^{g\top} \nonumber \\
&- \underbrace{\sum_{j,k\in\mathcal{N}^a_i}\left( \cos{\theta_{jk}^{i}}- \cos{\left(\theta_{jk}^{i}\right)^*}\right)^g \left(\frac{\partial}{\partial p_k}\cos{\theta_{jk}^{i}}\right)^{g\top}}_{\text{for }\exists (i,j,k) \in \mathcal{A}},
\end{align}where we have used the fact that
\begin{align}
\frac{\partial}{\partial p_k} \cos(\theta^k_{ij}) 
=\frac{\partial}{\partial p_k} \frac{z^\top_{ki}}{\norm{z_{ki}}} \frac{z_{kj}}{\norm{z_{kj}}} 
&=\frac{z^\top_{kj}}{\norm{z_{kj}}} \frac{1}{\norm{z_{ki}}}\left( I_d - \frac{z_{ki}z^\top_{ki}}{\norm{z_{ki}}^2}\right) 
 +\frac{z^\top_{ki}}{\norm{z_{ki}}} \frac{1}{\norm{z_{kj}}}\left( I_d - \frac{z_{kj}z^\top_{kj}}{\norm{z_{kj}}^2}\right)  \nonumber \\
&=\frac{z^{g\top}_{kj}}{\norm{z^g_{kj}}}Q_k^{-1} \frac{1}{\norm{z^g_{ki}}}\left( I_d - Q_k \frac{z^g_{ki}}{\norm{z^g_{ki}}} \frac{z^{g\top}_{ki}}{\norm{z^g_{ki}}} Q_k^{-1} \right)
+\frac{z^{g\top}_{ki}}{\norm{z^g_{ki}}}Q_k^{-1} \frac{1}{\norm{z^g_{kj}}}\left( I_d - Q_k \frac{z^g_{kj}}{\norm{z^g_{kj}}} \frac{z^{g\top}_{kj}}{\norm{z^g_{kj}}} Q_k^{-1} \right)   \nonumber \\
&=\left(\frac{\partial}{\partial p_k}\cos{\theta_{ij}^{k}}\right)^g Q_k^{-1},
\end{align} and, in the same way, $\frac{\partial}{\partial p_k} \cos(\theta^i_{jk})= \left(\frac{\partial}{\partial p_k}\cos{\theta_{jk}^{i}}\right)^g Q_k^{-1}$ .
Thus, we conclude the statement.

\quad(iii) Since $p^o=\frac{1}{n}\sum_{i=1}^{n}p_i=\frac{1}{n}(\mathds{1}_n\otimes I_d)^\top p\in \mathbb{R}^d$, the following time derivative holds.
\begin{align}
\dot{p}^o=\frac{1}{n}(\mathds{1}_n\otimes I_d)^\top\dot{p}&=-\frac{1}{n}(\mathds{1}_n\otimes I_d)^\top R_{W}^\top(p) e(p) 
=-\frac{1}{n}\left(
\begin{bmatrix}
    \frac{\partial \mbf{D}}{\partial {z'}}\\
    \frac{\partial \mbf{A}}{\partial {z'}}
  \end{bmatrix}\bar{H'} (\mathds{1}_n\otimes I_d)
  \right)^\top e(p)
\end{align} 
Since $\vspan(\mathds{1}_n\otimes I_d) \subseteq \vnull (\bar{H'})\subseteq \vnull \left(R_{W}(p) \right)$, $R_{W}(p) (\mathds{1}_n\otimes I_d)=0$ and this implies that $\dot{p}^o=0$. Moreover, it also holds that $\dot{p}^o=0$ for the case of $\mathcal{E} = \emptyset$.

In the case of $\mathcal{E} = \emptyset$, there is no constraint for the scale of the given framework. Note that $p^s=\sqrt{\frac{1}{n}\sum_{i=1}^{n}\norm{p_i-p^o}^2}=\norm{p-\mathds{1}_n\otimes p^o}/\sqrt{n}$. With the fact that $\dot{p}^o=0$, we have
\begin{equation}
\dot{p}^s=\frac{1}{\sqrt{n}}\frac{(p-\mathds{1}_n\otimes p^o)^\top}{\norm{p-\mathds{1}_n\otimes p^o}}\dot{p}.
\end{equation}
It holds that $p^\top\dot{p}=-\left(R_{W}(p)p\right)^\top e(p)=0$ and $(\mathds{1}_n\otimes p^o)^\top\dot{p} =-\left(R_{W}(p)(\mathds{1}_n\otimes p^o)\right)^\top e(p) =0$ since $\vspan(p)\subseteq \vnull(R_{W})$ and $\vspan(\mathds{1}_n\otimes p^o)\subseteq \vnull(\bar{H'})\subseteq \vnull\left(R_{W}(p) \right)$. Therefore, $\dot{p}^s=0$. Hence, the statement is proved.

\quad(iv) \label{Lem:proof_rank_pre_cp}Since $\dot{p}(t)=-(E(p)\otimes I_d)p(t)$, the vector differential equation can be expressed as the following matrix differential equation.
\begin{equation}
\dot{C}_p(t) = - C_p(t)E^\top(p(t)) \in \mathbb{R}^{d\times n}. \label{rank-preserving}
\end{equation}
From Lemma \ref{Lem:rank_preserving}, the matrix differential equation (\ref{rank-preserving}) is rank-preserving for any finite time $t \geq 0$.

If $C_p$ is not of full row rank, then there exists a nontrivial solution $x$ such that $C^\top_p x=0$. This implies that $p_1^\top x=p_2^\top x=\cdots=p_n^\top x = 0$ and $(p_i^\top-p_j^\top) x=z_{ij}^\top x=0$ for all $i,j \in \mathcal{V}$ and $i \neq j$, which means that all of vectors $z_{ij}$ are orthogonal to the vector $x$ and further all of vectors $z_{ij}$ lie on a hyperplane. Hence, there exists a hyperplane containing all $p_i, \forall i \in \set{1,2,\cdots,n}$ if $C_p$ is not of full row rank. 

\quad(v) For any $i,j \in \mathcal{V}$ and $t\geq0$, we have the following equation
\begin{align}
\norm{p_i(t)-p_j(t)}  
&=\norm{\left(p_i(t)-p_i^*\right)-\left(p_j(t)-p_j^*\right)+\left(p_i^*-p_j^*\right)} \nonumber \\
&\geq \norm{p_i^*-p_j^*}-\norm{p_i(t)-p_i^*}-\norm{p_j(t)-p_j^*} \nonumber \\
&\geq \norm{p_i^*-p_j^*}-\sum_{l=1}^n\norm{p_l(t)-p_l^*},
\end{align}
where
\begin{align}
\norm{p_i^*-p_j^*}-\sum_{l=1}^n\norm{p_l(t)-p_l^*}
&= \norm{p_i^*-p_j^*}-\sum_{l=1}^n\norm{\left(p_l(t)-p^o\right) + \left(p^o-p_l^*\right)} \nonumber \\
&\geq \norm{p_i^*-p_j^*}-\sum_{l=1}^n\norm{p_l(t)-p^o}-\sum_{l=1}^n\norm{p^o-p_l^*} \nonumber \\
&\geq \norm{p_i^*-p_j^*}-\sqrt{n}\norm{p(t)-(\mathds{1}_n \otimes p^o)}-\sum_{l=1}^n\norm{p^o-p_l^*}. \label{eq:cor_sub02}
\end{align}
In the above inequality (\ref{eq:cor_sub02}), it holds that $\sqrt{n}\norm{p(t)-(\mathds{1}_n \otimes p^o)} \geq \sum_{l=1}^n\norm{p_l(t)-p^o}$
by using the following inequality for positive real numbers $x_1,\cdots,x_n$.
\begin{align}
\sqrt{\frac{x_1^2+\cdots+x_n^2}{n}} \geq \frac{x_1+\cdots+x_n}{n}.
\end{align}
Since $\norm{p(t)-(\mathds{1}_n \otimes p^o)}$ has the similar form to $p^s$ as given in the proof of Lemma \ref{Lem:grad_law_properties}-\eqref{Lem:cent_scale_inv}, the time derivative of $\norm{p(t)-(\mathds{1}_n \otimes p^o)}$ equals zero, and this follows that $\norm{p(t)-(\mathds{1}_n \otimes p^o)}$ is invariant for all $t \geq 0$. Here $p^o$ is also invariant. Thus, if $\norm{p_i^*-p_j^*}-\sqrt{n}\norm{p(0)-(\mathds{1}_n \otimes p^o)}-\sum_{l=1}^n\norm{p^o-p_l^*}$ is greater than $\zeta$ for $\zeta>0$ at $t=0$, then $\norm{p_i(t)-p_j(t)}$ is also greater than $\zeta$ for all $t \geq 0$.

\quad(vi) This proof is motivated by Theorem 4.4 in Jing et al. (2018)\cite{jing2018weak}, and follows the logic as shown in Table~\ref{table:logic_box01}. We can state $R_W(p(0))=\begin{bmatrix}\mbf{r}_1 & \mbf{r}_2 & \cdots & \mbf{r}_\sigma\end{bmatrix}^\top=\begin{bmatrix}\mbf{c}_1 & \mbf{c}_2 & \cdots & \mbf{c}_n\end{bmatrix}$, where $\mbf{r}_i \in \mathbb{R}^{dn}$, $i \in \set{1, \cdots, \sigma}$, $\mbf{c}_j \in \mathbb{R}^{\sigma \times d}$, $j \in \set{1, \cdots, n}$, and $\sigma=m+w$. 
We define a set $\mathcal{N'}_l$ of neighbors of vertex $l$ as $\mathcal{N'}_l = \set{i,j \in \mathcal{V} \given (l,i) \in \mathcal{E} \lor (l,i,j) \in \mathcal{A}}$.
If a framework $(\mathcal{G},\mathcal{A},p)$ with $n=d+1$ vertices is minimally GIWR, then each agent has exactly $d$ neighbors, i.e., $\card{\mathcal{N'}_l}=n-1=d$.

Let a framework $(\mathcal{G},\mathcal{A},p(0))$ with $n=d+1$ vertices be minimally GIWR, and let $C_p(0)$ be of full row rank. Suppose that the framework $(\mathcal{G},\mathcal{A},p(t^*))$ is not GIWR at specific time $t^*>0$. Then, $R_W(p(t^*))$ does not have full row rank, and further there exists a nonzero vector $\tau=\begin{bmatrix}\tau_1 & \tau_2 & \cdots & \tau_\sigma\end{bmatrix}^\top \in \mathbb{R}^\sigma$ such that $\tau^\top R_W(p(t^*))=\tau_1\mbf{r}^\top_1+\tau_2\mbf{r}^\top_2+\cdots+\tau_\sigma\mbf{r}^\top_\sigma=0$ (or equivalently $\tau_1\mbf{r}_1+\tau_2\mbf{r}_2+\cdots+\tau_\sigma\mbf{r}_\sigma=0$). Since $\tau^\top R_W(p(t^*)) = \tau^\top \begin{bmatrix}\mbf{c}_1 & \mbf{c}_2 & \cdots & \mbf{c}_n\end{bmatrix}=0$, $\tau^\top \mbf{c}_l=\tau^\top \frac{\partial F_W}{\partial p_l}=0$ for all $l\in\set{1, 2, \cdots, n}$. 
Note that each entry for the weak rigidity matrix $R_W$ is composed of inter-neighbor relative position vectors from a framework $(\mathcal{G},\mathcal{A},p)$. 
From the fact that $\frac{\partial F_W}{\partial p_l}$ consists of ${z'}_{lk}^\top(t^*)$ for $k\in\mathcal{N'}_l$ and $\tau^\top \mbf{c}_l=0$, 
there must exist at least one case from $l=1$ to $l=n$ such that ${z'}_{lk}^\top(t^*)$ for $k\in\mathcal{N'}_l$ are linearly dependent.

With $\card{\mathcal{N'}_l}=n-1=d$, we can denote an oriented incidence matrix $H_l$ associated with the vertex $l$ (for example, see Fig.~\ref{exam:inciden_l}), where $H_l \in \mathbb{R}^{d \times (d+1)}$ for all $l \in \set{1, \cdots, n}$. 
We define a matrix $E_l(t^*)$ composed of ${z'}_{lk}^\top(t^*)$ for $k\in\mathcal{N'}_l$ as $E_l(t^*)=H_l C^\top_p(t^*) \in \mathbb{R}^{d \times d}$. We can state $E_l(t^*)$ as $E_l(t^*)=\left[\cdots, {z'}_{lk}(t^*), \cdots \right]^\top$. Consider $E_l(t^*) x=0$ for any nontrivial $x \in \mathbb{R}^d$ and $l \in \set{1, \cdots, n}$, then either the equality $C^\top_p(t^*) x =0$ or the equality ${z'}_{ij}^\top x = 0, \forall i,j \in \mathcal{V'}$ holds. The equality ${z'}_{ij}^\top x = 0, \forall i,j \in \mathcal{V'}$ means that all of vectors ${z'}_{ij}$ are orthogonal to the vector $x$, and further all of vectors ${z'}_{ij}$ lie on a hyperplane. Thus, the equality ${z'}_{ij}^\top x = 0$ cannot hold as proved in Lemma \ref{Lem:grad_law_properties}-\eqref{Lem:grad_law_properties_Cp}. The equality $C^\top_p(t^*) x =0$ cannot also hold since $C_p(t^*)$ has the full row rank for all $t \geq 0$ as proved in Lemma \ref{Lem:grad_law_properties}-\eqref{Lem:grad_law_properties_Cp}. Hence, $\vnull\left(E_l(t)\right)=\emptyset$ and the rank of $E_l(t^*)$ equals $d$. However, 
there exist at least one case such that ${z'}_{lk}^\top(t^*)$ for $k\in\mathcal{N'}_l$ are linearly dependent, and this follows that $\rank\left(E_l(t^*)\right)<d$. This conflicts with $\rank\left(E_l(t^*)\right)=d$. Hence, we can conclude that $(\mathcal{G},\mathcal{A},p(t))$ is minimally GIWR for all $t\geq0$ if $(\mathcal{G},\mathcal{A},p(0))$ with $n=d+1$ vertices is minimally GIWR and $C_p(0)$ is of full row rank.
\end{proof}

\begin{assumption}
In formation control addressed in this paper, it is assumed that any two agents at the initial time are sufficiently far from each other to not make any collision between agents with Lemma~\ref{Lem:grad_law_properties}-\ref{Lem:col_avo}.
\end{assumption}

\begin{table}
 \begin{center}
\begin{tabular}{r l}
\hline\hline
Suppose 	& $(\mathcal{G},\mathcal{A},p(t^*))$ is not GIWR at specific time $t^*>0$.\\
Thus 	& $\rank(E_l(t^*))<d$, where \\ 
		& $E_l(t^*)=\left[\cdots, {z'}_{lk}(t^*), \cdots \right]^\top\in \mathbb{R}^{d \times d}$.  \\ 
Assumption 1& $(\mathcal{G},\mathcal{A},p(0))$ is minimally GIWR.  \\
Assumption 2& $C_p(0)$ is of full row rank. \\
Thus 	& $\rank(E_l(t^*))=d$. \\
Contradiction& Consequently, $(\mathcal{G},\mathcal{A},p(t^*))$ is GIWR at specific time \\ 
		&$t^*>0$. \\
\hline
\end{tabular}
\end{center} 
\caption{The logic for the proof of Lemma \ref{Lem:grad_law_properties}-\eqref{Lem:rank_preserv_d+1}.} \label{table:logic_box01}
\end{table}

\begin{figure}[]
\centering
\subfigure[Subgraph for $H_1$]{ 
\begin{tikzpicture}[scale=0.9]
\node[place] (node1) at (-1.2,0) [label=left:$1$] {};
\node[place] (node2) at (0,1) [label=above:$2$] {};
\node[place] (node3) at (0,-1) [label=below:$3$] {};
\node[place] (node4) at (1.2,0) [label=right:$4$] {};

\draw[lineUD] (node1)  -- (node2);
\draw[lineUD] (node1)  -- (node3);
\draw[dashed] (node2)  --  (node3);
\draw[dashed] (node2)  -- (node4);
\draw[dashed] (node3)  -- (node4);
\end{tikzpicture}
}\qquad\,
\subfigure[Subgraph for $H_2$]{ 
\begin{tikzpicture}[scale=0.9]
\node[place] (node1) at (-1.2,0) [label=left:$1$] {};
\node[place] (node2) at (0,1) [label=above:$2$] {};
\node[place] (node3) at (0,-1) [label=below:$3$] {};
\node[place] (node4) at (1.2,0) [label=right:$4$] {};

\draw[lineUD] (node1)  -- (node2);
\draw[dashed] (node1)  -- (node3);
\draw[lineUD] (node2)  --  (node3);
\draw[lineUD] (node2)  -- (node4);
\draw[dashed] (node3)  -- (node4);
\end{tikzpicture}
} \qquad\,
\subfigure[Subgraph for $H_3$]{ 
\begin{tikzpicture}[scale=0.9]
\node[place] (node1) at (-1.2,0) [label=left:$1$] {};
\node[place] (node2) at (0,1) [label=above:$2$] {};
\node[place] (node3) at (0,-1) [label=below:$3$] {};
\node[place] (node4) at (1.2,0) [label=right:$4$] {};

\draw[dashed] (node1)  -- (node2);
\draw[lineUD] (node1)  -- (node3);
\draw[lineUD] (node2)  --  (node3);
\draw[dashed] (node2)  -- (node4);
\draw[lineUD] (node3)  -- (node4);
\end{tikzpicture}
}\qquad\,
\subfigure[Subgraph for $H_4$]{ 
\begin{tikzpicture}[scale=0.9]
\node[place] (node1) at (-1.2,0) [label=left:$1$] {};
\node[place] (node2) at (0,1) [label=above:$2$] {};
\node[place] (node3) at (0,-1) [label=below:$3$] {};
\node[place] (node4) at (1.2,0) [label=right:$4$] {};

\draw[dashed] (node1)  -- (node2);
\draw[dashed] (node1)  -- (node3);
\draw[dashed] (node2)  --  (node3);
\draw[lineUD] (node2)  -- (node4);
\draw[lineUD] (node3)  -- (node4);
\end{tikzpicture}
}
\caption{Example of subgraphs for $H_l$ when $n=4$. The dashed lines indicate the removed edges. The graphs have the same vertex set but do not have the same edge set.} \label{exam:inciden_l}
\end{figure}
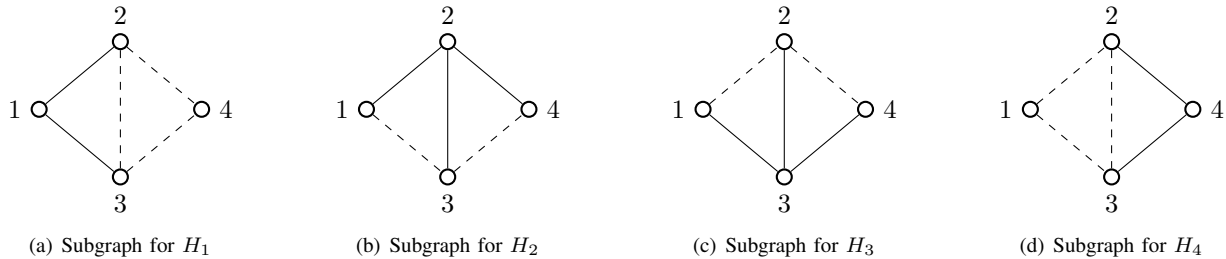

\subsection{Exponential stability of minimally GIWR formations with $n$ agents in $\mathbb{R}^d$}

We first explore the stability of minimally GIWR formations with $n$ agents in $\mathbb{R}^d$.  
In this subsection, we assume that the desired formation is minimally GIWR, which is relaxed in the next subsection.

\begin{theorem}\label{Thm:exp_min}
Suppose that the desired formation is minimally GIWR. If any initial formation is close to the desired formation, then the error system (\ref{error_dym}) has an exponentially stable equilibrium at the origin, and the initial formation locally exponentially converges to the desired formation shape.
\end{theorem}
\begin{proof}
We first define the potential function $V(e)$ as $V(e)=\frac{1}{2}e^\top e$ 
which is also the Lyapunov function candidate. We also define a sub-level set $\Psi$ as $\Psi=\set{e \given V(e) \leq \epsilon}$ for $\epsilon>0$ such that all formations in the set $\Psi$ are minimally GIWR close to the desired formation. 

With the equation (\ref{error_dym}), the derivative of $V(e)$ along a trajectory of ${e}$ is calculated as 
\begin{equation}
\dot{V}(e) = e^\top \dot{e} 
=-e^\top R_W(e) R_W^\top(e) e
= - \norm{R_W^\top(e) e}^2. \label{eq:dot_potential_fn}
\end{equation}
Since the formation in the set $\Psi$  is minimally GIWR, the weak rigidity matrix has the full row rank. Therefore, since $\rank\left(R_W(e) R_W^\top(e) \right) = \rank\left(R_W(e)\right) $, $R_W(e) R_W^\top(e) \in \mathbb{R}^{(m+w) \times (m+w)}$ is of full rank and $R_w(e) R_w^\top(e)$ is positive definite (all eigenvalues of $R_w(e) R_w^\top(e)$ are positive). Moreover, this implies
\begin{equation}
\dot{V}(e) \leq  -\lambda \norm{e}^2,  \label{eq:V_leq_fn01}
\end{equation}
where $\lambda$ denotes the minimum eigenvalue of $R_w(e) R_w^\top(e)$.
The inequality (\ref{eq:V_leq_fn01}) indicates that $\dot{V} < 0$ for $e \in \Psi \setminus \{0\}$. Thus, the origin of the error system (\ref{error_dym}) is asymptotically stable near the desired formation.
Also, since $V=\frac{1}{2}e^\top e$, the following inequality holds.
\begin{equation}
\dot{V}(e) \leq  -2\lambda V(e),  \label{eq:V_leq_fn02}
\end{equation}
and it follows that $V(e(t))\leq V(e(0))\text{exp}(-2\lambda t)$ by Gronwall-Bellman Inequality \cite[Lemma A.1]{khalil2002nonlinear}. Therefore, the error system (\ref{error_dym}) has an exponentially stable equilibrium at the origin, and the solution of (\ref{control_law01}) exists and is finite as $t \rightarrow \infty$. By the above result, the control law (\ref{control_law01}) guarantees that $p$ exponentially converges to a fixed point. The initial formation in the set $\Psi$ is close to the desired formation. Hence, the initial formation locally exponentially converges to the desired formation shape.
\end{proof}

\subsection{Stability on non-minimally GIWR formations with $n$ agents in $\mathbb{R}^d$} \label{Subsec:n_stability}
In this subsection, we explore the stability in the case of non-minimally GIWR formation systems with $n$ agents in $\mathbb{R}^d$. To this end, we make use of a linearization approach of perturbed systems motivated by \cite{sun2016exponential,mou2016undirected}. 

We denote a minimally GIWR sub-framework induced from $(\mathcal{G},\mathcal{A},p)$ by $(\bar{\mathcal{G}},\bar{\mathcal{A}},p)$, where $\bar{\mathcal{G}}=(\mathcal{V},\bar{\mathcal{E}})$. We also denote the remaining part of $(\mathcal{G},\mathcal{A},p)$ except $(\bar{\mathcal{G}},\bar{\mathcal{A}},p)$ by $(\tilde{\mathcal{G}},\tilde{\mathcal{A}},p)$, where $\tilde{\mathcal{G}}=(\mathcal{V},\tilde{\mathcal{E}})$, $ \tilde{\mathcal{E}} = \mathcal{E} \setminus \bar{\mathcal{E}}$ and $ \tilde{\mathcal{A}} = \mathcal{A} \setminus \bar{\mathcal{A}} $ (See an example in Fig.~\ref{Fig:non_min_ex}). Let $\sigma$ denote the sum of cardinalities of edges and angles, i.e., $\sigma=m+w$. Then, $\bar{\sigma}$ and $\tilde{\sigma}$ are defined as $\bar{\sigma}= \card{\bar{\mathcal{E}}}+\card{\bar{\mathcal{A}}}=\bar{m}+\bar{w}=dn-d(d+1)/2$ (or $dn-(d^2+d+2)/2$ when $\mathcal{E}=\emptyset$) and $\tilde{\sigma}= \card{\tilde{\mathcal{E}}}+\card{\tilde{\mathcal{A}}}=\tilde{m}+\tilde{w}=\sigma-\bar{\sigma}$, respectively.
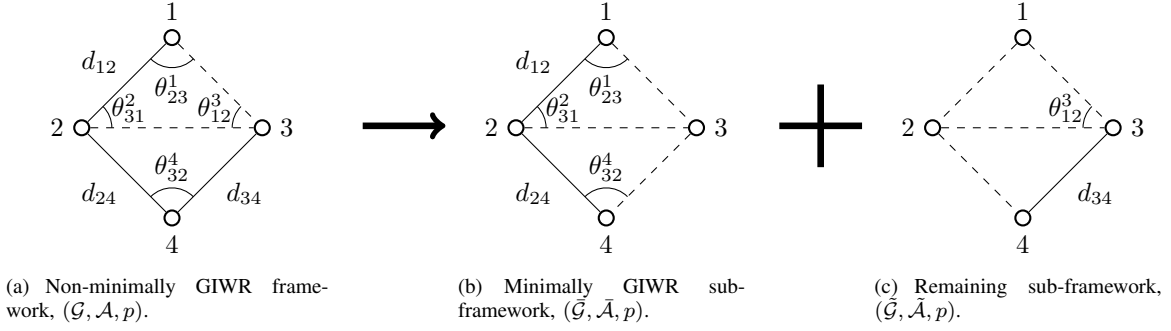
\begin{figure}[]
\centering
\subfigure[Non-minimally GIWR framework, $(\mathcal{G},\mathcal{A},p)$.]{\label{Fig:non_min_ex01}
\quad\begin{tikzpicture}[scale=0.6]
\node[place] (node1) at (0,2) [label=above:$1$] {};
\node[place] (node2) at (-2,0) [label=left:$2$] {};
\node[place] (node3) at (2,0) [label=right:$3$] {};
\node[place] (node4) at (0,-2) [label=below:$4$] {};

\draw[lineUD] (node1)  -- node [above left] {$d_{12}$} (node2);
\draw[dashed] (node1)  -- (node3);
\draw[dashed] (node2)  -- (node3);
\draw[lineUD] (node2)  -- node [below left] {$d_{24}$} (node4);
\draw[lineUD] (node3)  -- node [below right] {$d_{34}$} (node4);

\pic [draw, -, "${\theta}^1_{23}$", angle eccentricity=1.7, angle radius=0.4cm] {angle = node2--node1--node3};
\pic [draw, -, "${\theta}^2_{31}$", angle eccentricity=1.7, angle radius=0.4cm] {angle = node3--node2--node1};
\pic [draw, -, "${\theta}^3_{12}$", angle eccentricity=1.7, angle radius=0.4cm] {angle = node1--node3--node2};
\pic [draw, -, "${\theta}^4_{32}$", angle eccentricity=1.7, angle radius=0.4cm] {angle = node3--node4--node2};
\end{tikzpicture}\quad%
} 
{
\begin{tikzpicture}
\draw[draw=black, line width=2.5pt, ->] (0.2,1.8) -- (1.3,1.8){};
\draw[] (0,0){};
\end{tikzpicture}%
}  
\subfigure[Minimally GIWR sub-framework, $(\bar{\mathcal{G}},\bar{\mathcal{A}},p)$.]{\label{Fig:non_min_ex02}
\,\,\begin{tikzpicture}[scale=0.6]
\node[place] (node1) at (0,2) [label=above:$1$] {};
\node[place] (node2) at (-2,0) [label=left:$2$] {};
\node[place] (node3) at (2,0) [label=right:$3$] {};
\node[place] (node4) at (0,-2) [label=below:$4$] {};

\draw[lineUD] (node1)  -- node [above left] {$d_{12}$} (node2);
\draw[dashed] (node1)  -- (node3);
\draw[dashed] (node2)  -- (node3);
\draw[lineUD] (node2)  -- node [below left] {$d_{24}$} (node4);
\draw[dashed] (node3)  -- (node4);

\pic [draw, -, "${\theta}^1_{23}$", angle eccentricity=1.7, angle radius=0.4cm] {angle = node2--node1--node3};
\pic [draw, -, "${\theta}^2_{31}$", angle eccentricity=1.7, angle radius=0.4cm] {angle = node3--node2--node1};
\pic [draw, -, "${\theta}^4_{32}$", angle eccentricity=1.7, angle radius=0.4cm] {angle = node3--node4--node2};
\end{tikzpicture}\,\,%
}
{
\begin{tikzpicture}
\draw[draw=black, line width=2.5pt] (0.2,1.8) -- (1.3,1.8){};
\draw[draw=black, line width=2.5pt] (0.75,1.25) -- (0.75,2.35){};
\draw[] (0,0){};
\end{tikzpicture}%
}  
\subfigure[Remaining sub-framework, $(\tilde{\mathcal{G}},\tilde{\mathcal{A}},p)$.]{\label{Fig:non_min_ex03}
\,\,\begin{tikzpicture}[scale=0.6]
\node[place] (node1) at (0,2) [label=above:$1$] {};
\node[place] (node2) at (-2,0) [label=left:$2$] {};
\node[place] (node3) at (2,0) [label=right:$3$] {};
\node[place] (node4) at (0,-2) [label=below:$4$] {};

\draw[dashed] (node1)  -- (node2);
\draw[dashed] (node1)  -- (node3);
\draw[dashed] (node2)  -- (node3);
\draw[dashed] (node2)  -- (node4);
\draw[lineUD] (node3)  -- node [below right] {$d_{34}$} (node4);

\pic [draw, -, "${\theta}^3_{12}$", angle eccentricity=1.7, angle radius=0.4cm] {angle = node1--node3--node2};
\end{tikzpicture}\,\,%
} 
\caption{Example of framework decomposition of a non-minimally GIWR framework. The dashed lines indicate virtual edges which do not belong to $\mathcal{E}$, $\bar{\mathcal{E}}$ and $\tilde{\mathcal{E}}$. Distance and angle constraints are denoted by $d_{ij}, (i,j) \in \mathcal{E}$ and $\theta_{ij}^{k}, (k,i,j) \in \mathcal{A}$, respectively.} \label{Fig:non_min_ex} 
\end{figure}
Moreover, we denote the sub-vector $\bar{e}\in\mathbb{R}^{\bar{\sigma}}$ whose entries are those entries in $e$ corresponding to edges and angles in  $(\bar{\mathcal{G}},\bar{\mathcal{A}},p)$, and $\tilde{e}\in\mathbb{R}^{\tilde{\sigma}}$ whose entries are those entries in $e$ corresponding to edges and angles in  $(\tilde{\mathcal{G}},\tilde{\mathcal{A}},p)$.
We denote the permutation matrix $\mbf{P} = \begin{bmatrix}\bar{\mbf{P}}^\top & \tilde{\mbf{P}}^\top\end{bmatrix}$ such that 
$\begin{bmatrix} 
\bar{e}& \tilde{e}
\end{bmatrix}^\top=\mbf{P}^\top e$ or equivalently $\bar{e}=\bar{\mbf{P}}e$ and $\tilde{e}=\tilde{\mbf{P}}e$, where $\mbf{P}\in\mathbb{R}^{\sigma\times\sigma}$, $\bar{\mbf{P}}\in\mathbb{R}^{\bar{\sigma}\times\sigma}$ and $\tilde{\mbf{P}}\in\mathbb{R}^{\tilde{\sigma}\times\sigma}$. The permutation matrix has properties such that $\bar{\mbf{P}}\bar{\mbf{P}}^\top=I_{\bar{\sigma}\times\bar{\sigma}}$, $\tilde{\mbf{P}}\tilde{\mbf{P}}^\top=I_{\tilde{\sigma}\times\tilde{\sigma}}$, $\bar{\mbf{P}}\tilde{\mbf{P}}^\top=0_{\bar{\sigma}\times\tilde{\sigma}}$, $\bar{\mbf{P}}^\top\bar{\mbf{P}}+\tilde{\mbf{P}}^\top\tilde{\mbf{P}}=I_{\sigma \times \sigma}$ and $e=\bar{\mbf{P}}^\top\bar{e}+\tilde{\mbf{P}}^\top\tilde{e}$. We now show that $\tilde{e}$ is a function of $\bar{e}$ locally.
\begin{lemma}
\label{Lem:exist_smf}
Let a framework $(\mathcal{G},\mathcal{A},q)$ be the desired formation, and non-minimally GIWR.
Then, there (locally) exists a smooth function $f: \bar{e}(q) \rightarrow \mathbb{R}^{(\sigma-\bar{\sigma})}$ such that $\tilde{e}(q)=f(\bar{e}(q))$ close to $(\mathcal{G},\mathcal{A},q)$. Furthermore, it holds that $f(\bar{e})=0$ if and only if $\bar{e}=0$.
\end{lemma}
\begin{proof}
This proof is motivated by Proposition 1 in Mou et al. (2016)\cite{mou2016undirected}. 
\quad (i) For the 2-dimensional case,  we first
denote a rotation matrix $S(\mbf{x})$ such that $S(\mbf{x})=\frac{1}{\norm{\mbf{x}}}
\begin{bmatrix} 
x_2 &-x_1\\
x_1 &x_2
\end{bmatrix}$ for a nonzero vector $\mbf{x} = \begin{bmatrix}x_1 & x_2\end{bmatrix}^\top\in\mathbb{R}^{2}$. The equality $S(\mbf{x})\mbf{x}=\begin{bmatrix}0 & \norm{\mbf{x}}\end{bmatrix}^\top$ always holds. We denote a vector $\varsigma: p\rightarrow\mathbb{R}^{\bar{\sigma}}$ with $\bar{\sigma}=2n-3$ when $\mathcal{E} \neq \emptyset$ in $\mathbb{R}^2$ such as:
\begin{equation}
\varsigma(p)=
\begin{bmatrix}
\norm{ z_{21}  } & (S(z_{21}) z_{31}    )^\top &\cdots & ( S(z_{21})  z_{n1}    )^\top
\end{bmatrix}^\top.
\end{equation}
Since the rotation matrix does not change a magnitude of a vector, we see that $\norm{  z_{j1}    }^2=\norm{S(z_{21}) z_{j1} }^2$ and  $\norm{z_{ij}}^2=\norm{S(z_{21})  z_{i1}    -S(z_{21})   z_{j1}   }^2$, and further $\varsigma(p)$ includes all information on the relative vectors $z_{21},z_{31},\dots,z_{n1}$. 
Thus, any entry in $\tilde{e}$ is composed of entries in $\varsigma(p)$. Moreover, there exists a function $\tilde{f}_e: \mathbb{R}^{\bar{\sigma}} \rightarrow \mathbb{R}^{(\sigma-\bar{\sigma})}$ such that $\tilde{e}=\tilde{f}_e(\varsigma(p))$. Similarly, there exists a function $\bar{f}_e: \mathbb{R}^{\bar{\sigma}} \rightarrow\mathbb{R}^{\bar{\sigma}}$ such that $\bar{e}=\bar{f}_e(\varsigma(p))$.

In the same way, for the case of $\mathcal{E} = \emptyset$ in $\mathbb{R}^2$, we can define a vector $\varsigma: p\rightarrow\mathbb{R}^{\bar{\sigma}}$ with $\bar{\sigma}=2n-4$ such that
\begin{align}
\varsigma(p)
=\begin{bmatrix}
(S(z_{21})z_{31})^\top & (S(z_{21})z_{41})^\top & \cdots & (S(z_{21})z_{n1})^\top
\end{bmatrix}^\top.
\end{align}
Then, with the fact in Lemma \ref{Lem:poly_sub01} in Appendix, it is obvious that there exist $\tilde{e}=\tilde{f}_e(\varsigma(p))$ and $\bar{e}=\bar{f}_e(\varsigma(p))$.

The derivative of $\bar{e}$ at $q$, i.e., $\left.\frac{\partial \bar{e}(p)}{\partial p} \right|_{p=q}$ is the weak rigidity matrix of $(\bar{\mathcal{G}},\bar{\mathcal{A}},q)$. Then, $\rank \left(\left.\frac{\partial \bar{e}(p)}{\partial p} \right|_{p=q}\right)=\bar{\sigma}$ since $(\bar{\mathcal{G}},\bar{\mathcal{A}},q)$ is minimally GIWR. Thus, with $\left.\frac{\partial \bar{e}(p)}{\partial p} \right|_{p=q}=\left.\frac{\bar{f}_e(\varsigma(p))}{\partial \varsigma(p)} \frac{\partial \varsigma(p)}{\partial p} \right|_{p=q}$ from $\bar{e}=\bar{f}_e(\varsigma(p))$, it holds that $\rank \left(\left.\frac{\bar{f}_e(\varsigma(p))}{\partial \varsigma(p)} \right|_{p=q}\right)\geq\bar{\sigma}$ by the rank property.
Since $\left.\frac{\bar{f}_e(\varsigma(p))}{\partial \varsigma(p)} \right|_{p=q}$ is an $\bar{\sigma}\times\bar{\sigma}$ matrix, we can see that $\left.\frac{\bar{f}_e(\varsigma(p))}{\partial \varsigma(p)} \right|_{p=q}$ is of full rank and $\left.\frac{\bar{f}_e(\varsigma(p))}{\partial \varsigma(p)} \right|_{p=q}$ is nonsingular. Hence, from the inverse function theorem, there is an open set $\mathcal{W}\subset\mathbb{R}^{\bar{\sigma}}$ containing $\varsigma(q)$ such that $\bar{f}_e$ has a smooth inverse $\bar{f}_e^{-1}: \bar{f}_e(\mathcal{W}) \rightarrow \mathcal{W}$. Then, the following equality holds.
\begin{equation}
\bar{f}_e^{-1}(\bar{f}_e(\varsigma(p)))=\varsigma(p), \varsigma(p)\in\mathcal{W},
\end{equation}
which implies that $\bar{f}_e^{-1}(\bar{f}_e(\varsigma(p)))=\bar{f}_e^{-1}(\bar{e})=\varsigma(p)$.
Since $\tilde{e}=\tilde{f}_e(\varsigma(p))$, the equality $\tilde{e}=\tilde{f}_e(\bar{f}_e^{-1}(\bar{e}))$ holds. Therefore, we can say that there exists a smooth function $f: \bar{e}(q) \rightarrow \mathbb{R}^{(\sigma-\bar{\sigma})}$ such that $\tilde{e}(q)=f(\bar{e}(q))$ close to $(\mathcal{G},\mathcal{A},q)$.
In addition, since $\tilde{\mbf{P}}e=\tilde{e}=\tilde{f}_e\left(\bar{f}_e^{-1}(\bar{e})\right)=\tilde{f}_e\left(\bar{f}_e^{-1}\left(\bar{\mbf{P}}e\right)\right)=f\left(\bar{\mbf{P}}e\right)$ and $e=0$ at the desired formation $(\mathcal{G},\mathcal{A},q)$, it holds that $f(0)=0$.

\quad (ii)
For the 3-dimensional case, let us consider rotation matrices $S_{x_1}(\mbf{x})$ and $S_{x_2}(\mbf{x})$ rotating a vector $\mbf{x} = \begin{bmatrix}x_1 & x_2 & x_3\end{bmatrix}^\top\in\mathbb{R}^{3}$ about $x_1$ and $x_2$ axes into $x_1x_3$-plane and $x_1x_2$-plane, respectively, as follows
\begin{flalign} \label{eq_tttemp}
S_{x_1}(\mbf{x})=
	\begin{bmatrix} 
	1 & 0 & 0\\
	0 & \frac{x_3}{\sqrt{x_2^2+x_3^2}} &-\frac{x_2}{\sqrt{x_2^2+x_3^2}}\\
	0 & \frac{x_2}{\sqrt{x_2^2+x_3^2}} &\frac{x_3}{\sqrt{x_2^2+x_3^2}}
	\end{bmatrix},\,
S_{x_2}(\mbf{x})=
	\begin{bmatrix} 
	\frac{x_1}{\sqrt{x_1^2+x_3^2}} & 0 & \frac{x_3}{\sqrt{x_1^2+x_3^2}}\\
	0 & 1 &0\\
	-\frac{x_3}{\sqrt{x_1^2+x_3^2}} & 0 &\frac{x_1}{\sqrt{x_1^2+x_3^2}}
	\end{bmatrix}.
\end{flalign}
We then have
\begin{align}
S_{x_1}(z_{21})z_{21}&=\begin{bmatrix}z_{21}^{(1)} & 0 & \frac{\left(z_{21}^{(2)}\right)^2+\left(z_{21}^{(3)}\right)^2}{\sqrt{\left(z_{21}^{(2)}\right)^2+\left(z_{21}^{(3)}\right)^2}}\end{bmatrix}^\top =\bar{z}_{21}, \label{eq_ttemp01} \\
S_{x_2}(\bar{z}_{21}) \bar{z}_{21}&=\begin{bmatrix}\norm{z_{21}} & 0 & 0\end{bmatrix}^\top =\check{\bar{z}}_{21}, \label{eq_ttemp02}
\end{align}
where $z_{21}= \begin{bmatrix}z_{21}^{(1)} & z_{21}^{(2)} & z_{21}^{(3)}\end{bmatrix}^\top\in\mathbb{R}^{3}$.
We also have
\begin{align}
S_{x_1}(\check{\bar{z}}_{31}) \check{\bar{z}}_{21}&=\begin{bmatrix}\norm{z_{21}} & 0 & 0\end{bmatrix}^\top, \label{eq_ttemp03} \\
S_{x_1}(\check{\bar{z}}_{31}) \check{\bar{z}}_{31}&=\begin{bmatrix}\check{\bar{z}}_{31}^{(1)}  & 0 & \frac{\left(\check{\bar{z}}_{31}^{(2)} \right)^2+\left(\check{\bar{z}}_{31}^{(3)} \right)^2}{\sqrt{\left(\check{\bar{z}}_{31}^{(2)} \right)^2+\left(\check{\bar{z}}_{31}^{(3)} \right)^2}}\end{bmatrix}^\top, \label{eq_ttemp04}
\end{align} 
where $\check{\bar{z}}_{31}=S_{x_2}(\bar{z}_{21}) S_{x_1}(z_{21}) z_{31}=\begin{bmatrix}\check{\bar{z}}_{31}^{(1)} & \check{\bar{z}}_{31}^{(2)}  & \check{\bar{z}}_{31}^{(3)} \end{bmatrix}^\top\in\mathbb{R}^{3}$.
With the facts of \eqref{eq_ttemp03} and \eqref{eq_ttemp04},
we can denote a vector $\varsigma: p\rightarrow\mathbb{R}^{\bar{\sigma}}$ with $\bar{\sigma}=3n-6$ when $\mathcal{E} \neq \emptyset$ in $\mathbb{R}^3$ such that
\begin{align}
\varsigma(p)= 
\begin{bmatrix}
\norm{z_{21}}&
 \check{\bar{z}}_{31}^{(1)} &
\frac{\left(\check{\bar{z}}_{31}^{(2)} \right)^2+\left(\check{\bar{z}}_{31}^{(3)} \right)^2}{\sqrt{\left(\check{\bar{z}}_{31}^{(2)} \right)^2+\left(\check{\bar{z}}_{31}^{(3)} \right)^2}}  &
 \left(\bar{S} z_{41}\right)^\top&
 \cdots &
 \left(\bar{S} z_{n1}\right)^\top
\end{bmatrix}^\top,
\end{align}where $\bar{S}= S_{x_1}(\check{\bar{z}}_{31}) S_{x_2}(\bar{z}_{21}) S_{x_1}(z_{21})$. $\varsigma(p)$ includes all information on the relative vectors $z_{21},z_{31},\dots,z_{n1}$. Since the rotation matrices do not change a magnitude of a vector, any entry in $\tilde{e}$ is a function composed of entries in $\varsigma(p)$, and further there exists a function $\tilde{f}_e: \mathbb{R}^{\bar{\sigma}} \rightarrow \mathbb{R}^{(\sigma-\bar{\sigma})}$ such that $\tilde{e}=\tilde{f}_e(\varsigma(p))$. Moreover, there exists a function $\bar{f}_e: \mathbb{R}^{\bar{\sigma}} \rightarrow\mathbb{R}^{\bar{\sigma}}$ such that $\bar{e}=\bar{f}_e(\varsigma(p))$.
In the same manner, for the case of $\mathcal{E} = \emptyset$ in $\mathbb{R}^3$, we can define a vector $\varsigma: p\rightarrow\mathbb{R}^{\bar{\sigma}}$ with $\bar{\sigma}=3n-7$ such that
\begin{equation}
\varsigma(p)=
\begin{bmatrix}
 \check{\bar{z}}_{31}^{(1)} &
 \frac{\left(\check{\bar{z}}_{31}^{(2)} \right)^2+\left(\check{\bar{z}}_{31}^{(3)} \right)^2}{\sqrt{\left(\check{\bar{z}}_{31}^{(2)} \right)^2+\left(\check{\bar{z}}_{31}^{(3)} \right)^2}}  &
 \left(\bar{S} z_{41}\right)^\top&
 \cdots &
 \left(\bar{S} z_{n1}\right)^\top
\end{bmatrix}^\top.
\end{equation}Then, with the fact in Lemma \ref{Lem:poly_sub01} in Appendix, there exist $\tilde{e}=\tilde{f}_e(\varsigma(p))$ and $\bar{e}=\bar{f}_e(\varsigma(p))$.
The rest of this proof is proved in the same way as  the 2-dimensional case.
\end{proof}

We denote $\bar{R}_W \in \mathbb{R}^{\bar{\sigma}\times dn}$ as the weak rigidity matrix for the sub-framework $(\bar{\mathcal{G}},\bar{\mathcal{A}},p)$, and $\tilde{R}_W \in \mathbb{R}^{\tilde{\sigma}\times dn}$ as the weak rigidity matrix for the sub-framework $(\tilde{\mathcal{G}},\tilde{\mathcal{A}},p)$. Then, it holds that $\bar{R}_W=\bar{\mbf{P}}R_{W}$ and $\tilde{R}_W=\tilde{\mbf{P}}R_{W}$.
From the fact that $\bar{e}=\bar{\mbf{P}}e$ and $e=\bar{\mbf{P}}^\top\bar{e}+\tilde{\mbf{P}}^\top\tilde{e}$, we have
\begin{align}
\dot{\bar{e}}=\bar{\mbf{P}}\dot{e}=\bar{\mbf{P}}\frac{\partial e}{\partial p}\dot{p} & = -\bar{\mbf{P}}R_{W}R_{W}^\top e \nonumber \\
& = -\bar{\mbf{P}}R_{W}R_{W}^\top (\bar{\mbf{P}}^\top\bar{e}+\tilde{\mbf{P}}^\top\tilde{e}) \nonumber \\
& = -\bar{R}_W\bar{R}_W^\top\bar{e} -\bar{R}_W\tilde{R}_W^\top\tilde{e}. \label{perturb_eq}
\end{align}
From Lemma \ref{Lem:exist_smf}, the equality \eqref{perturb_eq} can be rewritten as
\begin{equation}
\dot{\bar{e}}=-\bar{R}_W\bar{R}_W^\top\bar{e} -\bar{R}_W\tilde{R}_W^\top f(\bar{e}), \label{perturb_eq02}
\end{equation}
which locally holds only around the desired formation. 
It also holds that $\tilde{R}_W=\frac{\partial \tilde{e}}{\partial p}=\frac{\partial \tilde{e}}{\partial \bar{e}}\frac{\partial \bar{e}}{\partial p}=\frac{\partial f(\bar{e})}{\partial \bar{e}}\frac{\partial \bar{e}}{\partial p}=\frac{\partial f(\bar{e})}{\partial \bar{e}}\bar{R}_W$.
Therefore, we can consider the error system \eqref{perturb_eq02} as a perturbed system. We then reach the following theorem.

\begin{theorem}\label{Thm:stability_non_minimal}
Under the gradient flow law \eqref{control_law01}, the perturbed error system \eqref{perturb_eq02}  for the non-minimally GIWR formation has an exponential stable equilibrium at the origin.
\end{theorem}
\begin{proof}
Note that $\tilde{R}_W=\frac{\partial \tilde{e}}{\partial p}=\frac{\partial f}{\partial \bar{e}}\frac{\partial \bar{e}}{\partial p}=F\bar{R}_W$, where $F = \frac{\partial f}{\partial \bar{e}}$
$g_k(e_k)$. We define a neighborhood set $\Psi$ around $\bar{e}=0$ as $\Psi=\set{\bar{e}\in\mathbb{R}^{\bar{\sigma}}\given \norm{\bar{e}}^2<\epsilon}$ for $\epsilon>0$. Then, the remainder of this proof is similar to Theorem 3 in Sun et al. (2016)\cite{sun2016exponential}. 
\end{proof}
\section{Application to formation control: almost global convergence of $3$-agent formations in $\mathbb{R}^2$}
\label{Sec:Formation_control_almost}
This section aims to provide analysis for almost global stability on special cases of minimally GIWR $3$-agent formations in $\mathbb{R}^2$.
In this section, we also use the control system \eqref{control_law01} as discussed in the subsection \ref{subsec:eq_motion}.
We first classify all equilibrium points to explore the stability of the system \eqref{control_law01}
with a set $\mathcal{P}$ composed of all equilibrium points defined as $\mathcal{P} = \set{p \in \mathbb{R}^{2n} \given R_{W}^\top e=0} $ as follows.
\begin{align}
\mathcal{P}^* &= \set{p \in \mathbb{R}^{2n} \given e=0}, \\
\mathcal{P}_i &= \set{p \in \mathbb{R}^{2n} \given R_{W}^\top e=0, e\neq0},
\end{align} 
where $\mathcal{P}^*$ and $\mathcal{P}_i$ denote the sets for desired equilibria and incorrect equilibria, respectively.
Both of $\mathcal{P}^*$ and $\mathcal{P}_i$ constitute the set of all equilibria, i.e., $\mathcal{P}=\mathcal{P}^* \cup \mathcal{P}_i$.
An equilibrium point $\bar{p}=[\bar{p}^\top_{1},...,\bar{p}^\top_{n}]^\top \in \mathbb{R}^{2n}$ is called \myemph{incorrect equilibrium point} if $\bar{p}$ belongs to $\mathcal{P}_i$.

\subsection{Analysis of the incorrect equilibria} \label{Subsec:analysis_incorrect}
We show in this subsection that the system (\ref{control_law01}) at any incorrect equilibrium point $\bar{p}$ is unstable. We first explore what cases occur at the incorrect equilibria.

\begin{lemma}\label{incorrect_collinear}
In the case of three-agent formations, incorrect equilibria take place only when the three agents are collinear.
\end{lemma}
\begin{proof}
From the viewpoint of a minimally GIWR formation composed of three agents, there are only three formation cases: the first one is a formation with one angle constraint and two distance constraints; the second one is that with two angle constraints and one distance constraint; the third one is that with only two angle constraints. Each example for the three cases is illustrated in Fig.~\ref{Fig:diff_figs_a}, Fig.~\ref{Fig:diff_figs_b} and Fig.~\ref{Fig:diff_figs_c}, respectively.

Let $\mathcal{N'}_l$ denote a set of neighbors of vertex $l$ by $\mathcal{N'}_l = \set{i,j \in \mathcal{V} \given (l,i) \in \mathcal{E} \lor (l,i,j) \in \mathcal{A}}$. If a framework $(\mathcal{G},\mathcal{A},p)$ with $n=3$ vertices is minimally GIWR, then each agent has exactly two neighbors, i.e., $\card{\mathcal{N'}_l}=2$. 
In the weak rigidity matrix $R_W$, all elements are composed of inter-neighbor relative position vectors, i.e., $\frac{\partial F_W}{\partial p_l}$ consists of ${z'}_{lk_1}^\top$ and ${z'}_{lk_2}^\top$ for $k_1,k_2\in\mathcal{N'}_l$. 
Thus, at the incorrect equilibria, the following form holds:
\begin{align}
{z'}_{lk_1}^\top= c_l {z'}_{lk_2}^\top, \, k_1,k_2\in\mathcal{N'}_l,
\end{align}
where $c_l \in\mathbb{R}$ is a coefficient.
This implies that incorrect equilibria take place only when the three agents are collinear for 3-agent formations in $\mathbb{R}^2$. 

We next show an example with a formation in Fig.~\ref{Fig:diff_figs_a}.
For the case of the formation with one angle constraint and two distance constraints as shown in Fig.~\ref{Fig:diff_figs_a}, the equation (\ref{control_law01}) can be written as
\begin{subequations}
\begin{align}
\dot{p}_1 &= -2z_{12}e_{12}-2z_{13}e_{13}-\alpha^\top e^1_{23}, \label{collinear_01} \\ 
\dot{p}_2 &= 2z_{12}e_{12}-\beta^\top e^1_{23}, \label{collinear_02} \\ 
\dot{p}_3 &= 2z_{13}e_{13}-\gamma^\top e^1_{23}, \label{collinear_03}
\end{align}
\end{subequations} where $e_{ij}=\norm{z_{g_{ij}}}^2-\norm{z_{g_{ij}}^*}^2$, $(i,j) \in \mathcal{E}$, $e^1_{23}=A_{h_{123}}-A_{h_{123}}^*$, $\alpha=\frac{\partial}{\partial p_1}\cos\theta^1_{23}$, $\beta=\frac{\partial}{\partial p_2}\cos\theta^1_{23}$ and $\gamma=\frac{\partial}{\partial p_3}\cos\theta^1_{23}$. 
In the incorrect equilibrium set $\mathcal{P}_i$, the equation (\ref{collinear_03}) is calculated as
\begin{equation}
z_{12}=\left.\left(\frac{\norm{z_{12}}}{\norm{z_{13}}}\cos\theta^1_{23}-2\norm{z_{12}}\norm{z_{13}}\frac{e_{13}}{e^1_{23}} \right) z_{13} \right|_{p \in \mathcal{P}_i} \label{Eq:collinear_04}
\end{equation}
It follows from \eqref{Eq:collinear_04} that $p_1$, $p_2$ and $p_3$ must be collinear.
The equations (\ref{collinear_01}) and (\ref{collinear_02}) also give us similar results. Therefore, the three agents must be collinear.
The formation shape of the three agents falls into one of three cases as depicted in Fig.~\ref{Fig:formation_forms}.
Two cases illustrated in Fig.~\ref{Fig:diff_figs_b} and Fig.~\ref{Fig:diff_figs_c} also give us similar results to the case of Fig.~\ref{Fig:diff_figs_a}.
\end{proof}

\begin{figure}[]
\centering
\subfigure[]{ 
\begin{tikzpicture}[scale=1]
\node[place] (node2) at (-1,0) [label=above:$2$] {};
\node[place] (node1) at (0,0) [label=above:$1$] {};
\node[place] (node3) at (1,0) [label=above:$3$] {};

\draw[lineUD] (node1)  -- node {}(node2);
\draw[lineUD] (node1)  -- node {}(node3);
\end{tikzpicture}
} \qquad \qquad
\subfigure[]{ 
\begin{tikzpicture}[scale=1]
\node[place] (node1) at (-1,0) [label=above:$1$] {};
\node[place] (node2) at (0,0) [label=above:$2$] {};
\node[place] (node3) at (1,0) [label=above:$3$] {};

\draw[lineUD] (node1)  -- node {}(node2);
\draw[lineUD] (node2)  -- node {}(node3);
\end{tikzpicture}
}\qquad \qquad
\subfigure[]{ 
\begin{tikzpicture}[scale=1]
\node[place] (node1) at (-1,0) [label=above:$1$] {};
\node[place] (node3) at (0,0) [label=above:$3$] {};
\node[place] (node2) at (1,0) [label=above:$2$] {};

\draw[lineUD] (node1)  -- node {}(node3);
\draw[lineUD] (node2)  -- node {}(node3);
\end{tikzpicture}
} 
\caption{Three formation forms which can occur at the incorrect equilibria.} \label{Fig:formation_forms}
\end{figure}
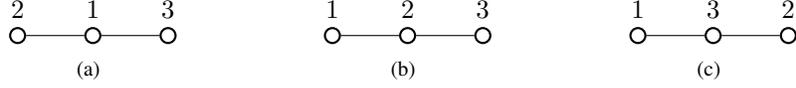

Next, to analyze the stability at the incorrect equilibria, we linearize the system (\ref{control_law01}). One can observe the following negative Jacobian $J(p)$ of the system (\ref{control_law01}) with respect to $p$:
\begin{align}
J(p)=-\frac{\partial}{\partial p}\dot{p} 
=R_{W}(p)^\top R_{W}(p)+E(p)\otimes I_2
+ \sum_{(k,i,j) \in \mathcal{A}} e_{A_h}\bigg( (I_3\otimes p_1) \frac{\partial}{\partial p} C_1  + (I_3\otimes p_2) \frac{\partial}{\partial p} C_2
+ (I_3\otimes p_3) \frac{\partial}{\partial p} C_3\bigg), \label{negative_jaco}
\end{align}
where $p=\begin{bmatrix}p_{1}^\top & p_{2}^\top & p_{3}^\top\end{bmatrix}^\top \in \mathbb{R}^{6}$, $e_{A_h}=A_{h_{kij}}-A_{h_{kij}}^*$, and $C_l \in \mathbb{R}^3$ for $l \in \{1,2,3\}$ denotes a vector composed of entries of $l$-th column associated with $e_{A_h}$ in $E(p)$ (See an example (17) in Kwon et al. (2018)\cite{kwon2018infinitesimal}).
If $J(p)$ has at least one negative eigenvalue at the incorrect equilibrium point $\bar{p}$, then the system at $\bar{p}$  is unstable.
In order to check this fact, we first reorder columns of $J(p)$, which does not have an effect on any eigenvalue of $J(p)$.
We make use of a permutation matrix $T$ which reorders columns of matrix such that
\begin{align}
R_WT&=\begin{bmatrix}
R_x & R_y
\end{bmatrix}=\bar{R}, \nonumber \\
P_lT&=\begin{bmatrix}
P_{lx} & P_{ly}
\end{bmatrix}=\bar{P}_l, \nonumber \\
\frac{\partial}{\partial p} C_lT&=\begin{bmatrix}
C_{lx} & C_{ly}
\end{bmatrix}=\bar{C}_l, \label{Permu_reorder}
\end{align}
where $P_l= (I_3\otimes p_l^\top)\in \mathbb{R}^{3 \times 6}$ for $l \in \{1,2,3\}$.
In \eqref{Permu_reorder}, 
$R_u \in \mathbb{R}^{\sigma\times 3}$, $P_{lu} \in \mathbb{R}^{3 \times 3}$ and $C_{lu} \in \mathbb{R}^{3 \times 3}$ for $u = x,y$ denote matrices whose columns are composed of the columns of coordinate $u$ in the matrix $R_W$, $P_l$ and $\frac{\partial}{\partial p}C_l$, respectively. The formation is minimally GIWR, thus $\sigma=3$. It is remarkable that $TT^\top=I$ holds since $T$ is a permutation matrix. 
With the permutation matrix $T$, the permutated matrix $\bar{J}(p)$ is given by
\begin{align}
\bar{J}(p) = T^\top J(p) T   
=& \bar{R}^\top \bar{R} + I_2\otimes E(p)
+ \sum_{(k,i,j) \in \mathcal{A}} \left( \bar{P}_1^\top \bar{C}_1+\bar{P}_2^\top \bar{C}_2+\bar{P}_3^\top \bar{C}_3 \right)e_{A_h} \nonumber \\
=& \begin{bmatrix}
\bar{J}_{11} & \bar{J}_{12} \\
\bar{J}_{21} & \bar{J}_{22}
\end{bmatrix},
\end{align}
where
\begin{align}
\bar{J}_{11} =& R_x^\top R_x+E(p) +\sum_{(k,i,j) \in \mathcal{A}} (P_{1x}C_{1x}+P_{2x}C_{2x}
+P_{3x}C_{3x})e_{A_h}, \nonumber \\
\bar{J}_{12} =& R_x^\top R_y+\sum_{(k,i,j) \in \mathcal{A}} \left(P_{1x}C_{1y}+P_{2x}C_{2y}+P_{3x}C_{3y}\right)e_{A_h}, \nonumber \\
\bar{J}_{21} =& R_y^\top R_x+\sum_{(k,i,j) \in \mathcal{A}} \left(P_{1y}C_{1x}+P_{2y}C_{2x}+P_{3y}C_{3x}\right)e_{A_h}, \nonumber \\
\bar{J}_{22} =& R_y^\top R_y+E(p)+\sum_{(k,i,j) \in \mathcal{A}} (P_{1y}C_{1y}+P_{2y}C_{2y}
+P_{3y}C_{3y})e_{A_h}. \nonumber
\end{align}
Note that the stability of an equilibrium point is independent of a rigid-body translation, a rigid-body rotation and a scaling of an entire framework since relative distances and subtended angles only matter. Therefore, without loss of generality, we suppose that $\bar{p}$ lies on the x-axis since they are collinear. Then, we have $R_y=0$, $P_{1y}=0$, $C_{1y}=0$, $P_{2y}=0$, $C_{2y}=0$, $P_{3y}=0$ and $C_{3y}=0$. Then, $\bar{J}(\bar{p})$ is of the form 
\begin{align}
\bar{J}(\bar{p}) = \begin{bmatrix}
\bar{J}_{11}(\bar{p}) & 0 \\
0 & E(\bar{p})
\end{bmatrix}. \label{J_incor}
\end{align}
The following results show that the system (\ref{control_law01}) at $\bar{p}$ is unstable.

\begin{lemma}\label{incorrect_negative_eigen}
Let $\bar{p}$ be in the incorrect equilibrium set $\mathcal{P}_i$. Then, $E(\bar{p})$ has at least one negative eigenvalue.
\end{lemma}
\begin{proof}
We first define $\alpha,\beta$ and $\gamma$ as $\alpha=\frac{\partial}{\partial p_k}\cos\theta^k_{ij}$, $\beta=\frac{\partial}{\partial p_i}\cos\theta^k_{ij}$ and $\gamma=\frac{\partial}{\partial p_j}\cos\theta^k_{ij}$, and let
 $\alpha_{p_k}$, $\alpha_{p_i}$ and $\alpha_{p_j}$ denote coefficients of $p_k$, $p_i$ and $p_j$ in $\alpha$, respectively. Similarly, $\beta_{p_k}$, $\beta_{p_i}$, $\beta_{p_j}$, $\gamma_{p_k}$, $\gamma_{p_i}$ and $\gamma_{p_j}$ are denoted.
Then, from the structure of the matrix $E$, we can have the following equation when $\mathcal{E} \neq \emptyset$ for a configuration $\hat{p}=[\hat{p}_{1}^\top,...,\hat{p}_{n}^\top]^\top \in \mathbb{R}^{2n}$.
\begin{align} \label{dege_unstable_eq01}
&\hat{p}^\top[E(\bar{p})\otimes I_d]\hat{p}  \nonumber \\
&= 2\sum_{(i,j)\in\mathcal{E}}e_{ij}(\bar{p})\norm{\hat{p}_i-\hat{p}_j}^2  \nonumber \\
&\quad+\sum_{(k,i,j) \in \mathcal{A}}e_{A_h}(\bar{p}) \big(\hat{p}_k^\top\hat{p}_k \alpha_{\bar{p}_k} +\hat{p}_k^\top\hat{p}_i \alpha_{\bar{p}_i}+\hat{p}_k^\top\hat{p}_j \alpha_{\bar{p}_j}
+ \hat{p}_i^\top\hat{p}_k \beta_{\bar{p}_k} +\hat{p}_i^\top\hat{p}_i \beta_{\bar{p}_i}+\hat{p}_i^\top\hat{p}_j \beta_{\bar{p}_j}
+ \hat{p}_j^\top\hat{p}_k \gamma_{\bar{p}_k} +\hat{p}_j^\top\hat{p}_i \gamma_{\bar{p}_i}+\hat{p}_j^\top\hat{p}_j \gamma_{\bar{p}_j}\big)\nonumber \\
&=2\sum_{(i,j)\in\mathcal{E}}e_{ij}(\bar{p})\norm{\hat{p}_i-\hat{p}_j}^2 - \sum_{(k,i,j) \in \mathcal{A}} e_{A_h}(\bar{p}) \big( \norm{\hat{p}_k-\hat{p}_i}^2 \beta_{\bar{p}_k} 
+\norm{\hat{p}_k-\hat{p}_j}^2\alpha_{\bar{p}_j}+\norm{\hat{p}_i-\hat{p}_j}^2\gamma_{\bar{p}_i} \big),
\end{align}
where $e_{ij}(\bar{p})=\norm{z(\bar{p})_{ij}}^2-\norm{z^*_{ij}}^2$, $e_{A_h}(\bar{p})=\left.A_{h_{kij}}\right|_{p=\bar{p}}-A_{h_{kij}}^*$,
\begin{align}
\beta_{\bar{p}_k}&=\frac{-1}{\norm{\bar{z}_{ki}}\norm{\bar{z}_{kj}}}+\left(\frac{\norm{\bar{z}_{ki}}^{2} + \norm{\bar{z}_{kj}}^{2} - \norm{\bar{z}_{ij}}^{2}}{2\norm{\bar{z}_{ki}}\norm{\bar{z}_{kj}}}\right)\frac{1}{\norm{\bar{z}_{ki}}^2}, \nonumber\\
\alpha_{\bar{p}_j}&=\frac{-1}{\norm{\bar{z}_{ki}}\norm{\bar{z}_{kj}}}+\left(\frac{\norm{\bar{z}_{ki}}^{2} + \norm{\bar{z}_{kj}}^{2} - \norm{\bar{z}_{ij}}^{2}}{2\norm{\bar{z}_{ki}}\norm{\bar{z}_{kj}}}\right)\frac{1}{\norm{\bar{z}_{kj}}^2}, \nonumber\\
\gamma_{\bar{p}_i}&=\frac{1}{\norm{\bar{z}_{ki}}\norm{\bar{z}_{kj}}}, \nonumber
\end{align}
$\bar{z}_{ij} = \bar{p}_{i} - \bar{p}_{j}$
and it holds that $\alpha_{\bar{p}_i}=\beta_{\bar{p}_k}$,$\alpha_{\bar{p}_j}=\gamma_{\bar{p}_k}$ and $\beta_{\bar{p}_j}=\gamma_{\bar{p}_i}$, and it also holds that $\alpha_{\bar{p}_k}+\alpha_{\bar{p}_i}+\alpha_{\bar{p}_j}=0$, $\beta_{\bar{p}_k}+\beta_{\bar{p}_i}+\beta_{\bar{p}_j}=0$ and $\gamma_{\bar{p}_k}+\gamma_{\bar{p}_i}+\gamma_{\bar{p}_j}=0$. In the case of $\mathcal{E} = \emptyset$, we have 
\begin{align} \label{dege_unstable_eq02}
\hat{p}&^\top[E(\bar{p})\otimes I_d]\hat{p} \nonumber \\
=&- \sum_{(k,i,j) \in \mathcal{A}} e_{A_h}(\bar{p}) \big( \norm{\hat{p}_k-\hat{p}_i}^2 \beta_{\bar{p}_k} +\norm{\hat{p}_k-\hat{p}_j}^2\alpha_{\bar{p}_j} 
+\norm{\hat{p}_i-\hat{p}_j}^2\gamma_{\bar{p}_i} \big).
\end{align}

Suppose that $E(\bar{p})$ is positive semidefinite. Then, $\hat{p}^\top[E(\bar{p})\otimes I_d]\hat{p} \geq 0$ for any configuration $\hat{p} \in \mathbb{R}^{2n}$.
Consider the desired configuration $p^*=[p^{*\top}_{1},...,p^{*\top}_{n}]^\top \in \mathbb{R}^{2n}$ in $\mathcal{P}^*$. With the fact that the equality (\ref{dege_unstable_eq01}) and $\bar{p}^\top[E(\bar{p})\otimes I_d]\bar{p}=0$, the following equation holds.
\begin{align} \label{dege_unstable_eq03}
p&^{*\top}[E(\bar{p})\otimes I_d]p^* \nonumber \\
=& p^{*\top}[E(\bar{p})\otimes I_d]p^* - \bar{p}^\top[E(\bar{p})\otimes I_d]\bar{p}\nonumber \\
=&2\sum_{(i,j)\in\mathcal{E}}e_{ij}(\bar{p})\norm{p^*_i-p^*_j}^2 - 2\sum_{(i,j)\in\mathcal{E}}e_{ij}(\bar{p})\norm{\bar{p}_i-\bar{p}_j}^2 
-\sum_{(k,i,j) \in \mathcal{A}} e_{A_h}(\bar{p}) \left( \norm{z^*_{ki}}^2 \beta_{\bar{p}_k} + \norm{z^*_{kj}}^2\alpha_{\bar{p}_j}+\norm{z^*_{ij}}^2\gamma_{\bar{p}_i} \right) \nonumber \\
&+\sum_{(k,i,j) \in \mathcal{A}} e_{A_h}(\bar{p}) \left( \norm{\bar{z}_{ki}}^2 \beta_{\bar{p}_k} + \norm{\bar{z}_{kj}}^2\alpha_{\bar{p}_j}+\norm{\bar{z}_{ij}}^2\gamma_{\bar{p}_i} \right) \nonumber \\
=&2\sum_{(i,j)\in\mathcal{E}}e_{ij}(\bar{p})\norm{p^*_i-p^*_j}^2 - 2\sum_{(i,j)\in\mathcal{E}}e_{ij}(\bar{p})\norm{\bar{p}_i-\bar{p}_j}^2 
-\sum_{(k,i,j) \in \mathcal{A}} e_{A_h}(\bar{p}) \left( \norm{z^*_{ki}}^2 \beta_{\bar{p}_k} + \norm{z^*_{kj}}^2\alpha_{\bar{p}_j}+\norm{z^*_{ij}}^2\gamma_{\bar{p}_i} \right) \nonumber \\
&+ \sum_{(k,i,j) \in \mathcal{A}} e_{A_h}(\bar{p}) \frac{\norm{\bar{z}_{ki}}^{2} + \norm{\bar{z}_{kj}}^{2} - \norm{\bar{z}_{ij}}^{2}}{2\norm{\bar{z}_{ki}}\norm{\bar{z}_{kj}}}\left(\frac{2\norm{z^*_{ki}}\norm{z^*_{kj}}}{\norm{\bar{z}_{ki}}\norm{\bar{z}_{kj}}}\right)
- \sum_{(k,i,j) \in \mathcal{A}} e_{A_h}(\bar{p}) \frac{\norm{\bar{z}_{ki}}^{2} + \norm{\bar{z}_{kj}}^{2} - \norm{\bar{z}_{ij}}^{2}}{2\norm{\bar{z}_{ki}}\norm{\bar{z}_{kj}}}\left(\frac{2\norm{z^*_{ki}}\norm{z^*_{kj}}}{\norm{\bar{z}_{ki}}\norm{\bar{z}_{kj}}}\right) \nonumber \\ %
=&- 2\sum_{(i,j)\in\mathcal{E}} \card{e_{ij}(\bar{p})}^2 - \sum_{(k,i,j) \in \mathcal{A}}\card{e_{A_h}(\bar{p})}^2\left(\frac{2\norm{z^*_{ki}}\norm{z^*_{kj}}}{\norm{\bar{z}_{ki}}\norm{\bar{z}_{kj}}}\right) \nonumber \\
&+ \sum_{(k,i,j) \in \mathcal{A}} e_{A_h}(\bar{p}) \frac{\norm{\bar{z}_{ki}}^{2} + \norm{\bar{z}_{kj}}^{2} - \norm{\bar{z}_{ij}}^{2}}{2\norm{\bar{z}_{ki}}\norm{\bar{z}_{kj}}}\bigg(\frac{2\norm{z^*_{ki}}\norm{z^*_{kj}}}{\norm{\bar{z}_{ki}}\norm{\bar{z}_{kj}}} -\frac{\norm{z^*_{ki}}^2}{\norm{\bar{z}_{ki}}^2}-\frac{\norm{z^*_{kj}}^2}{\norm{\bar{z}_{kj}}^2} \bigg),
\end{align}
where $z^*_{ij} = p^*_{i} - p^*_{j}$ and it holds that
$\norm{\bar{z}_{ik}}^2 \beta_{\bar{p}_k} + \norm{\bar{z}_{jk}}^2\alpha_{\bar{p}_j}+\norm{\bar{z}_{ij}}^2\gamma_{\bar{p}_i}=0$. From Lemma \ref{incorrect_collinear}, the incorrect equilibrium point $\bar{p}$ lies on the 1-dimensional space. Thus, $\left(\frac{\norm{\bar{z}_{ki}}^{2} + \norm{\bar{z}_{kj}}^{2} - \norm{\bar{z}_{ij}}^{2}}{2\norm{\bar{z}_{ki}}\norm{\bar{z}_{kj}}}\right)^2 =\left.\left(\cos\theta^k_{ij}\right)^2\right|_{p=\bar{p}}= 1$, which implies that
\begin{align}
e_{A_h}(\bar{p}) \left(\frac{\norm{\bar{z}_{ki}}^{2} + \norm{\bar{z}_{kj}}^{2} - \norm{\bar{z}_{ij}}^{2}}{2\norm{\bar{z}_{ki}}\norm{\bar{z}_{kj}}}\right)
=1-\left(\left.\cos\theta^k_{ij}\right)\right|_{p=p^*} \left(\left.\cos\theta^k_{ij}\right)\right|_{p=\bar{p}} \geq 0.
\end{align}
It also holds that  $\left(\frac{2\norm{z^*_{ki}}\norm{z^*_{kj}}}{\norm{\bar{z}_{ki}}\norm{\bar{z}_{kj}}} -\frac{\norm{z^*_{ki}}^2}{\norm{\bar{z}_{ki}}^2}-\frac{\norm{z^*_{kj}}^2}{\norm{\bar{z}_{kj}}^2} \right) =-\left( \frac{\norm{z^*_{ki}}}{\norm{\bar{z}_{ki}}}-\frac{\norm{z^*_{kj}}}{\norm{\bar{z}_{kj}}}\right)^2 \leq 0$.
Therefore, we have $p^{*\top}[E(\bar{p})\otimes I_d]p^* < 0 $ when $\mathcal{E} \neq \emptyset$. Similarly, when $\mathcal{E} = \emptyset$, it also holds that $p^{*\top}[E(\bar{p})\otimes I_d]p^* < 0 $. However,
this conflicts with $\hat{p}^\top[E(\bar{p})\otimes I_d]\hat{p} \geq 0$ for any configuration $\hat{p}$. Hence, we have the statement.
\end{proof}

\begin{theorem}\label{Thm:incorrect_unstable}
The system (\ref{control_law01}) at any incorrect equilibrium point $\bar{p}$ is unstable. 
\end{theorem}
\begin{proof}
Since $\bar{J}(\bar{p})$ is of the form \eqref{J_incor}, if $E(\bar{p})$ has at least one negative eigenvalue then $\bar{J}(\bar{p})$ also has at least one negative eigenvalue.
From Lemma \ref{incorrect_negative_eigen}, we know that $E(\bar{p})$ has at least one negative eigenvalue and the matrix $\bar{J}(\bar{p})$ also does. Since eigenvalues of $\bar{J}(\bar{p})$ and $J(\bar{p})$ are the same, $J(\bar{p})$ also has at least one negative eigenvalue. Hence, the system (\ref{control_law01}) at any incorrect equilibrium point $\bar{p}$ is unstable. 
\end{proof}
\subsection{Almost global stability on $3$-agent formation in $\mathbb{R}^2$}
This subsection shows that if a configuration $p$ does not belong to $\mathcal{P}_i$ then $p$ does not approach $\mathcal{P}_i$ as time goes on. Finally, this subsection provides the main result of the almost global stability on $3$-agent formations in $\mathbb{R}^2$.

\begin{lemma}\label{Lem:incorrect_not_approching}
Let $p(0)$ denote an initial formation. If $p(0)$ given by the gradient flow law (\ref{control_law01}) does not belong to the set of incorrect equilibria, $\mathcal{P}_i$, then $p(t)$ does not approach $\mathcal{P}_i$ for any time $t \geq 0$.
\end{lemma}
\begin{proof}
For a 3-agent formation in $\mathbb{R}^2$, an incorrect equilibrium point $\bar{p}$ always lies on a hyperplane, i.e., $\rank(C_{\bar{p}}(t))<d$ from Lemma \ref{incorrect_collinear}. Additionally, the linearized version of the system \eqref{control_law01}, i.e., negative Jacobian $J(p)$, at an incorrect equilibrium point $\bar{p}$ has at least one negative eigenvalue from Theorem \ref{Thm:incorrect_unstable}. 
Hence, this property is proved straightforward by a similar approach to the proof of Theorem 2 in Sun et al. (2015)\cite{sun2015rigid}.
\end{proof}

\begin{theorem}\label{Thm:stability_system}
If a framework $(\mathcal{G},\mathcal{A},p(0))$ with $n=3$ is minimally GIWR and $p(0)$ is not in the incorrect equilibrium set $\mathcal{P}_i$ in $\mathbb{R}^{2}$, then $p(0)$ exponentially converges to a point in the desired equilibrium set $\mathcal{P}^*$.
\end{theorem}
\begin{proof}
We define a Lyapunov function candidate as
$
V(e)=\frac{1}{2}e^\top e$. Notice that $V(e) \geq 0$  with $V(e)=0$ if and only if $e = 0$ and $V$ is radially unbounded.
The time derivative of $V(e)$ along a trajectory of ${e}$ is calculated as 
\begin{align}
\dot{V} = e^\top \dot{e} 
=-e^\top R_W R_W^\top e
= - \norm{R_W^\top e}^2. \label{eq:dot_potential_fn}
\end{align} 
We know that $\dot{V} \leq 0$, $\dot{V}$ is equal to zero if and only if $R_{W}^\top e=0$. From Theorem \ref{Thm:incorrect_unstable}, Lemma \ref{Lem:incorrect_not_approching} and the assumption that $p(0) \notin \mathcal{P}_i$, it follows that $e \to 0$ asymptotically fast and the error system (\ref{error_dym}) has an asymptotically stable equilibrium at the origin.

From $p(0) \notin \mathcal{P}_i$, the initial positions do not lie on the 1-dimensional space, i.e., $C_p(0)$ is of full row rank. Then, from Lemma \ref{Lem:grad_law_properties}-(\ref{Lem:rank_preserv_d+1}), $\rank\left(R_W(p(0))\right)=\rank\left(R_W(p(t))\right)$ for all $t\geq0$ in $\mathbb{R}^{d}$.
It follows from $p(0) \notin \mathcal{P}_i$ and Lemma \ref{Lem:grad_law_properties}-(\ref{Lem:rank_preserv_d+1}) that $R_W R_W^\top$ is positive definite for all $t \geq 0$. Henceforth, the equation (\ref{eq:dot_potential_fn}) satisfies
\begin{align}
\dot{V} \leq  -\lambda(R_W R_W^\top) \norm{e}^2, \nonumber
\end{align}
where $\lambda$ denotes the minimum eigenvalue of $R_W R_W^\top$ along this trajectory. 
Moreover, since $V=\frac{1}{2}e^\top e$, the following inequality holds.
\begin{equation}
\dot{V}(e) \leq  -2\lambda V(e),  \label{eq:V_leq_fn02}
\end{equation}
and it follows that $V(e(t))\leq V(e(0))\text{exp}(-2\lambda t)$ by Gronwall-Bellman Inequality \cite[Lemma A.1]{khalil2002nonlinear}.
Therefore, $e \to 0$ exponentially fast and the error system (\ref{error_dym}) has an exponentially stable equilibrium at the origin, which in turn implies that $p \to p^*$ for all initial positions outside the set $\mathcal{P}_i$, where $p^*$ is the desired formation. Hence, we conclude that the formation control system \eqref{control_law01} almost globally exponentially converges to the desired formation in $\mathcal{P}^*$.
\end{proof}

\section{Simulation examples}
\label{Sec:Simul}
We provide four examples to support our main results. We first define a squared distance error and a cosine error as $e_{ij}=\norm{z_{g_{ij}}}^2-\norm{z_{g_{ij}}^*}^2$, $(i,j) \in \mathcal{E}$ and $e^k_{ij}=A_{h_{kij}}-A_{h_{kij}}^*$, respectively.
\begin{figure}[]
\centering
\subfigure[Trajectories of six agents from initial formation to final formation.]{
\qquad\qquad \includegraphics[height = 6.0cm]{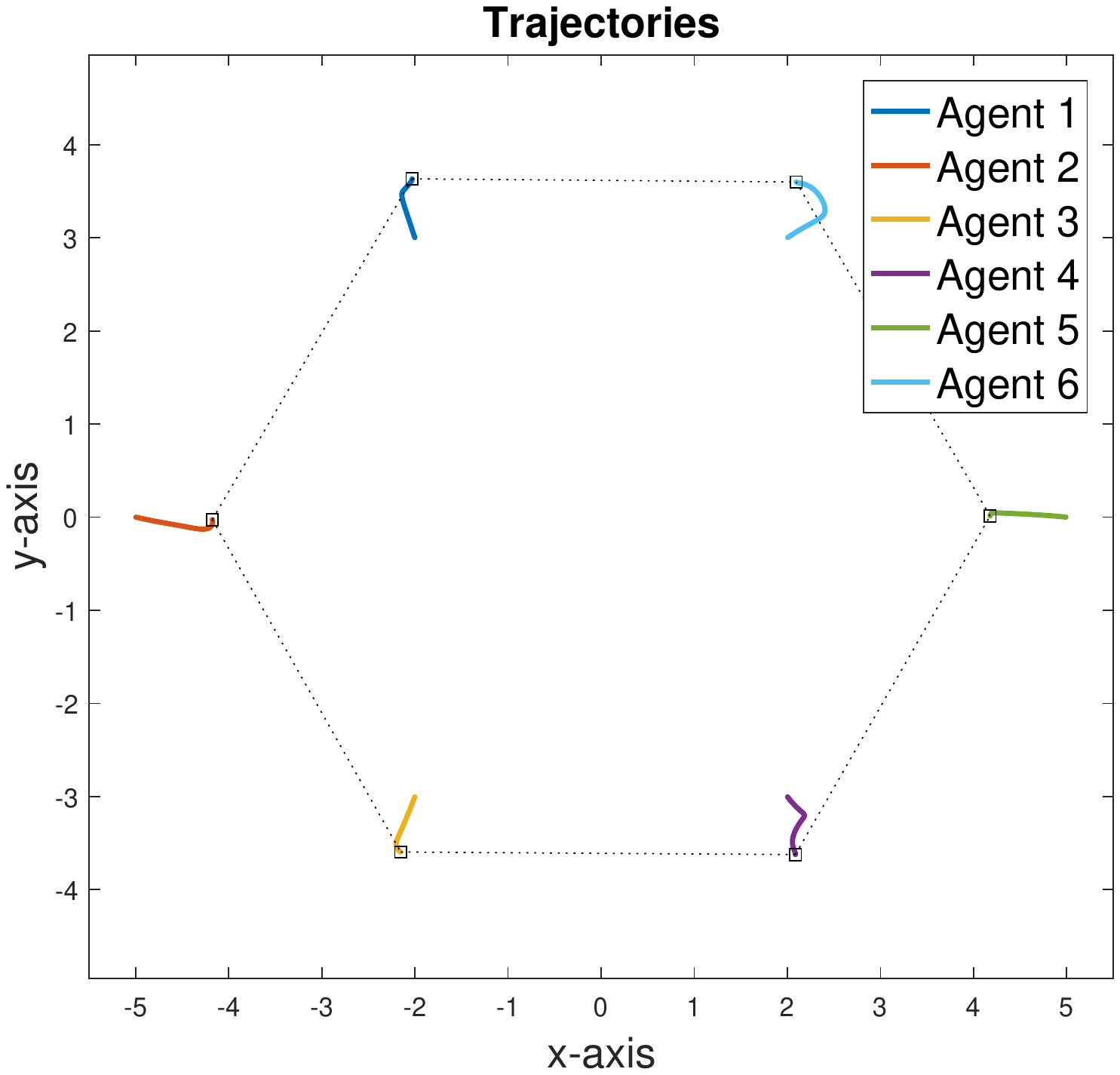} \qquad\qquad
\label{Fig:simul_1}
}\qquad
\subfigure[Exponential convergence of the errors. For high convergence rate of the errors related to the angle constraints, we add proportional gains to the errors.]{
\qquad\includegraphics[height = 6.0cm]{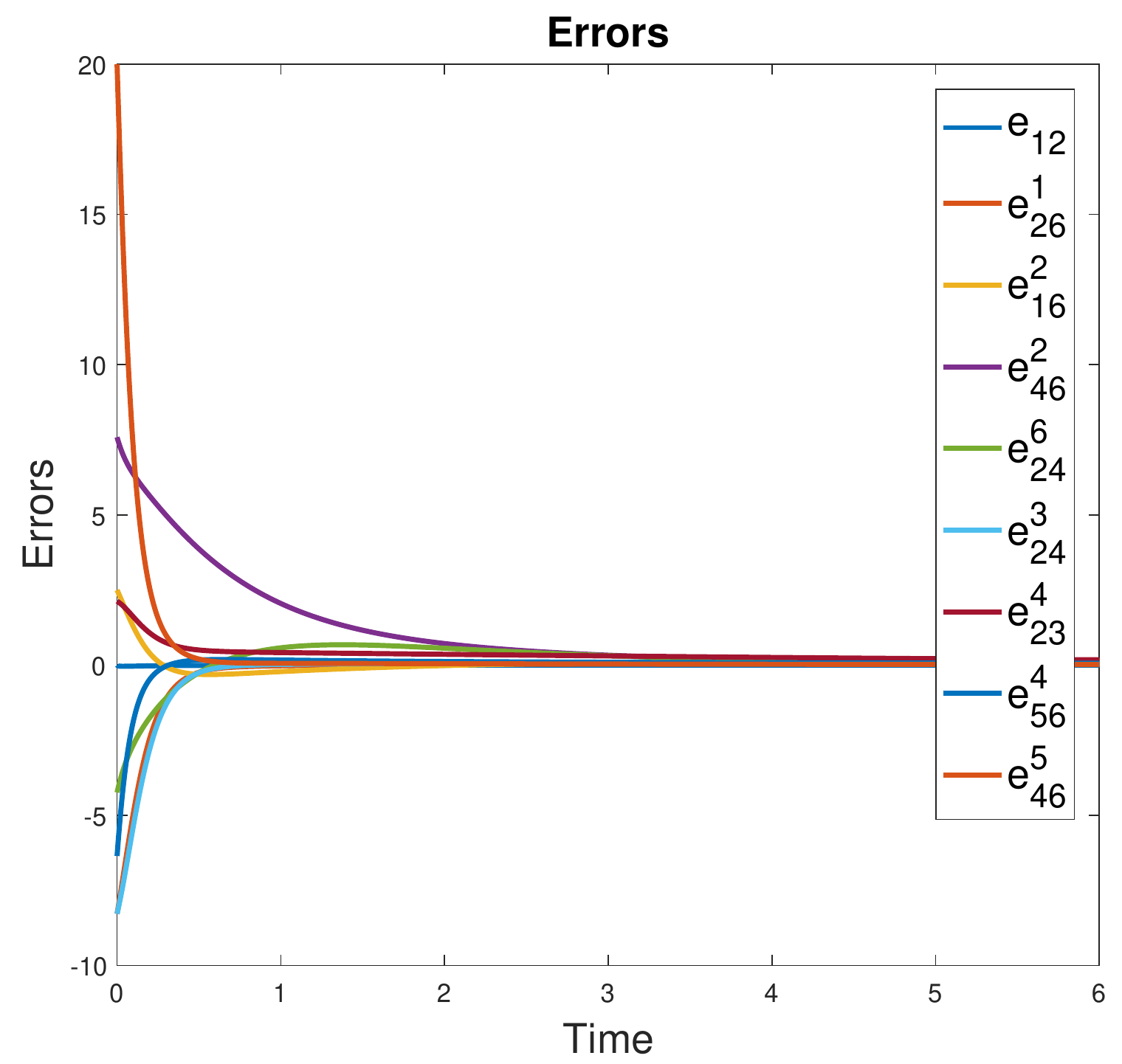}\qquad
\label{Fig:simul_2}
}
\caption{Simulation 1: 6-agent formation control with one distance and eight angle constraints in $\mathbb{R}^2$ under the system \eqref{control_law01}.}
\label{Fig:simul_12}
\end{figure} %
\begin{figure}[]
\centering
\subfigure[Trajectories of five agents from initial formation to final formation.]{
\qquad\qquad \includegraphics[height = 6.8cm]{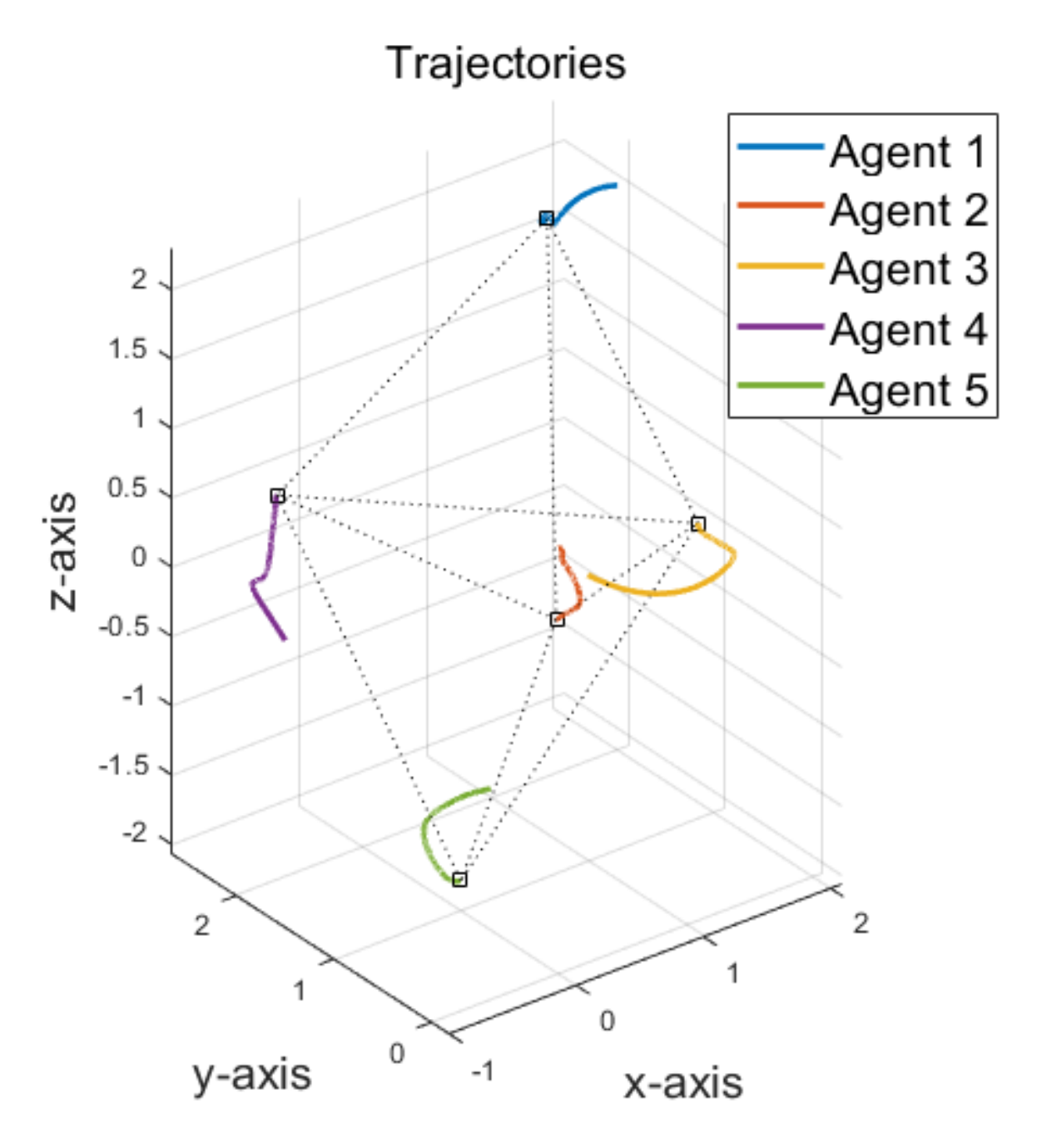} \qquad\qquad
\label{Fig:simul_3d_1}
} \quad
\subfigure[Exponential convergence of the errors. For high convergence rate of the errors related to the angle constraints, we add proportional gains to the errors.]{
\qquad\includegraphics[height = 6.1cm]{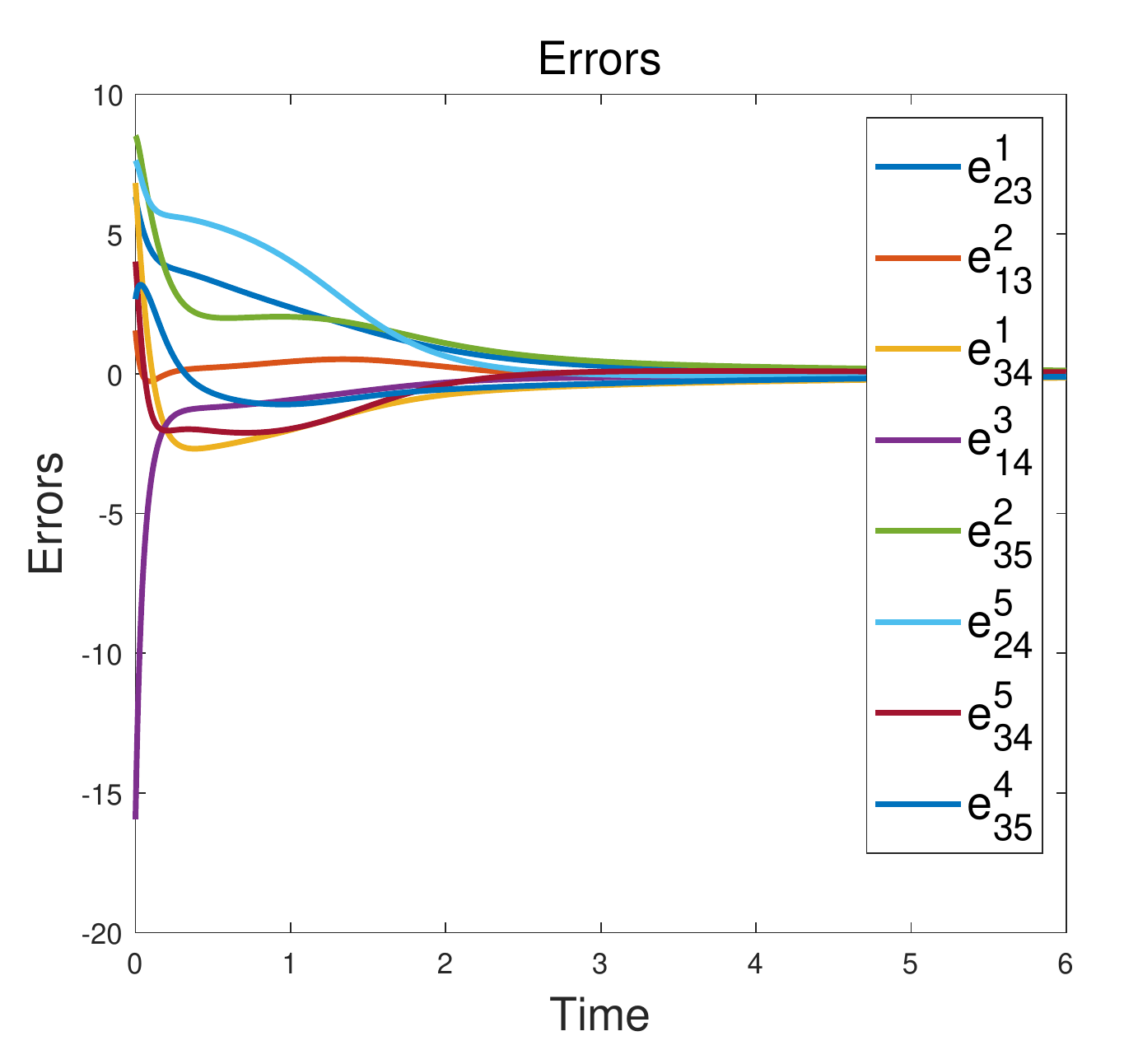}\qquad
\label{Fig:simul_3d_2}
}
\caption{Simulation 1: 5-agent formation control with eight angle constraints in $\mathbb{R}^3$ under the system \eqref{control_law01}.}
\label{Fig:simul_3D}
\end{figure} %

\begin{figure}[]
\centering
\subfigure[Trajectories of three agents from initial formation to final formation.]{\label{Fig:simul_3}
\qquad  \includegraphics[height = 5.8cm]{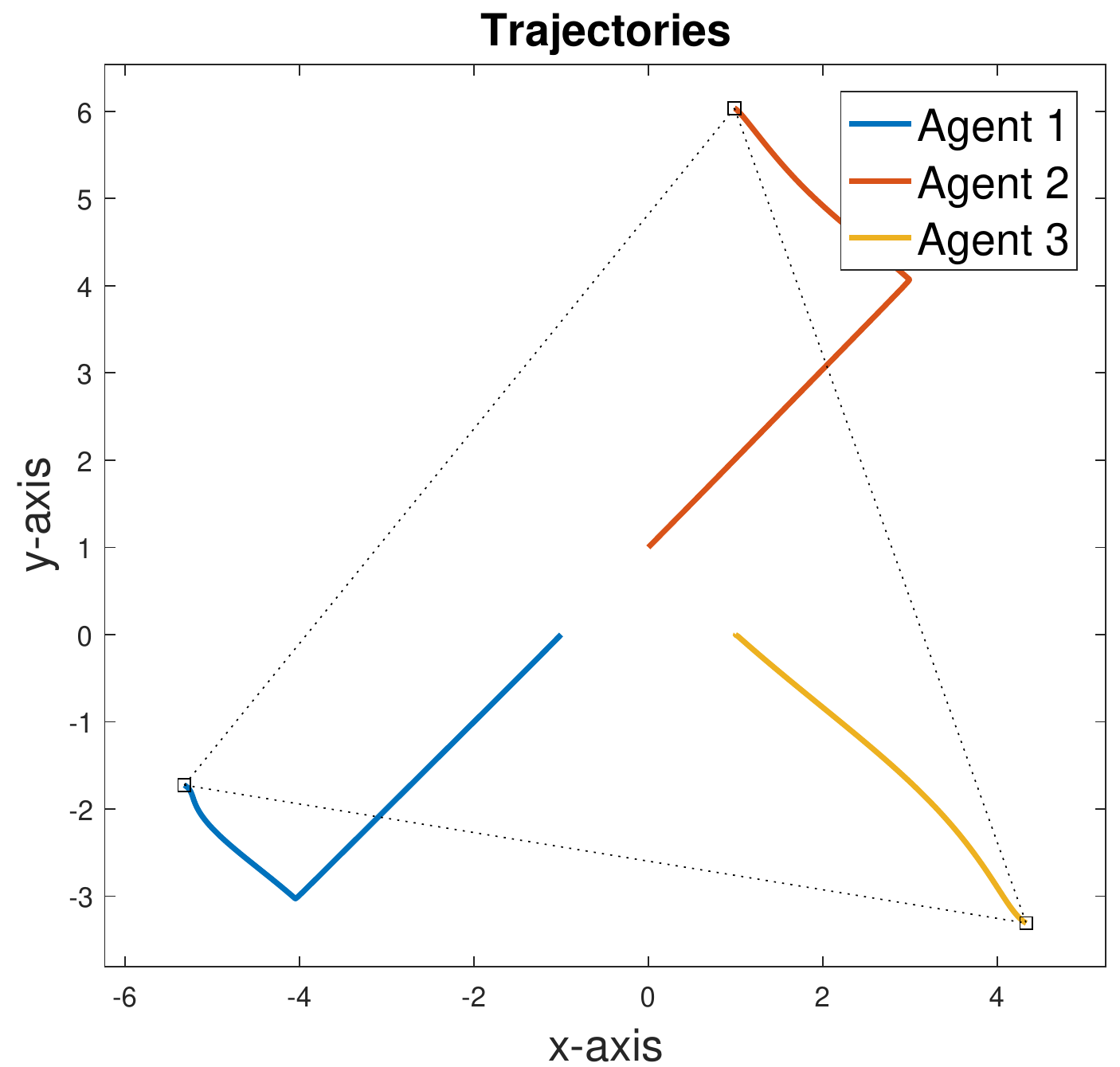}  \qquad 
}\qquad
\subfigure[Exponential convergence of the errors. For high convergence rate of the errors related to the angle constraints, we add proportional gains to the errors.]{\label{Fig:simul_4}
\qquad   \includegraphics[height = 5.8cm]{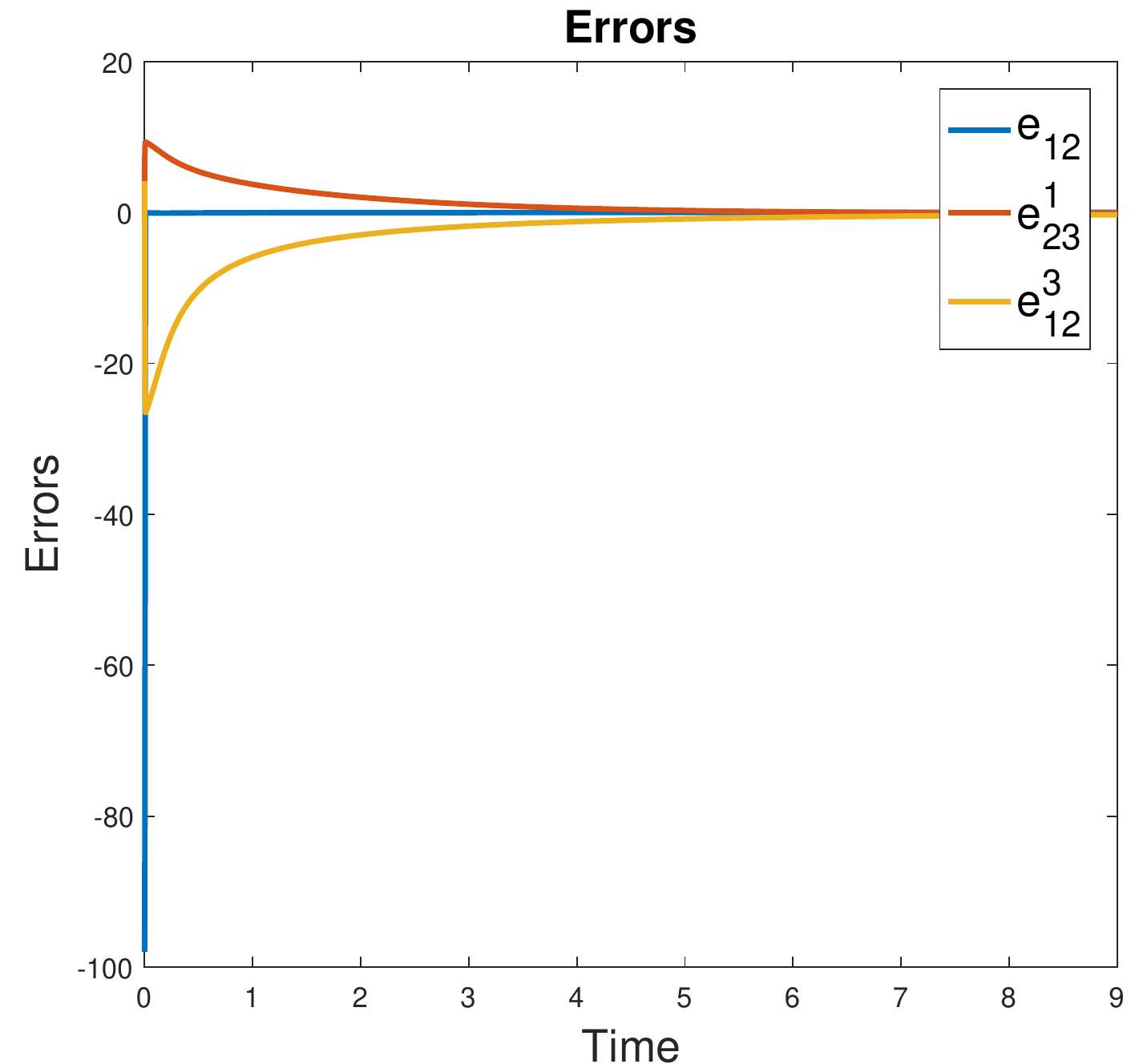}  \qquad  
} 
\caption{Simulation 2: 3-agent formation control with one distance and two angle constraints in $\mathbb{R}^2$ under the system \eqref{control_law01}.}
\label{Fig:simul_34}
\end{figure} %

\begin{figure}[]
\centering
\subfigure[Trajectories of three agents with collinear formation at the beginning.]{\label{Fig:simul_4}
\qquad \, \includegraphics[height = 5.8cm]{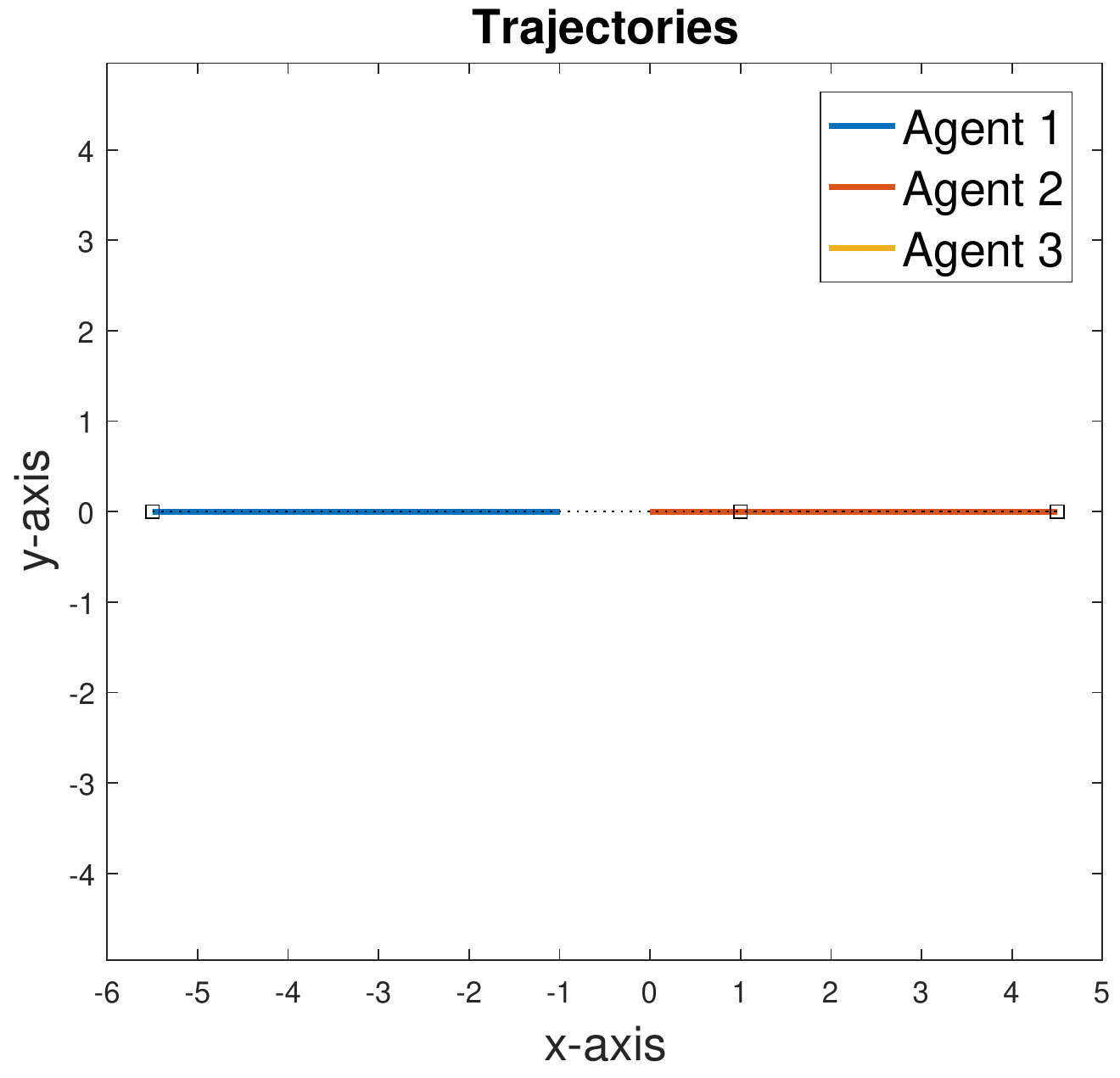} \qquad \,
}\qquad
\subfigure[Convergence of the errors with collinear formation at the beginning.]{\label{Fig:simul_5}
\qquad  \includegraphics[height = 5.8cm]{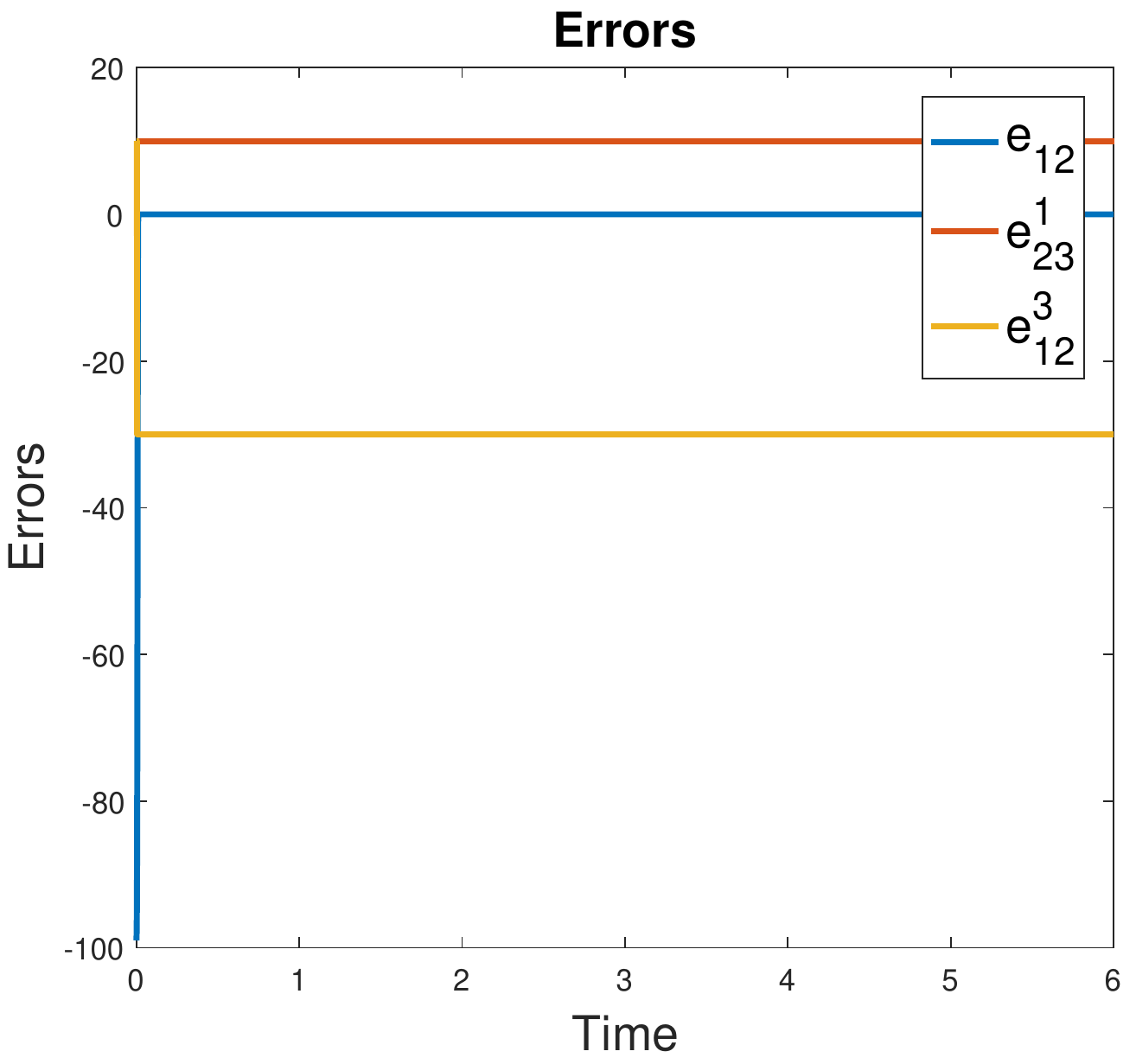} \qquad 
}
\caption{Simulation 3: 3-agent formation control with one distance and two angle constraints under initial collinear-formation in $\mathbb{R}^2$ and the system \eqref{control_law01}.}
\label{Fig:simul_56}
\end{figure} %

For the first simulation, consider a $6$-agent formation control system in $\mathbb{R}^2$ to show that the desired formation shape is locally achieved by the control law as discussed in Section \ref{Sec:Formation_control_local_stab.}. We choose 9 constraints which constitute 1 edge constraint and 8 angle constraints. By using the constraints, the desired formation is given as a minimally GIWR formation, and desired target values are chosen as $\norm{z_{g_{12}}^*}^2=20$, $A_{h_{126}}^*=A_{h_{324}}^*=A_{h_{546}}^*=\cos120^\circ$, $A_{h_{246}}^*=A_{h_{624}}^*=\cos60^\circ$ and $A_{h_{216}}^*=A_{h_{423}}^*=A_{h_{456}}^*=\cos30^\circ$. The local exponential convergence of the $6$-agent formation control system is shown in Fig.~\ref{Fig:simul_12}. In the simulation, the initial formation for each agent are given so that it is GIWR close to the desired formation. 
For the second simulation, Fig.~\ref{Fig:simul_3D} shows a $5$-agent formation control in $\mathbb{R}^3$, where the desired formation is minimally GIWR. Moreover, we choose 8 angle constraints such that $A_{h_{123}}^*=A_{h_{213}}^*=A_{h_{134}}^*=A_{h_{314}}^*=A_{h_{235}}^*=A_{h_{524}}^*=A_{h_{534}}^*=A_{h_{435}}^*=\cos60^\circ$. The local exponential convergence is shown in Fig.~\ref{Fig:simul_3d_2}.
For the third simulation, consider another formation control system such that the desired formation shape is almost globally achieved by the control law as discussed in Section \ref{Sec:Formation_control_almost}. In this simulation, we choose 3 constraints which constitute 1 edge constraint and 2 angle constraints, and set the constraints as $\norm{z_{g_{12}}^*}^2=100$, $A_{h_{123}}^*=A_{h_{312}}^*=\cos60^\circ$. The initial formation is randomly generated except that the initial formation is collinear. Then, the almost globally exponential convergence of the $3$-agent formation control system in $\mathbb{R}^2$ is shown in Fig.~\ref{Fig:simul_34}.
In particular, as the final simulation, if the initial formation is collinear then the formation converges to a point in incorrect equilibria as shown in Fig.~\ref{Fig:simul_56}.

\section{Conclusion} 
\label{Sec:Conclusion}
This paper studied the GWR theory and stability for formation control systems based on the GWR theory in the $2$- and $3$-dimensional spaces. 
Based on the GWR theory, we can determine a rigid formation shape with a set of pure inter-agent distances and angles.
 In particular, with using the rank condition of the weak rigidity matrix, we can conveniently examine whether a formation shape is locally rigid or not. We also showed that both GWR and GIWR for a framework are generic properties, and the GWR theory is necessary for the distance rigidity theory. 
We then applied the GWR theory to the formation control with the gradient flow law. As the first result of its applications,
we proved the locally exponential stability for GIWR formations in the $2$ and $3$-dimensional spaces. 
Finally, for $3$-agent formations in the $2$-dimensional space, we showed the almost globally exponential stability of the formation control system.

\bibliographystyle{IEEEtran}
\bibliography{rigidity2018}
\section*{Appendix}
\begin{lemma}\label{Lem:poly_sub01}
Let us define a vector $\mbf{v}(p)$ as
	\begin{equation}
	\mbf{v}(p)=
	\begin{bmatrix}
	z_{31}^\top &z_{41}^\top &\cdots &z_{n1}^\top
	\end{bmatrix}^\top \in \mathbb{R}^{d(n-2)}.
	\end{equation}
Then, under the control system \eqref{control_law01}, we can calculate $\norm{z_{21}}$ by using the entries in $\mbf{v}(p)$.
\end{lemma}
\begin{proof}
First, without loss of generality, suppose $z_{21}$ is on the $x$-axis.
Then, let us observe such fact
\begin{flalign} 
p-\mathds{1}_n\otimes p^o
	=\begin{bmatrix}
	\frac{n}{n} p_1 - p^o \\
	\frac{n}{n} p_2 - p^o \\
	\vdots\\
	\frac{n}{n} p_n - p^o
	\end{bmatrix}
=\begin{bmatrix}
\frac{(p_1-p_2)+(p_1-p_3)+\cdots+(p_1-p_n)}{n}\\
\frac{(p_2-p_1)+(p_2-p_3)+\cdots+(p_2-p_n)}{n}\\
\vdots\\
\frac{(p_n-p_1)+(p_n-p_2)+\cdots+(p_n-p_{n-1})}{n}
\end{bmatrix}.& \label{eq:appendix01}
\end{flalign}
Since $z_{21}$ is on the $x$-axis, the $x$-axis value of $z_{21}$ is equal to $\norm{z_{21}}$ and further
the variable in $\norm{p-\mathds{1}_n\otimes p^o}$ is only $\norm{z_{21}}$ with the fact that $z_{ij}=z_{i1}-z_{j1}$ 
for all $i,j \in \mathcal{V}$.
Therefore, since $p^s=\norm{p-\mathds{1}_n\otimes p^o}/\sqrt{n}$ is invariant for all $t\geq0$ from Lemma \ref{Lem:grad_law_properties}-\eqref{Lem:cent_scale_inv}, we can calculate $\norm{z_{21}}$ with the value of $\norm{p(0)-\mathds{1}_n\otimes p^o(0)}$ and entries in $\mbf{v}(p)$.
\end{proof}

%
\noindent \textbf{Example 1.}\,
We first denote some notations by $p=[p_1,\,p_2,\,p_3]^\top$, $e_{12}=\norm{z_{12}}^2-\norm{z_{12}^*}^2$, $e^1_{23}=\cos(\theta^1_{23})-\cos((\theta^1_{23})^*)$ and $e^3_{12}=\cos(\theta^3_{12})-\cos((\theta^3_{12})^*)$. Then, based on the framework in Fig.~\ref{Fig:diff_figs_b}, the weak rigidity function is given by $F_W=[\norm{z_{12}}^2,\,\cos(\theta^1_{23}),\, \cos(\theta^3_{12})]^\top$, and the weak rigidity matrix is given by
$$
R_W=\frac{\partial F_W}{\partial p}=\begin{bmatrix}
2z_{12}^\top & -2z_{12}^\top & 0 \\
\frac{\partial}{\partial p_1} \cos(\theta^1_{23}) & \frac{\partial}{\partial p_2} \cos(\theta^1_{23}) & \frac{\partial}{\partial p_3} \cos(\theta^1_{23})\\
\frac{\partial}{\partial p_1} \cos(\theta^3_{12}) & \frac{\partial}{\partial p_2} \cos(\theta^3_{12}) & \frac{\partial}{\partial p_3} \cos(\theta^3_{12})
\end{bmatrix}.
$$
Let us observe the proposed controller $\dot{p}=-R_W^\top e$, where $e=[e_{12},\,e^1_{23},\,e^3_{12}]^\top$. Then, the controllers for each agent are given by 
\begin{align}
\dot{p}_1=&-2z_{12}e_{12}-e_{23}^1 \left(\frac{\partial}{\partial p_1} \cos(\theta^1_{23})\right)^\top-e_{12}^3 \left(\frac{\partial}{\partial p_1} \cos(\theta^3_{12})\right)^\top
\end{align}
\begin{align}
\dot{p}_2=&-2z_{21}e_{12}-e_{23}^1\left(\frac{\partial}{\partial p_2} \cos(\theta^1_{23})\right)^\top -e_{12}^3\left(\frac{\partial}{\partial p_2} \cos(\theta^3_{12})\right)^\top \label{Eq_sub_agent2}\\
\dot{p}_3=&-e_{23}^1\left(\frac{\partial}{\partial p_3} \cos(\theta^1_{23})\right)^\top -e_{12}^3\left(\frac{\partial}{\partial p_3} \cos(\theta^3_{12})\right)^\top.
\end{align}

Consider the controller for agent $2$, where it holds that
\begin{align}
\frac{\partial}{\partial p_2} \cos(\theta^1_{23})
=\frac{\partial}{\partial p_2} \frac{z^\top_{21}}{\norm{z_{21}}} \frac{z_{31}}{\norm{z_{31}}}
=\frac{z^\top_{31}}{\norm{z_{31}}} \frac{1}{\norm{z_{21}}}\left( I_d - \frac{z_{21}z^\top_{21}}{\norm{z_{21}}^2}\right), \label{Eq_sub_01}\\
\frac{\partial}{\partial p_2} \cos(\theta^3_{12})
=\frac{\partial}{\partial p_2} \frac{z^\top_{13}}{\norm{z_{13}}} \frac{z_{23}}{\norm{z_{23}}}
=\frac{z^\top_{13}}{\norm{z_{13}}} \frac{1}{\norm{z_{23}}}\left( I_d - \frac{z_{23}z^\top_{23}}{\norm{z_{23}}^2}\right). \label{Eq_sub_02}
\end{align}Then, it is obvious from \eqref{Eq_sub_01} and \eqref{Eq_sub_02} that we need only two relative positions w.r.t. neighbor agents, i.e., $z_{21}$ and $z_{23}$, for \eqref{Eq_sub_agent2}. Other controllers for agents $1$ and $3$ give similar results. Thus, with the aid of this example, it is easy to see that the general controller for each agent is given by
\begin{align}
\dot{p}_k=&
-\underbrace{2\sum_{j\in\mathcal{N}^d_k}\left(\norm{z_{kj}}^2-\norm{z_{kj}^*}^2\right)(p_k-p_j)}_{(j,k) \in \mathcal{E}} 
- \underbrace{\sum_{i,j\in\mathcal{N}^a_k}\left( \cos{\theta_{ij}^{k}}- \cos{\left(\theta_{ij}^{k}\right)^*}\right) \left(\frac{\partial}{\partial p_k}\cos{\theta_{ij}^{k}}\right)^\top}_{(k,i,j) \in \mathcal{A}} \nonumber \\
&- \underbrace{\sum_{j,k\in\mathcal{N}^a_i}\left( \cos{\theta_{jk}^{i}}- \cos{\left(\theta_{jk}^{i}\right)^*}\right) \left(\frac{\partial}{\partial p_k}\cos{\theta_{jk}^{i}}\right)^\top}_{\text{if }\exists (i,j,k) \in \mathcal{A}}. 
\end{align}
\end{document}